\documentclass[acmsmall,screen,nonacm]{acmart}
\pdfoutput=1

\usepackage[T1]{fontenc}
\usepackage[utf8]{inputenc}
\usepackage{proof}
\usepackage{mathtools}
\usepackage{xfrac}
\usepackage{physics}
\usepackage{tensor}
\usepackage{stmaryrd}
\usepackage{thm-restate}
\usepackage{stmaryrd}
\usepackage{booktabs}
\usepackage{tikz}
\usetikzlibrary{arrows.meta,calc,positioning,decorations.pathmorphing,fit,shapes.misc,patterns,patterns.meta}

\usepackage{apxproof}

\newcommand{\nil}{\mathbf{0}}
\newcommand{\sop}[3][{}]{\mathcal{#2}_{#1}\!\left(#3\right)\!}
\newcommand{\unitary}[3][{}]{\text{#2}_{#1}(#3)}
\newcommand{\meas}[3][M]{#1(#2 \rhd{} #3)}
\newcommand{\varset}{\text{Var}}
\newcommand{\opset}{\text{Op}}
\newcommand{\measset}{\text{Meas}}
\newcommand{\chanset}{\text{Chan}}
\newcommand{\chtype}[1]{\widehat{#1}}
\newcommand{\ntype}{\mathbb{N}}
\newcommand{\btype}{\mathbb{B}}
\newcommand{\qtype}{\mathcal{Q}}
\newcommand{\confset}{\mathit{Conf}}

\newcommand{\hilbert}{\mathcal{H}}
\newcommand{\qubit}{\widehat{\mathcal{H}}}
\newcommand{\qubits}[1]{\widehat{\mathcal{H}}^{\otimes{}#1}}
\newcommand{\ite}[3]{\textbf{if } #1 \textbf{ then } #2 \textbf{ else } #3}
\newcommand{\proc}[1]{\text{\textbf{#1}}}
\newcommand{\iconf}{\mathcal{C}}
\DeclarePairedDelimiter{\ptag}{{\llparenthesis\,}}{{\,\rrparenthesis}}
\newcommand{\barr}{\iota}

\newcommand{\confl}{\mathmbox{\big\langle\hspace{-4.4pt}\big\langle}}
\newcommand{\confr}{\mathmbox{\big\rangle\hspace{-4.4pt}\big\rangle}}
\newcommand{\conf}[1]{\confl{#1}\confr}

\newcommand{\singleton}{\overline}
\newcommand{\distelem}[2]{#1 \bullet #2}
\newcommand{\psum}[1]{\tensor[_{#1}]{\oplus}{}}
\newcommand{\rel}{\mathcal{R}}
\newcommand{\dist}[1]{\mathcal{D}(#1)}
\newcommand{\density}[1]{\mathcal{DO}(#1)}
\newcommand{\pdensity}[1]{\mathcal{D}^{\leq 1}(#1)}
\newcommand{\ptrace}[2]{\tr_{\tilde{#1}}\left(#2\right)}
\newcommand{\soset}[1]{\mathcal{SO}(#1)}
\newcommand{\tsoset}[1]{\mathcal{TS}(#1)}
\newcommand{\cfield}{\mathbb{C}}
\newcommand{\support}[1]{\lceil#1\rceil}

\DeclareMathOperator{\longsquiggly}{--\hspace{-4.15pt}\rightsquigarrow}

\newcommand{\downsquiggly}[1][{}]{\rotatebox[origin=c]{270}{$\longsquiggly$}_{#1}}
\newcommand{\expar}{\parallel}

\newcommand{\kp}{\ket{\psi}}
\newcommand{\kf}{\ket{\phi}}
\newcommand{\kz}{\ket{0}}
\newcommand{\ko}{\ket{1}}
\newcommand{\kpl}{\ket{+}}
\newcommand{\km}{\ket{-}}
\newcommand{\dket}[1]{\ketbra{#1}}

\newcommand{\confbot}{\confset_{\!\!\bot}}
\newcommand{\sconf}[1]{\singleton{\conf{#1}}}
\newcommand{\true}{\mathit{tt}}
\newcommand{\false}{\mathit{ff}}

\newcommand{\rulename}[1]{\mbox{\scriptsize\scshape #1}}
\newcommand{\blank}{{\,\cdot\,}}

\newcommand{\com}[1]{{\color{red} #1}}

\title[Quantum Bisimilarity via Barbs and Contexts]{Quantum Bisimilarity via Barbs and Contexts:\\Curbing the Power of Non-Deterministic Observers}

\author{Lorenzo Ceragioli}
\orcid{0000-0002-1288-9623}
\affiliation{%
  \institution{IMT School for Advanced Studies Lucca}
  \city{Lucca}
  \country{Italy}}
\email{lorenzo.ceragioli@imtlucca.it}
\author{Fabio Gadducci}
\orcid{0000-0003-0690-3051}
\affiliation{%
  \institution{University of Pisa}
  \department{Department of Computer Science}
  \city{Pisa}
  \country{Italy}}
\email{fabio.gadducci@unipi.it}
\author{Giuseppe Lomurno}
\orcid{0009-0000-0573-7974}
\affiliation{%
  \institution{University of Pisa}
  \department{Department of Computer Science}
  \city{Pisa}
  \country{Italy}}
\email{giuseppe.lomurno@phd.unipi.it}
\author{Gabriele Tedeschi}
\orcid{0009-0002-5345-9141}
\affiliation{%
  \institution{University of Pisa}
  \department{Department of Computer Science}
  \city{Pisa}
  \country{Italy}}
\email{gabriele.tedeschi@phd.unipi.it}


\begin{document}

\begin{abstract}
	Past years have seen the development of a few proposals for quantum
	extensions of process calculi. The rationale is clear: with the development
	of quantum communication protocols, there is a need to abstract and focus on
	the basic features of quantum concurrent systems, like CCS and CSP have done
	for their classical counterparts. So far, though, no accepted standard has
	emerged, neither for the syntax nor for the behavioural semantics.
  Indeed, the various proposals do not agree on what should be the observational
  properties of quantum values, and as a matter of fact,
  the soundness of such properties has never been validated against the prescriptions of quantum theory.

  To this aim, we introduce a new calculus, Linear Quantum CCS (lqCCS), and investigate
  the features of behavioural equivalences based on barbs and contexts.
	Our calculus can be thought of as an asynchronous,
	linear version of qCCS, which is in turn based on value-passing CCS\@. 
	The
	combination of linearity and asynchronous communication fits well with the properties of quantum
	systems (e.g.\ the no-cloning theorem),
	since it ensures that each qubit is sent
	exactly once, precisely specifying
	which qubits of a process interact with the
	context.

  We exploit contexts to examine how bisimilarities relate to quantum theory.
  We show that the observational power of general contexts is incompatible with quantum theory:
  roughly, they can perform non-deterministic moves depending on quantum values
  without measuring (hence perturbing) them.

	Therefore, we refine the operational semantics in order to prevent contexts from performing unfeasible
	non-deterministic choices. This induces a coarser bisimilarity that better
	fits the quantum setting: ($i$) it lifts the indistinguishability of quantum states to the distributions of processes and,
  despite the additional constraints, ($ii$) it preserves the expressivity of non-deterministic choices based on classical information.
  To the best of our knowledge, our semantics is the first one that satisfies the two properties above.
  
  \textbf{
  This is the extended version of the POPL2024 paper \emph{Quantum Bisimilarity via Barbs and Contexts: Curbing the Power of Non-Deterministic Observers}.
  }
\end{abstract}

\maketitle

\section{Introduction}

Quantum computing is a promising emerging technology that
exploits non-classical phenomena described by quantum mechanics, such as
entanglement and superposition. The basic component of quantum algorithms and
protocols is the \emph{qubit}, a system that can be in one of two basis states $\kz$ and $\ko$, as well as in 
any linear combination of them, called a \emph{superposition}. 
The state of a quantum system is modelled as a set of qubits on which the programmer
applies various transformations. Differently from
classical systems, the state of a composite quantum system can be
\emph{entangled}, i.e.\ the subsystems cannot be described separately.
Moreover, reading the state of a qubit (``measuring'' it, in quantum
jargon), causes its state to probabilistically change to one of the
basis states. 
Finally, the \emph{no-cloning theorem} forbids copying qubits and thus poses serious constraints to programmers.

Both theory and implementations of quantum computing attracted considerable research
efforts in the last decades, leading to quantum algorithms with
more than polynomial speedup over classical
counterparts~\cite{shor_algorithms_1994,harrow_quantum_2009} and quantum protocols, e.g.\ for key
distribution~\cite{bb84,poppe_practical_2004} and leader
election~\cite{tani_exact_2012}. Practical applications, though, require quantum computers with large enough
memories. This is a challenging task, as it is difficult to maintain quantum properties among a big number of qubits.
A solution seems to lie on distributed computing,
by suitably linking multiple quantum computers~\cite{kimble_quantum_2008}.

With the recent advances, the need has emerged for verification techniques
applicable to quantum distributed algorithms and protocols. 
Concerning purely probabilistic systems, several models have been proposed so far:
some only target the probabilistic behaviour, like Markov Chains~\cite{SokolovaProbabilistic}, while others also take into account pure non-determinism, like Segala Automata~\cite{Segala95} and Probabilistic Transition Systems~\cite{hennessy_exploring_2012}.
This second approach appears to be more adequate to model protocols where actors can perform their choices freely, 
and not only according to some predefined probability distribution.
Process calculi and bisimilarity have been successful in modelling and verifying classical concurrent systems characterized by probabilistic and non-deterministic behaviours.
We expect the same will hold in the quantum case, where different process calculi and behavioural equivalences have been proposed
that display both quantum and non-deterministic features.
While the features of these calculi are mostly comparable, the proposed bisimilarities greatly vary from one work to the other, and are seldom compared with each other and with quantum theory.
Indeed, they turn out to disagree on several simple cases which naturally occur when modelling real-world protocols. Moreover, such discrepancies are partially due to the fact that some proposed bisimilarities do not fit what is prescribed by quantum theory. 
In fact, \citet{davidson_formal_2012} and \citet{kubota_application_noyear} prove that processes sending indistinguishable quantum values are spuriously discriminated.
This discrepancy is yet to be investigated in depth,
and there are no correctness results relating bisimilarity to indistinguishability in quantum theory.
We exemplify in~\autoref{tab:examples} some cases on which the current proposals diverge, and investigate the underlying causes.
The works we compare are QPAlg by~\citet{lalire_relations_2006}, CQP by~\citet{davidson_formal_2012} and qCCS, which comes with two different bisimilarities: $\sim_{p}$ proposed by~\citet{deng_open_2012}, and $\sim_{d}$ by~\citet{feng_toward_2015-1}.
For each row of the table, we discuss the prescriptions of quantum theory, reporting violations and proposing a possible solution.

\begin{table}
	\caption{Recap of the main differences between the proposed bisimilarities.}
	\label{tab:examples}
	{\small
	\centering

	\renewcommand{\arraystretch}{1.1} %

    \begin{tabular}{l@{\hspace*{0.5em}}c@{\hspace*{0.8em}}c@{\hspace*{0.8em}}c@{\hspace*{0.8em}}c}
			\toprule
			\textbf{PAIR OF PROCESSES (in lqCCS syntax)} & QPAlg & CQP & qCCS & lqCCS \\
			\midrule
      $c?x.\unitary{H}{x}.\nil\ $ \textbf{and} $\ c?x.\unitary{X}{x}.\nil$ & $\sim$ & $\sim$ & $\not\sim_{p},\, \not\sim_{d}$ & illegal                          \\
      $c?x.\unitary{H}{x}.\nil_x$ \textbf{and} $c?x.\unitary{X}{x}.\nil_x$ & $\sim$ & $\sim$ & $\sim_{p},\, \sim_{d}$
			 & $\sim_{s},\, \sim_{cs}$ \\
      $c?x.\unitary{H}{x}.d!x$ \textbf{and} $c?x.\unitary{X}{x}.d!x$ & $\not\sim$ & $\not\sim$ & $\not\sim_{p},\, \not\sim_{d}$ & $\not\sim_{s},\, \not\sim_{cs}$  \\
      $\unitary[\frac{1}{4}I]{Set}{q_1, q_2}. (c!q_1 \parallel c!q_2)$ \textbf{and} $\unitary[\scriptscriptstyle\dket{\Phi^+}]{Set}{q_1, q_2}. (c!q_1 \parallel c!q_2)$ & $\sim$ & $\not\sim$ & $\not\sim_{p},\, \not\sim_{d}$ & $\not\sim_{s},\, \not\sim_{cs}$  \\
      $\unitary[\scriptscriptstyle\dket{+}]{Set}{q}. \meas[M_{01}]{q}{x}. c!q$ \textbf{and}
      $\unitary[\scriptscriptstyle\dket{0}]{Set}{q}. \meas[M_{\pm}]{q}{x}. c!q$ & $\not\sim$    & $\sim$     &
			$\not\sim_{p},\, \sim_{d}$                                                                               & $\not \sim_{s},\, \sim_{cs}$                                                                                                              \\
        $\unitary[\scriptscriptstyle\dket{+}]{Set}{q}. \meas[M_{01}]{q}{x}. (c!q + d!q)$ \textbf{and}
        $\unitary[\scriptscriptstyle\dket{0}]{Set}{q}. \meas[M_{\pm}]{q}{x}. (c!q + d!q)$
			  & $\not\sim$ & $\sim$ & $\not\sim_{p},\, \sim_{d}$ & $\not \sim_{s},\, \not\sim_{cs}$ \\
			\bottomrule
		\end{tabular}}
\end{table}

In fact, to tackle these foundational issues,
we introduce a new process calculus, namely \emph{linear quantum CCS} (lqCCS), 
which offers the main features common to the previous proposals,
and use it as a framework for exploring behavioural equivalence for quantum systems.
Our calculus builds up on qCCS
(in turn inspired by value-passing CCS~\cite{vpCCS}), yet
offers asynchronous communication and a linear type system.
Quantum systems are modelled as \emph{configurations} in a stateful manner,
with the state of the qubits alongside the processes.
Furthermore, the
no-cloning theorem prescribes that once a qubit has been sent, the sender
cannot use it any more.
To comply with this requirement, quantum process calculi
usually enforce \emph{affinity}, guaranteeing that each qubit is sent at most
once.
We go one step further by requiring each qubit to be sent or discarded
(i.e.\ sent on a restricted channel) \emph{exactly} once,
thus forcing the observability of each qubit to be
clearly defined.
This allows us to resolve a superficial discrepancy of the proposed behavioural equivalences, focusing on unambiguous cases only.
Take as example the pair of processes of the first row of~\autoref{tab:examples}.
	Both receive a qubit on channel $c$ (with $c?x$),
	modify it (with either $H(x)$ or $X(x)$), and then terminate ($\nil$). 
	The two
	processes apply different transformations, resulting in different quantum
	states.
		Both CQP and QPAlg assume unsent qubits are not visible and deem $P$ and $Q$ bisimilar,
	 while the bisimilarities for qCCS do the opposite.
	The linear typing of lqCCS forces to specify visibility, and allows matching these different results employing an appropriate discarding discipline 
		(i.e.\ the processes are equivalent if the qubit is discarded using $\nil_x$ as in the second row of~\autoref{tab:examples}, they are distinguishable if it is sent on a visible channel as in the third row).

We define a saturated probabilistic bisimilarity for lqCCS, denoted as $\sim_{s}$,
which relies on contexts for
distinguishing quantum processes~\cite{bonchi_general_2014}.
This seems a fruitful choice, because the notion of observable property of a visible quantum
value is not straightforward, as witnessed by the variety of labelled
bisimulations proposed so far.
On the contrary, available operations over quantum values are well known (unitaries and measurements), and thus contexts are easily defined and uniquely determined (at least as far as quantum properties are concerned).
Take the fourth row of~\autoref{tab:examples} as an example, where both processes prepare and send a pair of qubits, but only the latter is entangled.
Remarkably, in both cases, the sent qubits are in the same state when taken separately.
QPAlg uses the value of each sent qubit as labels, 
thus incorrectly equating the two processes.
Indeed, detecting the entanglement requires considering the pair of qubits as a whole.

We also suffer from the problem highlighted by \citeauthor{davidson_formal_2012} and \citeauthor{kubota_application_noyear}, as our saturated bisimilarity spuriously discriminates between equivalent quantum values.
More in detail,
quantum theory prescribes that certain
probability distributions of quantum states should be indistinguishable
(namely if they are represented by the same density operator).
Consider the processes of the fifth row of~\autoref{tab:examples}. 
After setting and measuring the qubits, the state of the sent qubit of both processes is in a fair distribution, respectively of $\kz$ and $\ko$, and of $\kpl$ and $\km$.
Such distributions are prescribed to be indistinguishable.
Nonetheless, the two processes are not bisimilar according to QPAlg and $\sim_p$ of qCCS. Also $\sim_{s}$ suffers from the same issue, but the use of
contexts allow us to precisely pinpoint
the cause of the problem in the interaction between non-determinism and
quantum features.
Indeed, unconstrained non-determinism
allows contexts to choose a move based on the (in principle unknown)
current quantum state, without performing a measurement (and thus, without perturbing the state, violating a defining feature of quantum objects).
Non-determinism must be constrained for contexts to fit the limitations of quantum theory.

We give a new enhanced semantics and proper contexts, forbidding ill-formed moves
where the non-deterministic choices depend on unknown quantum values.
Moreover, we define \emph{constrained} saturated bisimilarity, denoted as $\sim_{cs}$, which
is strictly coarser than $\sim_s$.
We prove two main properties: $\sim_{cs}$ recovers the indistinguishability of quantum values prescribed by quantum theory, and, even if constrained, 
non-deterministic sum can still simulate boolean conditional statements.
\autoref{thm:propertyA} shows that our constraints suffice to equate lqCCS configurations with indistinguishable quantum states.
Notice indeed that the processes of the fifth row of~\autoref{tab:examples} are correctly deemed bisimilar by $\sim_{cs}$.
We argue that our constraints are not overly restrictive.
To confirm this, \autoref{thm:nondetVSite} shows
that non-deterministic choices can always perform different moves according to known classical values, thus replicating the behaviour of boolean guards.
This is an expected property that was missing in previous bisimilarities like the one of CQP and $\sim_d$ of qCCS, which also indirectly constrain non-determinism.
Consider the last row of~\autoref{tab:examples} where both processes send a qubit, choosing non-deterministically over two possible channels.
Remarkably, for each channel, the state of the sent qubits is represented by the same density operator, but the two processes could be distinguished if they 
had chosen the channel according to the outcome of the measurement (e.g.\ by using a boolean guard).
Since non-deterministic sum simulates such behaviour in our enhanced semantics, the two processes are correctly distinguished.
On the contrary, all the previous works that correctly equates the processes of the fifth row of the table fails in distinguishing the ones in the last row: 
indeed, they overly constrain non-determinism.

As a final contribution, we provide a few proof techniques
and employ them to analyse 
three real-world protocols: quantum teleportation, super-dense coding and quantum coin-flipping.

\paragraph{Synopsis} In~\autoref{sec:bg} we give some background about
probability distributions and quantum theory. In~\autoref{sec:lqCCS} we present
lqCCS, we discuss probabilistic bisimilarity and its
interaction with non-determinism. In~\autoref{sec:db} we
propose our novel semantics and bisimilarity equivalence.
In~\autoref{sec:real-world} we describe the capabilities of the novel bisimilarity through the lenses of well-known quantum protocols.
Finally, we compare with related works in~\autoref{sec:rw}, and we conclude in~\autoref{sec:conc}.
To favour readability, the full proofs are postponed to the appendices.

\section{Background}\label{sec:bg}
We recall some background on probability distributions 
and introduce quantum computing.
Finally, we present density operators that model probability distributions of quantum systems and their evolution.
We refer to~\citet{nielsen_quantum_2010} for further reading on quantum computing.

\subsection{Probability Distributions}\label{probabilisticBackground}

A \emph{probability distribution} over a set $S$ is a function $\Delta : S \to
	[0,1]$ such that $\sum_{s \in S}\Delta(s) = 1$. 
We call the \emph{support} of a distribution $\Delta$, written $\support{\Delta}$, the set $\{ s \in
	S\;|\;\Delta(s) > 0 \}$. 
We write $\dist{S}$ for the set of distributions over $S$, and restrict ourselves to distributions with finite support.

For each $s \in S$, we let $\singleton{s}$ be the \emph{point distribution} that assigns $1$ to $s$.
Given a finite set of non-negatives reals $\{p_i\}_{i \in I}$ such that $\sum_{i \in I} p_i = 1$, we write $\sum_{i \in I} \distelem{p_i}{\Delta_i}$ for the distribution
determined by $(\sum_{i \in I} \distelem{p_i}{\Delta_i})(s) = \sum_{i \in
		I}p_i\Delta_i(s)$.
Sometimes, we will use the notation above to write the distributions ``explicitly'' by listing the elements of the support with their probability as in $\Delta = \sum_{s \in \support{S}} p_s \bullet \singleton{s}$.
Finally, the notation $\Delta_1 \psum{p} \Delta_2$ is a shorthand for
$\distelem{p}{\Delta_1} + \distelem{(1 - p)}{\Delta_2}$.

A relation $\mathcal{R} \subseteq \dist{S} \times \dist{S}$ is said to be
\emph{linear} if $(\Delta_1 \psum{p} \Delta_2)\;\mathcal{R}\;(\Theta_1 \psum{p}
	\Theta_2)$ for any $p \in [0, 1]$ whenever $\Delta_i\;\mathcal{R}\;\Theta_i$
for $i = 1, 2$. $\rel$ is said to be \emph{left-decomposable} if $(\Delta_1 \psum{p}
	\Delta_2)\;\mathcal{R}\; \Theta$ implies $\Theta = (\Theta_1 \psum{p}
	\Theta_2)$ for some $\Theta_1, \Theta_2$ with $\Delta_i\;\mathcal{R}\;\Theta_i$ for $i = 1,2$ and for any $p \in [0, 1]$. Right-decomposability is defined symmetrically, and a relation is \emph{decomposable} when it is both left- and right-decomposable.

Given $\mathcal{R} \subseteq S \times \dist{S}$, its \emph{lifting}
$\text{lift}(\mathcal{R}) \subseteq \dist{S} \times \dist{S}$ is the smallest
linear relation such that $\overline{s}\;\text{lift}(\mathcal{R})\;\Theta$ when
$s\;\mathcal{R}\;\Theta$.

\subsection{State space}

A (finite-dimensional) \emph{Hilbert space}, denoted as $\hilbert$, is a
complex vector space equipped with a binary operator $\braket\blank: \hilbert
	\times \hilbert \rightarrow \cfield$ called \emph{inner product}, defined as
$\braket{\psi}{\phi} = \sum_i \alpha_i^*\beta_i$, where $\kp =
	(\alpha_1,\ldots,\alpha_i)^T$ and $\kf = (\beta_1,\ldots,\beta_i)^T$. We
indicate column vectors as $\kp$ and their conjugate transpose as
$\bra\psi = \kp^\dagger$.
The state of an isolated physical system is represented as a \emph{unit vector}
$\kp$ (called \emph{state vector}), i.e.\ a vector such that $\braket\psi =
	1$. 
The simplest example of a quantum physical system is a \emph{qubit}, which
is associated with the two-dimensional Hilbert Space $\qubit = \cfield^2$. 
The
vectors $\{\ket0 = (1,0)^T, \ket1 = (0,1)^T\}$ form an orthonormal basis of
$\qubit$, called the \emph{computational basis}. Other important vectors in
$\qubit$ are $\ket+ = \frac{1}{\sqrt{2}}(\ket0 + \ket1)$ and $\ket-
	= \frac{1}{\sqrt{2}}(\ket0 - \ket1)$, which form the \emph{diagonal basis}, or
Hadamard basis.

Intuitively, different bases represent different observable properties of a quantum system. 
Note that $\kpl$ and $\km$ are non-trivial linear combinations of $\kz$ and $\ko$, roughly meaning that
the property associated with the computational basis is undetermined in $\kpl$ and $\km$.
In the quantum jargon, the states in the diagonal basis are \emph{superpositions} with respect to the standard basis.
Symmetrically, $\kz$ and $\ko$ are themselves superpositions with respect to the diagonal basis.

\subsection{Unitary Transformations}

For each linear operator $A$ on a Hilbert space $\hilbert$, there is a
linear operator $A^\dag$, the \emph{adjoint} of $A$, which is given by the
conjugate transpose of $A$ and is the unique operator such that
$\mel{\psi}{A}{\phi} = \braket{A^\dag\psi}{\phi}$. A linear operator $U$ is
said to be \emph{unitary} when $UU^\dag = U^\dag U = I$. 
In quantum physics, the evolution of a
closed system is described by a unitary transformation: the state
$\kp$ at time $t_0$ is related to $\ket{\psi^\prime}$ at time $t_1$ by a unitary operator $U$, which only depends
on $t_0$ and $t_1$, i.e.\ $\ket{\psi^\prime} = U\kp$.

In quantum computing, the programmer manipulates the state of qubits by
applying unitary transformations.
Some of the most common transformations on single qubits are:
$X$ that transforms the qubit $\kz$ into $\ko$ and vice-versa (corresponding to the classical logical not);
$Z$ that given $\kp = \alpha\kz + \beta\ko$ returns $\alpha\kz - \beta\ko$;
and $H$ that maps $\kz$ and $\ko$ into $\kpl$ and $\km$ respectively.
\begin{gather*}
	\small
	X = \begin{bmatrix}
		0 & 1 \\
		1 & 0
	\end{bmatrix}\quad
	Z = \begin{bmatrix}
		1 & 0  \\
		0 & -1
	\end{bmatrix}\quad
	H = \frac{1}{\sqrt{2}}\begin{bmatrix}
		1 & 1  \\
		1 & -1
	\end{bmatrix}
\end{gather*}

\subsection{Measurement}

\emph{Quantum measurements} are needed for describing systems that exchange
information with the environment. 
Performing a measurement on a quantum state returns a classical result and causes the quantum state to change (i.e.\ to \emph{decay}).
Thus, measurements alter the state of the qubits.
Moreover, the result of a measurement is intrinsically probabilistic.

A \emph{measurement operator} is a linear transformation 
that associates each input quantum state with a probability and a resulting quantum state.
A \emph{measurement} is then a set of possible classical
outcomes: each of them is
associated with a \emph{measurement operator} that encodes how the probability of the outcome and the resulting quantum state depends on the current state $\kp$.

Formally, a measurement is a set $\{M_m\}_m$ of measurement operators, where $m$ refers
to the classical outcomes, such that the \emph{completeness} equation $\sum_m
	M_m^\dag M_m = I$ holds. If the state of the system is $\kp$ before the
measurement, then the probability of $m$ occurring is $p_m = \bra\psi M_m^\dag
	M_m\kp$. If $m$ is the outcome, then the state after the measurement will be
$\frac{1}{\sqrt{p_m}}M_m\kp$.

The simplest measurements project a state into one of the basis of $\hilbert$,
e.g. $M_{01} = \{M_0, M_1\}$ with $M_0 = \ketbra{0},M_1 = \ketbra{1}$ in the
computational basis of $\qubit$, and $M_\pm = \{M_+, M_-\}$ with $M_+ =
	\ketbra{+}, M_- = \ketbra{-}$ in the Hadamard basis.
	
As expected, applying the measurement $M_{01}$ on $\kz$ returns the classical outcome $0$ and the state $\kz$ with probability $1$.
When applying the same measurement on $\kpl$, instead, the result may be $0$ and $\kz$, or $1$ and $\ko$ with equal probability.
Note also that measuring $\kz$ with $M_\pm$ leads to either $0$ and $\kpl$, or $1$ and $\km$, also with equal probability.

\subsection{Composite quantum systems}

We represent the state space of a composite physical system as the \emph{tensor
	product} of the state spaces of its components. Let $\hilbert_A$ and
$\hilbert_B$ be $n$ and $m$-dimensional Hilbert spaces: their tensor product
$\hilbert_A \otimes \hilbert_B$ is an $n\cdot m$ Hilbert space. Moreover, if
$\{\ket{\psi_1}, \ldots, \ket{\psi_n}\}$ and $\{\ket{\phi_1}, \ldots,
	\ket{\phi_m}\}$ are bases of respectively $\hilbert_A$ and $\hilbert_B$, then $
	\{\ket{\psi_i}\otimes\ket{\phi_j} \mid i = 1, \ldots, n, j = 1, \ldots, m\}$ is
a basis of $\hilbert_A \otimes \hilbert_B$, where $\kp \otimes \kf$
is the Kronecker product, defined as
\begin{align*}
	\begin{bmatrix} x_{1,1} & \cdots & x_{1,n} \\ \vdots & \ddots & \vdots\\ x_{m,1} & \cdots & x_{m,n} \end{bmatrix} \otimes A =
	\begin{bmatrix} x_{1,1} A & \cdots & x_{1,n} A \\ \vdots & \ddots & \vdots \\ x_{m,1} A & \cdots &  x_{m,n} A \end{bmatrix}
\end{align*}

We often omit the tensor product and write $\kp\kf$ or $\ket{\psi\phi}$.
We write $\qubits{n}$ for the $2^n$-dimensional Hilbert space defined as the tensor product of $n$
copies of $\qubit$ (i.e.\ the possible states of $n$ qubits).

What is said above about unitary transformations and measurements also applies to composite systems.
Given a list of single-qubits unitaries $U_1, U_2 \dots U_n$, their tensor product $U_1 \otimes U_2 \dots \otimes U_n$ is a unitary transformation over $n$ qubits.
Not all unitaries on $n$ qubits can be obtained in this way:
the most used that cannot be obtained via tensor product are
the SWAP operator exchanging the state of two qubits, i.e. $\text{SWAP} \ket{\psi\phi} = \ket{\phi\psi}$, and
the controlled not (CNOT) over two qubits, defined as
\[
  \text{CNOT} \ket{00} = \ket{00},\quad \text{CNOT} \ket{01} = \ket{01},\quad \text{CNOT} \ket{10} = \ket{11}, \quad \text{CNOT} \ket{11} = \ket{10}.
\]

A measurement for a composite system may measure only some of the qubits and leave others unaltered, e.g. $\{M_0 \otimes I, M_1 \otimes I\}$ measures (in the computational basis) the first qubit of a pair.

A quantum state in $\hilbert_A \otimes \hilbert_B$ is  \emph{separable}
when it can be expressed as the Kronecker product of two vectors of
$\hilbert_A$ and $\hilbert_B$. Otherwise, it is \emph{entangled},
like the Bell state $\ket{\Phi^+} = \frac{1}{\sqrt{2}}(\ket{00} +
	\ket{11})$. When two qubits are entangled, the evolution of the one depends on
the transformations applied to the other.
Measuring e.g.\ the first qubit of $\ket{\Phi^+}$ in the computational basis causes the state to decay into either $\ket{00}$ or $\ket{11}$ with equal probability, where also the state of the second qubit is updated.

The Bell states $\{\ket{\Phi^+}, \ket{\Phi^-}, \ket{\Psi^+}, \ket{\Psi^-}\}$ form the Bell basis for $\qubits{2}$, with
\begin{align*}
  \ket{\Phi^-} = \frac{1}{\sqrt{2}}(\ket{00} - \ket{11})\qquad
  \ket{\Psi^+} = \frac{1}{\sqrt{2}}(\ket{01} + \ket{10})\qquad
  \ket{\Psi^-} = \frac{1}{\sqrt{2}}(\ket{01} - \ket{10})
\end{align*}

A defining feature of quantum computing is that qubits cannot be
duplicated.

\begin{proposition}[No-Cloning Theorem]\label{noCloning}
	There is no unitary operation $U$ on $\hilbert_A \otimes \hilbert_B$ such
	that $U(\kp \otimes \kf) = \kp \otimes \kp$ for all states $\kp \in
		\hilbert_A$ and $\kf \in \hilbert_B$.
\end{proposition}

As a result, qubits cannot be stored in multiple locations or broadcast to
multiple receivers.

\subsection{Density operator formalism}

The density operator formalism puts together quantum systems and probability distributions by considering
mixed states, i.e.\ \emph{probabilistic mixture of quantum states}.
A point distribution $\singleton{\kp}$ (called a pure state) is represented by the matrix $\ketbra{\psi}$.
In general, a mixed state $\Delta \in \dist{\qubits{n}}$ for $n$ qubits is represented as the matrix $\rho \in \cfield^{2^n\times 2^n}$, known as its \emph{density operator}, with $\rho = \sum_i \Delta(\psi_i)
	\dket{\psi_i}$.
We write $\density{\hilbert}$ for the set of density operators of $\hilbert$.

For example, the mixed state $\singleton{\kz} \psum{1/3} \singleton{\kpl}$ being $\kz$ with probability $1/3$ and in $\kpl$ with probability $2/3$ is represented as 
\[
\frac{1}{3} \ketbra{0} + \frac{2}{3} \ketbra{+} = \frac{1}{3}
\begin{bmatrix} 2 & 1 \\ 1 & 1 \end{bmatrix}
\]

Note that the encoding of probabilistic mixtures of quantum states as density operators is not injective.
For example, $\frac{1}{2}I$ is called the \emph{maximally mixed state} and
represents both the distribution $\Delta_{C} = \singleton{\ket0} \psum{1/2}
	\singleton{\ket1}$ and $\Delta_{H} = \singleton{\ket+} \psum{1/2}
	\singleton{\ket-}$.
This is a desired feature, as the laws of quantum mechanics deem indistinguishable all the distributions that result in the same density operator.

The evolution of mixed states is given as a \emph{trace-preserving
  superoperator} $\mathcal{E}: \density{\hilbert} \rightarrow \density{\hilbert}$, a function defined by its \emph{Kraus operator sum decomposition} $\{E_i\}_i$ for $i = 1, \ldots,
	\dim(\hilbert)^2$, satisfying that
$\sop{E}{\rho} = \sum_i E_i\rho E_i^\dag$ and $\sum_i E_i^\dag E_i = I_\hilbert$,
where $I_{\hilbert}$ is the identity operator on $\hilbert$. 
Notice that the operators $E_i$ are not unitaries in general, see \citet[Section 8.2.3]{nielsen_quantum_2010}.
We write
$\tsoset{\hilbert}$ for the set of trace-preserving superoperators on $\hilbert$.

The tensor product of density operators $\rho \otimes \sigma$ is defined as their Kronecker product, and of superoperators $\mathcal{E} \otimes
	\mathcal{F}$ as the superoperator having Kraus decomposition $\{ E_i \otimes
	F_j \}_{i,j}$ with $\{ E_i \}_i$ and $\{ F_j \}_j$ Kraus decompositions of
$\mathcal{E}$ and $\mathcal{F}$.
Superoperators represent unitary
transformations $U$ as $\mathcal{E}_U$ with $\{ U \}$ its Kraus
decomposition.
Other transformations are possible, like 
the constant $\text{Set}_{\rho}$ superoperators, transforming any input state in the given state $\rho$; and the probabilistic combination of 
unitaries having Kraus decomposition $\{\sqrt{p_i} U_i\}_i$ with $\sum_i p_i = 1$ and $U_i$ a unitary transformation
(each $U_i$ is applied with a given probability, which is useful for modelling noisy channels and gates).

Density operators can be used to describe the state of a subsystem of a
composite quantum system. Let $\hilbert_{AB} = \hilbert_A \otimes \hilbert_B$
represents a composite system, with subsystems $A$ and $B$. Given a (not
necessarily separable) $\rho^{AB} \in \hilbert_{AB}$, the \emph{reduced density
	operator} of system $A$, $\rho^A = \tr_B(\rho^{AB})$, describes the state of
the subsystem $A$, with $\tr_B$ the \emph{partial trace over $B$}, defined as
the linear transformation such that
$\tr_B(\ketbra{\psi}{\psi'} \otimes \ketbra{\phi}{\phi'}) =
	\ketbra{\psi}{\psi'}\tr(\ketbra{\phi}{\phi'})$.
When applied to pure separable states, the partial trace returns the actual
state of the subsystem. When applied to an entangled state, instead, it
produces a probability distribution of states, because ``forgetting'' the
information on the subsystem $B$ leaves us with only partial
information on subsystem $A$. As an example, the partial trace over the first
qubit of $\rho = \dket{\Phi^+}$ is the maximally mixed state.

\section{A Linear Process Algebra}\label{sec:lqCCS}
In the following sections we describe the syntax and the type system of lqCCS
processes, as well as a reduction-style semantics and a first notion of bisimilarity. Our process calculus is enriched with a linear type system both for reflecting the no-cloning theorem (see~\autoref{noCloning}) and for resolving the minor discrepancy on qubit visibility summarized in the first three rows of~\autoref{tab:examples}.

\subsection{Syntax and Type System}\label{Syntax}

\begin{figure}[!t]
	{
    \small
		\begin{gather*}
			\begin{matrix}
				\infer[\rulename{Nil}]{\Sigma \vdash \nil_{\tilde{e}}}{\tilde{e} \in \tilde{\Sigma}} \qquad
				\infer[\rulename{Tau}]{\Sigma \vdash \tau . P}{\Sigma \vdash P} \qquad
				\infer[\rulename{QMeas}]{\Sigma \vdash \meas{\tilde{e}}{y} . P}{M : \measset(n)                            & |E| = n                           &
				\tilde{e} \in \tilde{E}                                                                                    & E \subseteq \Sigma                & y : \ntype                        & \Sigma \vdash P}                      \\
				\infer[\rulename{Restrict}]{\Sigma \vdash P \setminus c}{\Sigma \vdash P} \qquad
				\infer[\rulename{Sum}]{\Sigma \vdash P + Q}{\Sigma \vdash P                                                & \Sigma \vdash Q} \qquad
				\infer[\rulename{QOp}]{\Sigma \vdash \sop{E}{\tilde{e}} . P}{\mathcal{E} : \opset(n)                       & |E| = n                           & \tilde{e} \in \tilde{E}           & E \subseteq \Sigma & \Sigma \vdash P} \\
				\infer[\rulename{CRecv}]{\Sigma \vdash c?x . P}{c : \chtype{T}                                             & x : T \in\{\btype,\ntype\}        & \Sigma \vdash P} \qquad
				\infer[\rulename{QRecv}]{\Sigma \vdash c?x . P}{c : \chtype{\qtype}                                        & x : \qtype                        & \Sigma \cup \{x\} \vdash P} \qquad
				\infer[\rulename{QSend}]{\{e\} \vdash c!e}{c : \chtype{\qtype}                                             & e : \qtype}                                                                                                   \\
				\infer[\rulename{CSend}]{\emptyset \vdash c!e}{c : \chtype{T}                                              & e : T \in\{\btype,\ntype\}} \qquad
				\infer[\rulename{ITE}]{\Sigma \vdash \ite{e}{P_1}{P_2}}{e : \btype                                         & \Sigma \vdash P_1                 & \Sigma \vdash P_2} \qquad
				\infer[\rulename{Par}]{\Sigma_1 \cup \Sigma_2 \vdash P_1 \parallel P_2}{\Sigma_1 \cap \Sigma_2 = \emptyset & \Sigma_1 \vdash P_1               & \Sigma_2 \vdash P_2}
			\end{matrix}
		\end{gather*}
	}
	\caption{Typing rules for lqCCS}
	\label{typing_rules}
\end{figure}

The syntax of lqCCS is defined as follows
\begin{align*}
	P \Coloneqq & \ K \mid P \parallel P \mid P \setminus c \mid \ite{e}{P}{P}                                                           \\
	K \Coloneqq & \ \nil_{\tilde{e}} \mid \tau . P \mid \sop{E}{\tilde{e}}.P \mid \meas{\tilde{e}}{x}.P \mid c?x . P \mid c!e \mid K + K \\
	e \Coloneqq & \ x \mid b \mid n \mid q \mid \neg e \mid e \lor e \mid e \leq e
\end{align*}
where $b \in \btype$, $n \in \ntype$, $q \in \qtype$, $x \in \varset$, $c \in
	\chanset$ with $\qtype$, $\varset$, $\chanset$ denumerable
sets of respectively qubit names, variables and channels, each typed.
We use $\tilde{e}$ to denote a (possibly empty)
tuple $e_1, \ldots, e_n$ of expressions. 
The process $\nil_{\tilde{e}}$ \emph{discards} the qubits in ${\tilde{e}}$.
It behaves as a deadlock process that maintains ownership of the qubits in ${\tilde{e}}$ and makes them inaccessible to other processes.
As we will see, discard processes are semantically equivalent to any deadlock process using the same qubits.
The process $\nil_{q_1,q_2}$ is e.g.\ equivalent to $\nil_{q_2,q_1}$ and to $c?x.(c!q_1 \parallel c!q_2 \parallel c!x) \setminus c$, since $q_1$ and $q_2$ will never be available.
When $\tilde{e}$ is the empty sequence, we simply write $\nil$ to stress the equivalence with the nil process of standard CCS\@.
	This feature of lqCCS allows to clearly mark which qubits are hidden to the environment, thus relieving bisimilar processes to agree on them.
	Note that ``discard-like'' processes can be written also in qCCS, but are never used in the literature.
A symbol $\mathcal{E}$ denotes a trace-preserving superoperator on
$\qubits{n}$ for some $n > 0$, and we write $\mathcal{E} : \opset(n)$ to
indicate that $\mathcal{E}$ is a superoperator with arity $n$. A symbol $M$
denotes a measurement $\{M_0,\ldots,M_{k-1}\}$ with such operators acting
on $n$ qubits and with $k$ different outcomes: we write $M : \measset(n)$ to
indicate that $M$ is a measurement operator with arity $n$, and denote as $|M|$ the cardinality $k$ of $M$.
We let $M_{01}$ and $M_\pm$ be the projective measurement over the computational and Hadamard basis respectively.
We say that the channel
name $c$ is \emph{bound} in $P \setminus c$ and \emph{free} otherwise. We
denote with $\text{fc}(P)$ the set of free channels and with $\text{fv}(P)$ the
set of free classical variables of $P$, defined as usual.

Affinity is often enforced in quantum process calculi for preventing qubits from being broadcast, which is forbidden by the no-cloning theorem.
We decided to go one step further and to impose a linear type system,
which forces the processes to either send or explicitly discard each qubit that they own.
Therefore, the visibility of qubits is explicitly stated, relieving us from performing an arbitrary choice about the visibility of those qubits that are neither sent nor discarded.

Typing judgments are of the form $\Sigma \vdash P$. The use of quantum names is
subject to linearity and those in use are collected in $\Sigma \subseteq
	\qtype$. The set of types is $\{\qtype, \ntype, \btype\}$, respectively the
type of quantum names, naturals and booleans. The set of channel types is
$\{\chtype{\qtype}, \chtype{\ntype}, \chtype{\btype}\}$. From now on we will
assume that channels and variables are typed, and expressions involving natural
and booleans are typed as standard. The typing system is in~\autoref{typing_rules}, where for a set $A$ we use $\tilde{A}$
to denote the set of tuples $\tilde{a}$ such that any element of $A$ occurs exactly once in $\tilde{a}$.
Linearity is enforced by
the combination of rules \textsc{Nil}, \textsc{QSend} and \textsc{Par}.
In particular, the
former two are the only rules that introduce new qubits into the quantum context,
therefore each quantum name must be sent along some channel or discarded; while the latter ensures
that each qubit is not shared between parallel processes.

It is easy to show that the typing of processes is unique, therefore in the following we will simply write $\Sigma_P$ for the unique set of quantum names such that $\Sigma_P \vdash P$.
	\begin{restatable}[Unique Type]{proposition}{uniqueness}\label{thm:uniqueness}
		If $\,\Sigma \vdash P$ and $\Sigma' \vdash P$ then $\Sigma = \Sigma'$.
	\end{restatable}
	\begin{proofsketch}
		By induction on the derivations of $\Sigma \vdash P$ and $\Sigma' \vdash P$.
	\end{proofsketch}

As a simple example of lqCCS, we present the following quantum lottery
protocol \proc{QL}, which uses a qubit to randomly select a winner between two
competitors.
\begin{example}\label{ex:quantumlottery}
	Let $\proc{QL} = H(q).\meas[M_{01}]{q}{x}.((\ite{x = 0}{a!1}{b!1}) \parallel \nil_{q})$.
	The qubit $q$ is firstly transformed with $H$ and then measured in the computational basis.
	Depending on the outcome, stored in $x$, either Alice or Bob is announced as the winner (through $a!1$ and $b!1$ respectively).
	Finally, the qubit is discarded as it is no longer needed.
	The unique typing of $\proc{QL}$ is given by $\{q\} \vdash
		\proc{QL}$, with $a : \chtype{\ntype}$, $b : \chtype{\ntype}$, $x : \ntype$,
	and $q : \qtype$.
\end{example}

\subsection{Operational Semantics}\label{Semantics}

\begin{figure}[!t]
  \small
	\begin{gather*}
		\infer[\rulename{SCParNil}]{P \parallel \nil \equiv P}{} \qquad
		\infer[\rulename{SCParComm}]{P \parallel Q \equiv Q \parallel P}{} \qquad
		\infer[\rulename{SCParAssoc}]{P \parallel (Q \parallel R) \equiv (P \parallel Q) \parallel R}{} \\[0.2cm]
		\infer[\rulename{SCSumNil}]{P + \nil_{\tilde{e}} \equiv P}{} \qquad
		\infer[\rulename{SCSumComm}]{P + Q \equiv Q + P}{} \qquad
		\infer[\rulename{SCSumAssoc}]{P + (Q + R) \equiv (P + Q) + R}{} \\
		\infer[\rulename{SCIteF}]{\ite{\false}{P}{Q} \equiv Q}{} \qquad
		\infer[\rulename{SCIteT}]{\ite{\true}{P}{Q} \equiv P}{}  \qquad
		\infer[\rulename{SCValExpr}]{P \equiv P[\sfrac{v}{e}]}{e \Downarrow v} \\
		\infer[\rulename{SCRestrNil}]{\nil_{\tilde{e}} \setminus c \equiv \nil_{\tilde{e}}}{} \qquad
		\infer[\rulename{SCRestrOrd}]{P \setminus c \setminus d \equiv P \setminus d \setminus c}{} \qquad
		\infer[\rulename{SCRestrPar}]{(P \parallel Q) \setminus c \equiv P \parallel (Q \setminus c)}{c \not\in \text{fc}(P)}
	\end{gather*}
	\caption{Structural congruence}
	\label{struct_cong}
\end{figure}

The operational semantics of lqCCS is defined as a probabilistic reduction system
$(\confbot, \to)$ over closed processes (i.e., processes $P$ such that $\text{fv}(P) = \emptyset$), where
\begin{itemize}
	\item $\confbot$ is $\confset \cup \{\bot\}$, with $\confset$ the set of configurations of the form $\conf{\rho, P}$,
	      and $\bot$ the ``deadlock'' configuration that always evolves in $\singleton{\bot}$;
	\item $\to\ \subseteq \confbot \times \dist{\confbot}$ is the probabilistic transition relation.
\end{itemize}

Given a set $\Sigma = \{q_1,\ldots,q_n\} \subseteq \qtype$, a global
quantum state $\rho$ is a density operator over $\hilbert_\Sigma = \qubits{n}$, where $q_i$ refers to the $i$-th qubit in $\rho$. Expressions $e$ are
evaluated through a big step semantics $e \Downarrow v$ with $v$ a value, i.e.\
either $n \in \ntype$, $b \in \btype$, or $x \in \varset$.
We restrict
ourselves to standard boolean and arithmetic operations, and therefore
omit the rules and assume free variables are not
evaluated.

The type system is extended to configurations by considering the qubits of the
underlying quantum state.
In the following, we denote as $\Sigma_\rho$ the set of qubits appearing in $\rho$.
\begin{definition}
	Let $\conf{\rho, P} \in \confset$ and $\Delta \in \dist{\confbot}$.
	We let
	$(\Sigma_\rho, \Sigma_P) \vdash \conf{\rho, P}$ if $\Sigma_P \subseteq \Sigma_\rho$.
	We let $(\Sigma, \Sigma') \vdash \bot$ for any $\Sigma$ and $\Sigma'$, and $(\Sigma, \Sigma') \vdash \Delta$ if $(\Sigma, \Sigma') \vdash
		\iconf$ for any $\iconf$ in $\support{\Delta}$.
\end{definition}
Hereafter, we restrict ourselves to well-typed distributions.
We extend the standard structural congruence relation for CCS~\cite{milner_functions_1992} as presented in
\autoref{struct_cong}, and we impose congruent processes to be typed by the
same $\Sigma$. The new rules
allow the evaluation of expressions and reduction of
$\ite{\!\blank\!}{\!\blank\!}{\!\blank\!}$ occurrences. We lift the congruence
to distributions of configurations by linearity and imposing $\overline{\bot} \equiv \overline{\bot}$ and $\singleton{\conf{\rho, P}}
	\equiv \singleton{\conf{\rho, P'}}$ whenever $P \equiv P'$.

	The transition relation $\to$ is the smallest relation that satisfies the rules
	in \autoref{semantics}, augmented with $\mathcal{C} \to \singleton{\bot}$ if there is no $\Delta$ such that $\mathcal{C}
		\to \Delta$~\cite{deng_bisimulations_2018}.
We have the standard rules for CCS operators, so for example 
a process $\tau.P$ performs a silent action that does not affect the quantum state, and then continues its evolution as $P$,
while $P \setminus c$ behaves as $P$ but the channel $c$ is \emph{restricted}, i.e.\ $P$ cannot synchronize with other processes on that channel.
Along them we introduce rules for superoperators and measurements.
Since the arity of $\mathcal{E}$ can be smaller than the number of qubits in the quantum state $\rho$,
we define $\mathcal{E}_{\tilde{q}}$ as
the superoperator obtained by composing
$(i)$ a suitable set of SWAP unitaries to bring the qubits $\tilde{q}$ in the first positions;
$(ii)$ the tensor product of the superoperator $\mathcal{E}$ with the identity on
untouched qubits on the right; and
$(iii)$ the inverse of the SWAP operators of point $(i)$ to recover the original order of qubits~\cite{lalire_relations_2006}.
The same mechanism is applied to measurements
to obtain ${(M_m)}_{\tilde{q}}$ from $M_m$.
If $\Delta = \sum_{i} \distelem{p_i} \singleton{\conf{\rho_i, P_i}}$, we let $\Delta
	\setminus c$ and $\Delta \parallel Q$ denote distributions $\sum_{i}
	\distelem{p_i} \singleton{\conf{\rho_i, P_i\setminus c}}$ and $\sum_{i} \distelem{p_i}
	\singleton{\conf{\rho_i, P_i\parallel Q}}$.
In the following, we lift $\to$ to distributions, writing $\to$ for $\text{lift}(\to) \in \dist{\confbot} \times \dist{\confbot}$.

Noteworthy, the typing is preserved
 by the transition relation.
\begin{restatable}[Typing Preservation]{theorem}{typingpreservation}
  If $(\Sigma_\rho, \Sigma_P) \vdash \conf{\rho, P}$ and $\conf{\rho, P} \longrightarrow \Delta$ then
  $(\Sigma_\rho, \Sigma_P) \vdash \Delta$.
\end{restatable}
\begin{proofsketch}
		By induction on the derivation of $\conf{\rho, P} \longrightarrow \Delta$.
		The only interesting case is for \textsc{Reduce}, for which we prove that substitution works
		if the new name is not in the typing context of the process, i.e.\ that if $\Sigma \cup \{x\} \vdash P$ and $v \not\in \Sigma$ then $\Sigma \cup \{v\} \vdash P[\sfrac{v}{x}]$.
    Then it suffices to note that $v \not\in \Sigma$ is guaranteed by the typing of parallel processes. Note also that $\Sigma_\rho$ is not impacted by any rule, therefore it is trivially preserved.
\end{proofsketch}

\begin{example}
	The semantics of $\proc{QL}$ from~\autoref{ex:quantumlottery} on quantum state $\dket{0}$ is as follows
	\begin{align*}
		 & \singleton{\conf{\dket{0}, \proc{QL}}}
		\to \singleton{\confl\dket{+},\meas[M_{01}]{q}{x}.((\ite{x = 0}{a!1}{b!1}) \parallel \nil_q)\confr}                                            \\
		 & \ \to \Big(\singleton{\conf{\dket{0}, \ite{0 = 0}{a!1}{b!1}}} \psum{1/2} \singleton{\conf{\dket{1}, \ite{1 = 0}{a!1}{b!1}}}\Big) \parallel \nil_q \\
		 & \ \equiv \Big(\singleton{\conf{\dket{0}, a!1}} \psum{1/2} \singleton{\conf{\dket{1}, b!1}}\Big) \parallel \nil_q
	\end{align*}
\end{example}

\begin{figure}[!t]
  \small
	\begin{gather*}
		\infer[\rulename{Tau}]{\conf{\rho, \tau . P + Q} \longrightarrow \singleton{\conf{\rho, P}}}{} \quad
		\infer[\rulename{Restrict}]{\conf{\rho, P \setminus c} \longrightarrow \Delta \setminus c}{\conf{\rho, P} \longrightarrow \Delta} \quad
		\infer[\rulename{QOp}]{\conf{\rho, \sop{E}{\tilde{q}} . P + Q} \longrightarrow \singleton{\conf{\sop[\tilde{q}]{E}{\rho}, P}}}{} \\
		\infer[\rulename{QMeas}]{\conf{\rho, \meas{\tilde{q}}{y} . P + Q} \longrightarrow \sum_{m=0}^{|M| - 1} \distelem{p_m}{\singleton{\conf{\frac{\rho_m}{p_m}, P[\sfrac{m}{y}]}}}}{\rho_m = (M_m)_{\tilde{q}} (\rho) & p_m = \tr(\rho_m)} \qquad
		\infer[\rulename{Par}]{\conf{\rho, P \parallel Q} \longrightarrow \Delta \parallel Q}{\conf{\rho, P} \longrightarrow \Delta} \\
		\infer[\rulename{Reduce}]{\conf{\rho, (c!v + R) \parallel ((c?x . P) + Q)} \longrightarrow \singleton{\conf{\rho, P[\sfrac{v}{x}]}}}{} \qquad
		\infer[\rulename{Congr}]{\conf{\rho, P} \longrightarrow \Delta'}{P \equiv Q & \conf{\rho, Q} \longrightarrow \Delta & \Delta \equiv \Delta'}
	\end{gather*}

	\caption{lqCCS Semantics}
	\label{semantics}
\end{figure}

\subsection{A First Notion of Behavioural Equivalence}

Since quantum mechanics is intrinsically probabilistic, quantum processes are
commonly compared by using some probabilistic version of
bisimilarity~\cite{lalire_process_2004,lalire_relations_2006,feng_probabilistic_2007,feng_bisimulation_2012,deng_open_2012}.
We follow the approach of~\citet{hennessy_exploring_2012}, defining
bisimulations directly on distributions.
Differently from the previous proposals, we do not use labels but rely instead on contexts and barbs, i.e.\ defining a saturated bisimilarity à la~\citet{bonchi_general_2014}.

We start by defining barbs, i.e.\ atomic observable properties of the lqCCS
processes.
	\begin{definition}
		A \emph{process barb} is a predicate $\downarrow_{c}$ on processes satisfied
		by $P$ (written $P\!\!\downarrow_{c}$) if $P \equiv (c!e + R) \parallel Q$ for
		some $Q, R$.
		A \emph{distribution barb} is a predicate $\downarrow_{b}^p$ on distributions
		such that
		\begin{itemize}
			\item $\Delta$ satisfies $\downarrow_{c}^p$, written $\Delta\!\downarrow_{c}^p$, if $\sum_{P \downarrow_c} \Delta\big(\conf{\rho, P}\big) = p$;
			\item  $\Delta$ satisfies $\downarrow_{\bot}^p$, written $\Delta\!\downarrow_{\bot}^p$, if $\Delta(\bot) = p$.
		\end{itemize}
	\end{definition}
Intuitively, the barbs of a process are the visible channels on which a value is ready
to be sent, while the barbs of a distribution are defined as the probability of
having a process capable to send on a given channel, or as the probability of
having a deadlocked process. Notice that if $\Delta\!\downarrow_{\bot}^p$, then
it must be $\Delta = \singleton\bot \psum{p} \Delta'$ for some $\Delta'$ such that $\Delta'(\bot) = 0$.
Note also that barbs are purely classical.

	In saturated bisimilarities,
	contexts $B[\blank]$ play the
	role of observers that are used for discriminating processes.
	In this first version of bisimilarity, they are defined as lqCCS processes with a typed hole.

\begin{definition}
	A context $B[\blank]_{\Sigma}$ is generated by the production
	$B[\blank]_{\Sigma} \Coloneqq [\blank]_{\Sigma} \parallel P$, up to structural congruence and typed
  according to the rules in~\autoref{typing_rules}
	and to the following one
	\[
		\infer[\rulename{Hole}]{\Sigma' \vdash [\blank]_{\Sigma} \parallel P}
		{\Sigma' \setminus \Sigma \vdash P  & \Sigma \subseteq \Sigma'}
	\]
\end{definition}
A process $P$ is
applied to contexts by replacing the hole with $P$.
Intuitively, a context $\Sigma' \vdash B[\blank]_{\Sigma}$ is a function that
given a process $P$ returns a process $B[P]$ obtained by replacing $P$ for $[\blank]$,
where $\Sigma$ is the typing context of the valid inputs and $\Sigma'$ the one
of the outputs.
Note that a context can own some qubits and each qubit cannot be referred to in both $P$ and $B[\blank]$.
We apply $\Sigma' \vdash B[\blank]_{\Sigma}$ to configurations $(\Sigma_\rho, \Sigma_P) \vdash \conf{\rho, P}$ obtaining $(\Sigma_\rho, \Sigma') \vdash \conf{\rho, B[P]}$ when
	$\Sigma' \subseteq \Sigma_\rho$ and $\Sigma = \Sigma_P$, i.e.\ when the qubits referred by $B[\blank]$ are defined in $\rho$ and the process $P$ is as prescribed by $B[\blank]$.
	We write $B[\conf{\rho, P}]$ for $\conf{\rho, B[P]}$, $B[\bot]$ for $\bot$, and
	$B[\Delta]$ for the distribution obtained by applying $B[\blank]$ to the support of $\Delta$.
  It is trivial to show that if $\Delta$ and $\Theta$ are typed by the same typing context, then $B[\Delta]$ is defined if and only if $B[\Theta]$ is defined, and that their typing is unique. %

In the following, we only consider well-typed distributions and contexts, and we impose bisimulations to be over distributions of the same type (thus we avoid specifying types).

\begin{definition}[s-bisimilarity]
	A relation $\rel \subseteq \dist{\confbot} \times \dist{\confbot}$ is a
	\emph{saturated bisimulation} if $\Delta\,\rel\,\Theta$ implies
	that for any context $B[\blank]$ it holds
	\begin{itemize}
		\item $\Delta\!\downarrow_{b}^p$ if and only if
		      $\Theta\!\downarrow_{b}^p$;
		\item whenever $B[\Delta] \rightarrow \Delta'$, there exists $\Theta'$
		      such that $B[\Theta] \rightarrow \Theta'$ and $\Delta'\;\rel\;\Theta'$;
		\item whenever $B[\Theta] \rightarrow \Theta'$, there exists $\Delta'$
		      such that $B[\Delta] \rightarrow \Theta'$ and $\Delta'\;\rel\;\Theta'$.
	\end{itemize}
	Let \emph{saturated bisimilarity} $\sim_{s}$ be the largest
	saturated bisimulation.

	We say that two processes $P, Q$ are
	saturated bisimilar if $\singleton{\conf{\rho, P}} \sim_{s}
		\singleton{\conf{\rho, Q}}$ for any $\rho$.
\end{definition}
When encoding a protocol or its specification in lqCCS, each qubit must be sent on a visible channel if its value is a relevant aspect of the given protocol, discarded otherwise.
Notice that, thanks to linearity, the visibility of qubits cannot be ambiguous, and that $\sim_{s}$ replicates the results of the previous proposals in the unambiguous cases (see the first three rows of~\autoref{tab:examples}).

\begin{example}
	Consider \proc{QL} of~\autoref{ex:quantumlottery}. It is not
	difficult to show that $\singleton{\conf{\dket{0}, \proc{QL}}}$ is
	bisimilar to $\Delta_{\proc{QL}} = \Big(\overline{\conf{\rho,
				\tau.\tau.a!1}} \psum{1/2} \overline{\conf{\rho, \tau.\tau.b!1}}\Big) \parallel
		\nil_q$ for any $\rho \in \qubit$.
	Note that $\Delta_{\proc{QL}}
		\rightarrow \Delta_{\proc{QL}}' \rightarrow
		\Delta_{\proc{QL}}''$ with
		$\Delta_{\proc{QL}}' = \Big(\overline{\conf{\rho, \tau.a!0}} \psum{1/2} \overline{\conf{\rho, \tau.b!1}}\Big) \parallel \nil_q$ and
		$\Delta_{\proc{QL}}'' = \Big(\overline{\conf{\rho, a!0}} \psum{1/2} \overline{\conf{\rho, b!1}}\Big) \parallel \nil_q$.
  	It suffices then to give the relation $\rel$ below, which is a bisimulation once closed for congruence and contexts
    \[
      \rel = \bigg\{\Big(\overline{\conf{\dket{0},\proc{QL}}},\Delta_{\proc{QL}}\Big), \Big(\overline{\conf{\dket{+}, \meas[M_{01}]{q}{x}.(c!x \parallel \nil_q)}},\Delta_{\proc{QL}}'\Big), \Big(\Delta_{\proc{QL}}'',\Delta_{\proc{QL}}''\Big),\Big(\overline{\bot},\overline{\bot}\Big)\bigg\}
    \]

	Finally, note that this result depends on the use of discard, meaning that the messages over $a$ and $b$ are the outcome of the protocol, and not the resulting quantum state.
		For example, the distribution $\singleton{\conf{\ketbra{0}, H(q).\meas[M_{01}]{q}{x}.((\ite{x = 0}{a!1}{b!1}) \parallel c!q)}}$ is not bisimilar to
    $\overline{\conf{\ketbra{0}, \tau.\tau.a!1 \parallel c!q}} \psum{1/2} \overline{\conf{\ketbra{0}, \tau.\tau.b!1 \parallel c!q}}$. 
    To prove that, it suffices to consider the context $B[\blank] = [\blank] \parallel c?x.\meas[M_{01}]{x}{y}.\ite{y = 0}{\nil_x}{\mathit{fail}!x}$, which will eventually exhibit the barb $\mathit{fail}$ with the former distribution only.
\end{example}

As a result of our saturated approach, entangled pairs are correctly distinguished from separable states with the same partial trace (see the fourth row of the same table).
\begin{example}\label{ex:entanglement}
	The processes in the fourth row of~\autoref{tab:examples} are distinguished by the context $B[\blank] = [\blank] \parallel c?x.c?y.\meas{x, y}{z}.(\ite{z = (01)_2}{d!0}{\nil} \parallel \nil_{x, y})$,
	because it will eventually exhibit the barb $d$ if the state decays in $\ket{01}$ (which is only possible for $\frac{1}{4} I$).
\end{example}

Finally, notice that saturated bisimilarity equates the discard process $\nil_{q_1,q_2}$ with both $\nil_{q_2,q_1}$ and $c?x.(c!q_1 \parallel c!q_2 \parallel c!x) \setminus c$.
To prove that, it suffices to give the relation associating $\Delta$ with $\Delta[\sfrac{P}{Q}]$ for any choice of $P$ and $Q$ among the processes above.
This is clearly a bisimulation, once closed for contexts and congruence: for any context $B[\blank]$, if $B[P] \rightarrow \Delta$ then $B[Q] \rightarrow \Delta[\sfrac{P}{Q}]$.

\section{Curbing the Power of Non-Deterministic Contexts}\label{sec:db}
The example in the fifth row of~\autoref{tab:examples}, shows that saturated bisimilarity is inadequate for the quantum case
(similarly to QPAlg~\cite{lalire_relations_2006} and qCCS~\cite{deng_open_2012}).
Moreover, we trace the cause of this problem 
to the interaction between probabilistic and non-deterministic behaviour.
We define a new semantics where non-determinism in contexts is constrained.
These constraints solve the problem above, matching a defining feature of quantum systems,
i.e.\ that states cannot be observed without being affected.
Nonetheless, processes can still perform different non-deterministic choices upon known classical
values.
As a result of that, processes in the sixth row of~\autoref{tab:examples} are distinguished.

\subsection{Non-Deterministic Issues}\label{nonDeterministicIssues}

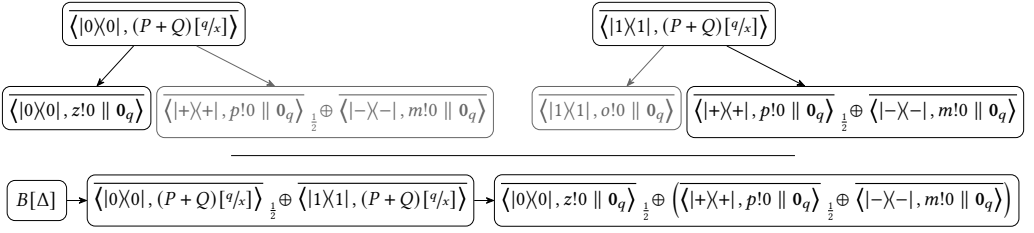
\begin{figure}[!t]
	\resizebox{\textwidth}{!}{
		\begin{tikzpicture}[>={Stealth[round]}]
			\tikzstyle{longsquiggly} = [->,line join=round,decorate,decoration={zigzag,segment length=4,amplitude=.9,pre=lineto,pre length=5pt,post=lineto,post length=2pt}]
			\tikzstyle{dim} = [minimum width=30pt, minimum height=20pt]
			\tikzstyle{subelem} = [dim, draw, rectangle, rounded corners]
			\tikzstyle{leftpat} = [pattern={Hatch[angle=45,distance=3pt]},pattern color=black!17]
			\tikzstyle{rightpat} = [pattern={Dots[angle=45,distance=3pt]},pattern color=black!22]
			\node (zpq) at (0,0) [subelem] {$\sconf{\ketbra{0}, (P+Q)[\sfrac{q}{x}]}$};
			\node (zp) [subelem,below left=20pt and -45pt of zpq] {$\sconf{\ketbra{0}, z!0 \parallel \nil_{q}}$};
			\node (zq) [subelem,below right=20pt and -45pt of zpq,black!60] {$\sconf{\ketbra{+}, p!0 \parallel \nil_{q}} \psum{\frac{1}{2}} \sconf{\ketbra{-}, m!0 \parallel \nil_{q}}$};
			\draw[->] (zpq) -- (zp);
			\draw[->,black!60] (zpq) -- (zq);
			\node (opq) [subelem,right=175pt of zpq] {$\sconf{\ketbra{1}, (P+Q)[\sfrac{q}{x}]}$};
			\node (op) [subelem,below left=20pt and -45pt of opq,black!60] {$\sconf{\ketbra{1}, o!0 \parallel \nil_{q}}$};
			\node (oq) [subelem,below right=20pt and -45pt of opq] {$\sconf{\ketbra{+}, p!0 \parallel \nil_{q}} \psum{\frac{1}{2}} \sconf{\ketbra{-}, m!0 \parallel \nil_{q}}$};
			\draw[->,black!60] (opq) -- (op);
			\draw[->] (opq) -- (oq);
			\node[fit=(zpq) (zp) (zq) (opq) (op) (oq)] (bound1) {};
			\draw ($(bound1) + (-5, -1.5)$) -- ($(bound1) + (5, -1.5)$);
			\begin{scope}[xshift=-59pt,yshift=-88pt]
				\node (step0) at (0,0) [subelem] {$B[\Delta]$};
				\node (step1) [subelem, right=10pt of step0] {$\sconf{\ketbra{0}, (P+Q)[\sfrac{q}{x}]} \psum{\frac{1}{2}} \sconf{\ketbra{1}, (P+Q)[\sfrac{q}{x}]}$};
				\node (step2) [subelem, right=10pt of step1] {$\sconf{\ketbra{0}, z!0 \parallel \nil_{q}} \psum{\frac{1}{2}} \Big(\sconf{\ketbra{+}, p!0 \parallel \nil_{q}} \psum{\frac{1}{2}} \sconf{\ketbra{-}, m!0 \parallel \nil_{q}}\Big)$};
				\draw[->] (step0) -- (step1);
				\draw[->] (step1) -- (step2);
			\end{scope}
		\end{tikzpicture}
	}
	\caption{
		On the bottom, the evolution of $B[\Delta]$ into $\Delta'$ and then $\Delta''$. The last step is built by convex combination
		of two freely chosen moves of the subdistributions, which are displayed above.
		It is clear that observers can make different
		non-deterministic choices for each subdistribution, also based on quantum states.
	}
	\label{fig:nondetissue2}
\end{figure}

In quantum theory, the encoding of probability distributions of quantum states as
density operators is not injective.
Formally, a density operator
represents an equivalence class of distributions of quantum states that behave the same according
to quantum theory.
They are defined by the relation $\rel \subseteq
\dist{\hilbert} \times \dist{\hilbert}$ such that $\Delta\,\rel\,\Theta$ whenever $\sum_{\ket{\psi}}
\Delta(\ket{\psi}) \ketbra{\psi}{\psi} = \sum_{\ket{\psi}} \Theta(\ket{\psi})
\ketbra{\psi}{\psi}$.

Consider a pair of non-biased random qubit sources, the first sending a qubit in state $\kz$ or
$\ko$, the second in state $\kpl$ or $\km$.
Quantum theory prescribes that such two sources cannot be
distinguished by any observer, as the received qubit behaves the same (see~\citet[Section 2.4.2]{nielsen_quantum_2010}).
Indeed, the
(mixed) states of the qubits sent by the two sources are represented by the same density operator $\frac{1}{2}I$.
One expects the lqCCS encoding of these
sources to be bisimilar. Somewhat surprisingly, this is not the case.
\begin{example}\label{ex:broken-nondet}
  Take the fifth row of~\autoref{tab:examples}. After two steps with empty contexts, the two distributions evolve into $\Delta = \singleton{\conf{\dket{0}, c!q}} \psum{1/2} \singleton{\conf{\dket{1}, c!q}}$ and $\Theta = \singleton{\conf{\dket{+}, c!q}} \psum{1/2} \singleton{\conf{\dket{-}, c!q}}$, encoding the qubit sources above.
	To see that $\Delta \not\sim_{s} \Theta$, take $B[\blank] = [\blank] \parallel
		c?x.(P+Q)$ where
	\begin{align*}
    P\ =\  & \meas[M_{01}]{x}{y} . ((\ite{y=0}{z!0}{o!0}) \parallel \nil_{x})\text{, and} \\
    Q\ =\  & \meas[M_{\pm}]{x}{y} . ((\ite{y=0}{p!0}{m!0}) \parallel \nil_{x}).
	\end{align*}
	$B[\Delta]$ reduces to $\Delta' = \Big(\singleton{\conf{\dket{0}, P+Q}} \psum{1/2}
		\singleton{\conf{\dket{1}, P+Q}}\Big)[\sfrac{q}{x}]$, and $B[\Theta]$ can only reduce to $\Theta' = \Big(\singleton{\conf{\dket{+}, P+Q}}
		\psum{1/2} \singleton{\conf{\dket{-}, P+Q}}\Big)[\sfrac{q}{x}]$ to match this move.
  By choosing $P$ and $Q$ respectively in the left and right part, $\Delta'$ reduces to $\Delta'' = \Big(\singleton{\conf{\dket{0}, z!0}} \psum{1/2} \big(\singleton{\conf{\dket{+}, p!0}} \psum{1/2} \singleton{\conf{\dket{-}, m!0}}\big)\Big) \parallel \nil_{q}$
	that exhibits the barb $z$ but not $o$.
	It is easy to check that $\Theta$ cannot replicate this behaviour: for any
	choice of $P$ and $Q$ it will either express both $z$ and $o$, or none of
	them.
\end{example}
This result is paradigmatic, where different mixtures of quantum states are discriminated by a non-deterministic context that chooses how to reduce
based on the value of the received qubit, which in theory should be unknown (see~\autoref{fig:nondetissue2}).
Note also that the two sources above have a deterministic behaviour, while non-determinism is only introduced by the
context.

We argue that unconstrained non-deterministic contexts are too strong for
representing the real capacity of discriminating quantum processes.
Therefore, in the following, we give a new semantics that constrains non-deterministic
contexts so that they cannot apply the strategy above to discriminate between
processes.  Note that removing $+$ from the contexts is not sufficient,
as we can replicate the same non-deterministic behaviour
of~\autoref{ex:broken-nondet} with just the parallel composition, e.g.\ with
$B'[\blank] = [\blank] \parallel c?x.P \parallel c?x.Q$.

\subsection{Constrained Bisimilarity}

\begin{figure}[!t]
  \small
	\centering
	\begin{gather*}
		\infer[\rulename{OQOp}]{\conf{\rho, P, \sop{E}{\tilde{q}} . R} \longsquiggly_\varepsilon \singleton{\conf{\sop[\tilde{q}]{E}{\rho}, P, R}}}{} \qquad
		\infer[\rulename{OQMeas}]{\conf{\rho, P, \meas{\tilde{q}}{y} . R} \longsquiggly_\varepsilon \sum_{m=0}^{|M| -1} \distelem{p_m}{\singleton{\conf{\frac{\rho_m}{p_m}, P, R[\sfrac{m}{y}]}}}}{\rho_m = (M_m)_{\tilde{q}}(\rho) & p_m = \tr(\rho_m)} \\
		\infer[\rulename{Input}]{\conf{\rho, ((c!v + P) \parallel Q) \setminus D, c?x.R + S} \longsquiggly_\varepsilon \singleton{\conf{\rho, Q \setminus D, R[\sfrac{v}{x}]}}}{c \not\in D} \\
		\infer[\rulename{Output}]{\conf{\rho, (((c?x . P) + P') \parallel Q) \setminus D, c!v} \longsquiggly_\varepsilon \singleton{\conf{\rho, (P[\sfrac{v}{x}] \parallel Q) \setminus D, \nil}}}{c \not\in D} \\
		\infer[\rulename{Process}]{O[\conf{\rho, P}] \longsquiggly_\diamond O[\Delta]}{\conf{\rho, P} \longrightarrow \Delta} \qquad
		\infer[\rulename{ParL}]{\conf{\rho, P,  R \expar S} \longsquiggly_{\ell\cdot\lambda} \Delta\expar S }{\conf{\rho, P, R} \longsquiggly_\lambda \Delta} \qquad
		\infer[\rulename{ParR}]{\conf{\rho, P,  R \expar S} \longsquiggly_{r\cdot \lambda} R\expar \Delta}{\conf{\rho, P, S} \longsquiggly_\lambda \Delta} \\
		\infer[\rulename{Congr}]{\conf{\rho, P, R} \longsquiggly_\pi \Delta}{P \equiv P' & R \equiv_o R' & \conf{\rho, P', R'} \longsquiggly_\pi \Delta' & \Delta' \equiv_o \Delta} \qquad
	\end{gather*}
	
	\caption{lqCCS enhanced semantics}
	\label{ext_semantics}
\end{figure}

We consider a special set of processes, called \emph{observers}, that are used as constrained contexts for lqCCS\@.
An observer $R$ is a process without silent action, restrictions and with
non-deterministic choice limited to sums of receptions. 
We constrain non-deterministic choices for matching the observational limitations prescribed by quantum theory, while silent actions and restrictions are safely omitted for convenience, as they do not increase the discriminating capabilities of contexts (as proved by \citet{bonchi_general_2014}).

Formally, an observer is defined by taking the pre-terms generated by the following grammar and by imposing additional constraints over parallel composition
\begin{align*}
	R \Coloneqq & \ \nil_{\tilde{e}} \mid c!e \mid T \mid R \expar R \mid \ite{e}{R}{R} \mid
	\sop{E}{\tilde{x}}.R \mid \meas{\tilde{x}}{y}.R\\
	T \Coloneqq & \ c?x.R \mid T + T
\end{align*}
A pre-term $R$ is an observer if $T \not\equiv c?x.R' + c?y.R'' + T'$ for any $T$ appearing in $R$.
This additional constraint forbids parallel processes in $R$ from performing non-deterministic choices upon reception, 
and it could be decided by a suitable extension of our type system.

Despite the syntactical constraints, the major difference between a process $P$ and an observer $R$ is their treatment in our enhanced semantics, which is given on extended configurations: triples with an observer as the third element.
In the following, we sometimes write $\conf{\rho, P}$
for $\conf{\rho, P, \nil}$.
We say that $(\Sigma_\rho, \Sigma) \vdash \conf{\rho, P, R}$ if
$\Sigma = \Sigma_P \cup \Sigma_R$, $\Sigma_P \cap \Sigma_R = \emptyset$, and $\Sigma \subseteq \Sigma_\rho$.
Distributions of configurations
are typed as usual.
Contexts $O[\blank]_\Sigma$ are defined from observers $R$ up to
structural congruence as $O[\blank]_{\Sigma} \Coloneqq [\blank]_{\Sigma} \expar
	R$.
Context application $O\big[\conf{\rho, P, R}\big]$ is defined over configurations as $\conf{\rho, P, O[R]}$.

\begin{figure}
	{
		\small
		\begin{gather*}
			\infer[\rulename{SCCSumNil}]{P + \nil_{\tilde{e}} \equiv_o P}{} \qquad
			\infer[\rulename{SCCSumComm}]{P + Q \equiv_o Q + P}{} \qquad
			\infer[\rulename{SCCSumAssoc}]{P + (Q + R) \equiv_o (P + Q) + R}{} \\
			\infer[\rulename{SCCIteF}]{\ite{\false}{P}{Q} \equiv_o Q}{} \qquad
			\infer[\rulename{SCCIteT}]{\ite{\true}{P}{Q} \equiv_o P}{}  \qquad
			\infer[\rulename{SCCValExpr}]{P \equiv_o P[\sfrac{v}{e}]}{e \Downarrow v}
		\end{gather*}
	}
	\caption{Structural congruence of lqCCS observers}
	\label{fig:obscongr}
\end{figure}

We define a new congruence for observers, named $\equiv_o$, based on $\equiv$ but
lacking rules for parallel  and restriction operators (see~\autoref{fig:obscongr}).
We lift $\equiv_o$ to distributions by linearity and imposing
$\overline{\bot} \equiv_o \overline{\bot}$ and $\singleton{\conf{\rho, P, R}}
	\equiv_o \singleton{\conf{\rho, P', R'}}$
if $P \equiv P'$ and $R \equiv_o R'$.
Note that the commutativity of the parallel operator is preserved for processes, while it does not hold for observers:
we want to distinguish the locus of the reductions for properly constraining the non-deterministic evolution of distributions.
We give an \emph{enhanced semantics} for lqCCS in~\autoref{ext_semantics}, adopting
the style of~\citet{DeganoP01}:
an arrow $\longsquiggly_\pi \subseteq \confbot \times \dist{\confbot}$ for any
index $\pi$, which encodes the non-deterministic choice of the observer (or $\diamond$ if only the process evolves).
Indices are strings in $\{\ell, r \}^* \cup \{\diamond\}$, and we write
$\varepsilon$ for the empty string, $\cdot$ for string concatenation, and
$\lambda$ for indices different from $\diamond$. 

In the rules \textsc{In} and
\textsc{Out}, we write $\setminus D$ for any sequence of restrictions. We
also write $\Delta \parallel R$ and $R \parallel \Delta$ meaning that $R$ is composed with the observers
of $\Delta$.
As for the
probabilistic semantics, the $\longsquiggly_\pi$ is extended by imposing that
$\iconf \longsquiggly_\pi \singleton\bot$ if there is no $\Delta$ such that $\mathcal{C}  \longsquiggly_\pi  \Delta$
In the following, we lift semantic transitions, and write
$\longsquiggly_\pi$ instead of $\text{lift}(\longsquiggly_\pi) \in
	\dist{\confbot} \times \dist{\confbot}$.
The observational power is limited by the indexing as the lifting of $\longsquiggly_\pi$ allows a
distribution to evolve only when all its configurations reduce by performing the same choice. 
The absence of rules for observer synchronization does not limit the observational power~\cite{bonchi_general_2014}.
Note that the transition system is still non-deterministic,
but different non-deterministic choices produce different indices. Nevertheless,
the observer is still capable to behave differently in different configurations
of the same distribution, but only through the
$\ite{\!\blank\!}{\!\blank\!}{\!\blank\!}$ construct or because of a sum of receptions
(where the choice is performed by the sending process and not by the observer).
This ensures that the observer choices are based
on \textit{classical information} obtained from the process.

We say that a configuration satisfies the barb $c$, written $\conf{\rho, P,
		R}\!\downarrow_c$, if $P \equiv (c!v + S) \parallel Q$ or $R \equiv c!v
	\parallel S$ (notice that we use $\equiv$ instead of $\equiv_o$ for $R$, thus
considering also commutativity, associativity, and identity). The extension of
barbs to distributions is as usual.

We now define constrained saturated bisimulation with the usual assumption that it is defined over distributions of the same type, and that contexts are taken accordingly.
\begin{definition}[cs-bisimilarity]
	A relation $\rel \subseteq \dist{\confbot} \times \dist{\confbot}$ is a
	\emph{constrained saturated bisimulation} if $\Delta\,\rel\,\Theta$ implies that
	for any context $O[\blank]$ it holds
	\begin{itemize}
		\item $\Delta\!\downarrow_{b}^p$ if and only if $\Theta\!\downarrow_{b}^p$;
		\item whenever $O[\Delta] \longsquiggly_\pi \Delta'$, there exists
		      $\Theta'$ such that $O[\Theta] \longsquiggly_\pi \Theta'$ and $\Delta'\;
			      \rel\;\Theta'$;
		\item whenever $O[\Theta] \longsquiggly_\pi \Theta'$, there exists $\Delta'$
		      such that $O[\Delta] \longsquiggly_\pi \Theta'$ and $\Delta'\;\rel\;\Theta'$.
	\end{itemize}
	Let \emph{constrained saturated bisimilarity} $\sim_{cs}$ be the largest constrained saturated
	bisimulation. We say that two processes $P, Q$ are constrained bisimilar if
	$\singleton{\conf{\rho, P, \nil}} \sim_{cs} \singleton{\conf{\rho, Q, \nil}}$ for each
	$\rho$.
\end{definition}

\begin{example}
	Consider $\Delta, P, Q$ from~\autoref{ex:broken-nondet} and take
	$O[\blank] = [\blank] \parallel c?x.P \parallel c?x.Q$ (notice that it emulates the context distinguishing $\Delta$ from $\Theta$ in~\autoref{ex:broken-nondet}).
	In the enhanced semantics, $O[\Delta]$ only performs
	$O[\Delta] \longsquiggly_{\ell} (\singleton{\conf{\dket{0}, \nil, P[\sfrac{q}{x}]}}
		\psum{1/2} \singleton{\conf{\dket{1}, \nil, P[\sfrac{q}{x}]}})\parallel c?x.Q
	$ or $O[\Delta] \longsquiggly_{r} (\singleton{\conf{\dket{0}, \nil, Q[\sfrac{q}{x}]}}
		\psum{1/2} \singleton{\conf{\dket{1}, \nil, Q[\sfrac{q}{x}]}}) \parallel c?x.P
	$, as the indices must coincide for any configuration in the support of $\Delta$.
	Indeed, we will later prove that $\Delta \sim_{cs} \Theta$.

\end{example}

\subsection{Behavioural Assessment of Constrained Bisimilarity}\label{behaviouralAssessment}

In the following, we state some important properties of $\sim_{cs}$.
A first result recovers the linearity of the relation.
We then prove that $\sim_{cs}$ is strictly coarser than the previously defined $\sim_s$.
Furthermore, we show that our constraints are strong enough to agree with the limitations of the observational power prescribed by quantum theory.
More precisely, we consider the equivalence classes over distributions of quantum states implicitly represented by density operators, and we generalize them
as classes of bisimilar distributions of lqCCS configurations.
An undue curbing of contexts is not better than a deficient one, 
therefore, we finally prove that the expressivity of the non-deterministic sum of processes is preserved 
when justified by classical information.
 
\begin{restatable}{theorem}{linearity}\label{thm:linearity}
	If $\Delta_i \sim_{cs} \Theta_i$ for $i = 1, 2$, then $\Delta_1 \psum{p} \Delta_2 \sim_{cs}  \Theta_1 \psum{p} \Theta_2$
	for any $p \in [0, 1]$.
\end{restatable}
\begin{proofsketch}
	It follows from the linearity of barbs and the decomposability of $\longsquiggly_\pi$, which hold by definition.
\end{proofsketch}

Regarding quantum properties, the enhanced semantics limits the capability of
the observer to perform different choices in different configurations of the
same distribution. 
As a result, cs-bisimilarity is strictly broader than s-bisimilarity, as observers are less powerful.

\begin{restatable}{theorem}{probsmallerthanconst}\label{thm:probsmallerthanconst}
	$\sim_{s}\ \subsetneq\ \sim_{cs}$.
\end{restatable}
\begin{proofsketch}
	\autoref{ex:zopm} shows that $\sim_{cs}\ \not\subseteq\ \sim_{s}$. 
	For $\sim_{s}\ \subseteq\ \sim_{cs}$ we provide a translation $\ptag{\blank}$ that annotates a given $R$ with fresh barbs encoding the non-deterministic choices.
	We prove by induction that the enhanced semantics of $\conf{\rho, P, R}$  
	corresponds to the standard semantics of $\conf{\rho, P || \ptag{R}}$, where $\ptag{R}$ is composed with the parallel operator (as required by $\sim_s$).
	In particular, $\conf{\rho, P, R} \longsquiggly_\pi \Delta$ if and only if $\ptag{R} \downarrow_\pi$ and $\conf{\rho, P || \ptag{R}} \rightarrow \Delta'$ with $\Delta'\!\downarrow_\pi^{0}$ 
	(roughly, a $\longsquiggly_\pi$ move corresponds to a $\rightarrow$ one where the barb $\pi$ is ``consumed'').
	This allows us to prove that $\sim_{s}$ is a cs-bisimulation.
\end{proofsketch}
From the proof above it follows that the enhanced semantics does not give cs-bisimilarity any additional discriminating power with respect to the standard semantics. In other words, we could have defined $\sim_{cs}$ without changing the semantics, just requiring instead that contexts express their non-deterministic choices as barbs. The enhanced semantics --- as well as the absence of congruence rules for parallel observers --- is then just a convenient way to assign a name to each possible non-deterministic choice, which fits well with the SOS-style rules.

We discuss now how cs-bisimilarity deals with the issue presented in~\autoref{nonDeterministicIssues}.
\begin{example}\label{ex:zopm}
	Let $\Delta = \singleton{\conf{\ketbra{+}, \meas[M_{01}]{q}{x}.c!q, \nil}}$, and $\Theta = \singleton{\conf{\ketbra{0}, \meas[M_{\pm}]{q}{x}.c!q, \nil}}$.
	For each $O[\blank]$, $O[\Delta]$ evolves in
	$O[\Delta']$ with $\Delta'= \singleton{\conf{\dket{0}, c!q, \nil}} \psum{1/2}  \singleton{\conf{\dket{1}, c!q, \nil}}$, and $O[\Theta]$ in
	$O[\Theta']$ with $\Theta'=  \singleton{\conf{\dket{+}, c!q, \nil}} \psum{1/2}  \singleton{\conf{\dket{-}, c!q, \nil}}$.
	As detailed in \autoref{nonDeterministicIssues}, one would expect $\Delta' \sim_{cs} \Theta'$, as they send indistinguishable quantum states.
	We prove that this is the case, in particular, because the process $c!q$ is deterministic.
\end{example}
A distribution $\Delta$ is \emph{deterministic} if its evolution is fully probabilistic, i.e.\ if it evolves in a single distribution up to bisimilarity. For example, distributions without parallel operators and non-deterministic sums are trivially deterministic.
\begin{definition}[Deterministic processes]
		A set of distributions $\mathcal{A}$ is deterministic if $\Delta \in \mathcal{A}$ implies that for any $O[\blank], \Delta', \Delta''$, if $O[\Delta] \longsquiggly_\pi \Delta'$ and $O[\Delta] \longsquiggly_\pi \Delta''$ then $\Delta' \sim_{cs} \Delta''$ and $\Delta', \Delta'' \in \mathcal{A}$.
		A distribution $\Delta$ is called deterministic if it is contained in a deterministic set.
		A process $P$ is deterministic if $\singleton{\conf{\rho, P, \nil}}$ is deterministic for any state $\rho$.
	\end{definition}

As previously stated, mixed states represented by the same density operator are indistinguishable.
We recover an analogous result for 
lqCCS.
Roughly, quantum states can be combined into a point distribution when paired with identical deterministic processes.

\begin{restatable}{theorem}{propertyA}\label{thm:propertyA}
	If $P$ is deterministic, then for any $\rho, \sigma, p, R$, $\singleton{\conf{\rho, P, R}} \psum{p} \singleton{\conf{\sigma, P, R}} \sim_{cs} \singleton{\conf{\rho \psum{p} \sigma, P, R}}$
	where $\rho \psum{p} \sigma$ is defined as the density operator $p\rho + (1-p)\sigma$.
\end{restatable}
\begin{proofsketch}
	For any deterministic $P$, we prove by structural induction that
	\[
	\rel = \{\,(\bot, \bot)\,\} \cup \left\{ \left(\singleton{\conf{\rho \psum{p} \sigma, P, R}} \ , \ \singleton{\conf{\rho, P, R}} \psum{p} \singleton{\conf{\sigma, P, R}}\right) \mid \rho, \sigma, p, R \ \right\}
	\]
	is a bisimulation up to convex hull~\cite{bonchi_power_2017} and up to bisimilarity (the soundness of which is given by~\autoref{thm:uptocssound} and according to~\citet{sangiorgi_enhancements_2011}).
	
	The result mainly follows from the fact that the classical components of the distributions are identical, while the quantum components are indistinguishable, as superoperators and measurements are convex, i.e. $\mathcal{F}(\rho) \psum{p} \mathcal{F}(\sigma) = \mathcal{F}(\rho \psum{p} \sigma)$ for any superoperator or measurement $\mathcal{F}$. 
	The hypothesis of determinism is required for the cases of non-deterministic sums and parallel compositions.
	Indeed, $\Delta = \singleton{\conf{\rho, P \parallel Q, \nil}} \psum{p} \singleton{\conf{\sigma, P \parallel Q, \nil}}$ may evolve differently with the left and right component, e.g.\ choosing $P$ in the left and $Q$ in the right.
	The hypothesis of the process being deterministic ensures that no information is leaked about $\rho$ and $\sigma$, 
	and ensures that the choice of $\Delta$ is irrelevant, allowing $\singleton{\conf{\rho \psum{p} \sigma, P, \nil}}$ to replicate its move.
	This is not possible in general, as shown in \autoref{ex:non-detproc}.
\end{proofsketch}

	Remarkably, transitivity of $\sim_{cs}$ suffices for proving the bisimilarity of deterministic processes paired with distributions of quantum states that are represented by the same density operator.
Note for example that $\sconf{\ketbra{0}, c!q}\psum{{1}/{2}}\sconf{\ketbra{1}, c!q}$ and $\sconf{\ketbra{+}, c!q}\psum{{1}/{2}}\sconf{\ketbra{-}, c!q}$ of~\autoref{ex:zopm} and~\ref{ex:broken-nondet}
 are both bisimilar to $\sconf{\frac{1}{2}{I}, c!q}$.
More in general, the equivalence classes represented by density operators are lifted to lqCCS distributions, yielding the equivalence $\rel \subseteq\ \sim_{cs}$ 
relating $\Delta$ and $\Theta$ deterministic whenever
 $\sum_{\rho} \Delta\big(\conf{\rho, P, R}\big) \rho = \sum_{\rho} \Theta\big(\conf{\rho, P, R}\big) \rho$ for all $P, R$.
Note also that, thanks to the linearity of $\sim_{cs}$, the property above is not limited to syntactically identical processes.

A consequence of~\autoref{thm:propertyA} is that $\sim_{cs}$ is not decomposable, similarly to other proposals addressing the limitations of probabilistic bisimilarity~\cite{feng_toward_2015-1,deng_bisimulations_2018}.
\begin{corollary}
	cs-bisimilarity is not a decomposable relation.
\end{corollary}
\begin{proofsketch}
	Take $\Delta = \sconf{\ketbra{0}, c!q}\psum{{1}/{2}}\sconf{\ketbra{1}, c!q}$ and  
	$\Theta = \sconf{\ketbra{+}, c!q}\psum{{1}/{2}}\sconf{\ketbra{-}, c!q}$.
	Notice that $\Delta \sim_{cs} \Theta$ by~\autoref{thm:propertyA}.
	Then, for $\sim_{cs}$ to be decomposable, $\Theta$ should be equal to $\Theta_1\psum{{1}/{2}}\Theta_2$ for some $\Theta_1 \sim_{cs} \sconf{\ketbra{0}, c!q}$.
	But since $\Theta_1$ can only be either $\sconf{\ketbra{+}, c!q}$, $\sconf{\ketbra{-}, c!q}$ or a combination of them, $\Theta_1 \not\sim_{cs} \sconf{\ketbra{0}, c!q}$ as they send observably different quantum values.
\end{proofsketch}

\autoref{thm:propertyA} targets deterministic processes because they represent fully defined physical processes, e.g., where all the choices are performed by boolean conditions.
Indeed, there is no reason to expect the property to hold for processes expressing non-determinism \`a la CCS, as it 
does not encode any physical behaviour considered in quantum theory.
More in detail, an extension of this theorem for general processes can only hold with overly constrained non-deterministic sums, and it would contradict~\autoref{thm:nondetVSite} below, attesting to the preservation of the expressiveness of non-determinism in processes.
We show an example of a non-deterministic process, derived from the sixth row of~\autoref{tab:examples}, for which~\autoref{thm:propertyA} does not apply.
\begin{example}\label{ex:non-detproc}
	Consider the following pair of processes of the last line of~\autoref{tab:examples}
	\begin{align*}
		\unitary[\scriptscriptstyle\dket{+}]{Set}{q}. \meas[M_{01}]{q}{x}. (c!q + d!q) \text{ and }
		\unitary[\scriptscriptstyle\dket{0}]{Set}{q}. \meas[M_{\pm}]{q}{x}. (c!q + d!q)
	\end{align*}
	We will later show that the distributions $\Delta = \singleton{\conf{\dket{+}, \meas[M_{01}]{q}{x}. \ite{x=0}{P}{Q}, \nil}}$ and $\Theta = \singleton{\conf{\dket{0}, \meas[M_{\pm}]{q}{x}. \ite{x=0}{P}{Q}, \nil}}$ to which the two processes reduce are not bisimilar.
\end{example}

As discussed previously, this is expected as we want to restrict the non-determinism of observers only.
Non-deterministic sum is typically used in processes to model unspecified behaviour, to be instantiated in future refinements.
Thus, we do not want to constrain non-determinism to the point that $+$ cannot replicate the behaviour of its refinements like boolean conditions. 

We say that $P'$ refines $P$, if $P'$ can be obtained from $P$ by substituting some occurrences of $Q + Q'$ with either $Q$, $Q'$, or $\ite{e}{Q}{Q'}$ for an arbitrary $e$.

\begin{figure}
  \small
	\begin{gather*}
	\infer[\square \in \{\parallel, +\}]{P'\ \square\ Q' \preceq P\ \square\ Q}{P' \preceq P & Q' \preceq Q} \qquad
	\infer[\mu \in \{ \tau, c?x, \sop{E}{\tilde{e}}, \meas{\tilde{e}}{x} \}]{\mu.P' \preceq \mu.P}{P' \preceq P} \qquad
	\infer{P' \preceq P + Q}{P' \preceq P} \qquad
	\infer{Q' \preceq P + Q}{Q' \preceq Q} \\
	\infer{P' \setminus c \preceq P \setminus c}{P' \preceq P} \qquad
	\infer{\ite{e}{P'}{Q'} \preceq P + Q}{P' \preceq P & Q' \preceq Q} \qquad
	\infer{\ite{e}{P'}{Q'} \preceq \ite{e}{P}{Q}}{P' \preceq P & Q' \preceq Q}
	\end{gather*}
	\caption{Refinement relation over lqCCS processes.}
	\label{fig:refinement}
\end{figure}

\begin{definition}
	The refinement relation $P' \preceq P$ is the smallest reflexive relation satisfying the rules in~\autoref{fig:refinement}.
	We say that $P'$ refines $P$, and that a configuration $\conf{\rho, P'}$ refines $\conf{\rho, P}$, if $P' \preceq P$.
	We let $\bot$ refine all the configurations, and define distribution refinement by linearity.
\end{definition}

A process is expected to be capable of matching any move of its refinements, thus, when considering $P = \meas[M_{01}]{q}{x}. (Q + Q')$, the moves of all $P' \preceq P$ should be available for $P$, included the ones of $\meas[M_{01}]{q}{x}. \ite{x = 0}{Q}{Q'}$ where the choice between $Q$ and $Q'$ depends on the outcome of the measurement.

We prove in the following that our constraints on non-determinism are not too restrictive, namely, that a distribution can simulate all its refinements.
\begin{restatable}{theorem}{nondetVSite}\label{thm:nondetVSite}
	Let $\Delta' \preceq \Delta$. If $\Delta' \longsquiggly_\pi \Theta'$ then $\Delta \longsquiggly_\pi \Theta$ for some $\Theta$ such that $\Theta' \preceq \Theta$.
\end{restatable}
\begin{proofsketch}
We prove by rule induction that 
	whenever $P' \preceq P$ and $\conf{\rho, P', R} \longsquiggly_\pi \Delta'$ then $\conf{\rho, P, R} \longsquiggly_\pi \Delta$ for some $\Delta$ such that $\Delta' \preceq \Delta$.
	The proof for \textsc{Congr} relies on the fact that refinement and structural congruence works well together.
	In detail, we show that given $P' \preceq P$ with $P' \equiv Q'$ we can find some $Q$ such that $Q' \preceq Q$ and $P \equiv Q$.
	In the other cases it suffices to use the induction hypothesis. 
	The theorem then holds by decomposability of $\longsquiggly_{\pi}$.
\end{proofsketch}
As a result of this property, the distributions $\Delta$ and $\Theta$ of~\autoref{ex:non-detproc} are distinguishable.
\begin{example}\label{ex:non-detproc2}
	Consider $\Delta$ and $\Theta$ of~\autoref{ex:non-detproc}, and notice
    that $\Delta' \preceq \Delta$, where $\Delta'$ sends on $c$ if and only if the qubit is in $\kz$. 
    Formally $\Delta' = \singleton{\conf{\dket{+}, \meas[M_{01}]{q}{x}. \ite{x=0}{c!q}{d!q}, \nil}}$ (see~\autoref{fig:exrefinement}).
	By performing this choice, $\Delta'$ is implicitly communicating the outcome of the measurement to the observer (through a side-channel, we could say).
	Consider the context
	\begin{gather*}
    O[\blank] = [\blank] \parallel (c?x.\meas[M_{01}]{x}{y}.(\ite{y = 0}{z!0}{o!0} \parallel \nil_{x})) + (d?x.\tau.\nil_x)
	\end{gather*}
	and note that, after two steps, $O[\Delta']$ reduces to a distribution expressing barb $z$ but not $o$, which is impossible for $O[\Theta]$.
	As a result of~\autoref{thm:nondetVSite}, $\Delta$ can replicate this move of $\Delta'$, hence $\Delta \not\sim_{cs} \Theta$. 
\end{example}

Our constrained bisimilarity is the first one to verify both~\autoref{thm:propertyA} and~\ref{thm:nondetVSite}, 
while all the previously proposed ones either fail in equating indistinguishable quantum states, or overly constrain non-determinism
(see~\autoref{tab:examples} for an in-depth comparison).

\begin{figure}[!t]
  \resizebox{0.8\textwidth}{!}
	  {\begin{tikzpicture}
    \tikzstyle{longsquiggly} = [-{Stealth[round]},line width=1pt,line join=round,decorate,decoration={zigzag,segment length=4,amplitude=.9,pre=lineto,pre length=5pt,post=lineto,post length=5pt}]
    \tikzstyle{box} = [draw, rectangle, rounded corners,inner sep=6pt]
    \node (d) [box]{$\sconf{\ketbra{+},\meas[M_{01}]{q}{x}.\ite{x = 0}{c!q}{d!q}, \nil}$};
    \node (rd) [box,right=of d]{$\sconf{\ketbra{+},\meas[M_{01}]{q}{x}.(c!q + d!q), \nil}$};
    \node (t) [box,below=20pt of d]{$\sconf{\ketbra{0}, c!q, \nil} \psum{\frac{1}{2}} \sconf{\ketbra{1}, d!q, \nil}$};
    \node (rt) [box,below=20pt of rd]{$\sconf{\ketbra{0}, c!q, \nil} \psum{\frac{1}{2}} \sconf{\ketbra{1}, d!q, \nil}$};
    \node at ($(d.east)!0.5!(rd.west)$) {$\preceq$};
    \node at ($(t.east)!0.5!(rt.west)$) {$=$};
    \draw[longsquiggly] (d) -- (t) node [above right=17pt and 2pt, inner sep=0pt] {$\scriptstyle \diamond$};
    \draw[dash pattern=on 3.45pt off 1.06pt, longsquiggly] (rd) -- (rt) node [above right=17pt and 2pt, inner sep=0pt] {$\scriptstyle \diamond$};
  \end{tikzpicture}}
  \caption{
    The existence of the dashed arrow on the right is guaranteed by the solid one on the left by~\autoref{thm:nondetVSite}.
    }
  \label{fig:exrefinement}
\end{figure}

\subsection{Properties of Constrained Bisimilarity}
\label{propertiesOfBisimilarity}

We now investigate our cs-bisimilarity.
We first recover two defining properties of~\cite{deng_open_2012}, 
namely that $\sim_{cs}$ is closed for superoperator application on qubits not appearing in the processes,
and that the state of such qubits is required to match in bisimilar distributions.
Then we show that discarded qubits can be ``traced out'' from the quantum state without affecting bisimilarity.
Finally, we discuss up-to techniques for proving constrained bisimilarity.

We start by recovering trace-identity and closure over superoperator application.
\begin{restatable}{proposition}{thmchinese}\label{thm:qccs}
	Let $\singleton{\conf{\rho, P, \nil}} \sim_{cs} \singleton{\conf{\sigma, Q, \nil}}$,
	then
	\begin{enumerate}
		\item $\singleton{\conf{\sop[\tilde{q}]{E}{\rho}, P, \nil}} \sim_{cs} \singleton{\conf{\sop[\tilde{q}]{E}{\sigma}, Q, \nil}}$, for any $\mathcal{E}_{\tilde{q}}$ and $\tilde{q}$ not in $\tilde{\Sigma}_P$;
		\item $\tr_{\Sigma_P}(\rho) = \tr_{\Sigma_P}(\sigma)$.
	\end{enumerate}
\end{restatable}
\begin{proofsketch}
	For the first point, take any superoperator $\mathcal{E}_{\tilde{q}}$, then there is a context that performs such transformation, so a context-closed relation is necessarily superoperator-closed. 
	For the second one, we proceed by contradiction: if the reduced density operators of two configurations are different, then we can build a context which measures such qubits obtaining a different distribution of outcomes and thus distinguishing the configurations via fresh barbs. 
\end{proofsketch}
Note that this is a useful result for disproving bisimilarity, as e.g., distributions with different partial trace are immediately deemed distinguishable.

A result that helps instead in proving bisimilarity is that the state of discarded qubits can be ignored.
In order to prove this in general, we extend the partial trace as follows.
\begin{definition}
	The partial trace $\ptrace{q}{\conf{\rho, P, R}}$ over $\tilde{q}$ of a configuration 
	is defined as $\conf{\ptrace{q}{\rho}, P', R}$ if $P \equiv P' \parallel \nil_{\tilde{q}}$. 
	We let $\ptrace{q}{\bot} = \bot$, and define the partial trace of distributions by linearity.
\end{definition}
Intuitively, we remove the discarded qubits \emph{together with} the discard processes.
Note that such an operation is defined only on distributions of configurations that discard the same qubits.

\begin{restatable}{proposition}{bisimilaritydiscarded}\label{thm:discarded}
	Let $\Delta, \Theta$ be distributions such that $\ptrace{q}{\Delta}$ and $\ptrace{q}{\Theta}$ are well-defined for a given $\tilde{q}$. If $\ptrace{q}{\Delta} \sim_{cs} \ptrace{q}{\Theta}$ then $\Delta \sim_{cs} \Theta$.
\end{restatable}
\begin{proofsketch}
	We show, by induction on $\longsquiggly_\pi$, that the semantics of $\mathcal{C}$ and of $\ptrace{q}{\mathcal{C}}$ are equivalent, and since the partial trace does not affect barbs, the desired property follows.
\end{proofsketch}
Note that, even if this property is given for the process $\nil_{\tilde{q}}$, we can apply it to any ``discard-like'' process thanks to the linearity of $\sim_{cs}$, i.e.\ to any process bisimilar to $\nil_{\tilde q}$.

The capacity to ignore discarded qubit is useful in a lot of proofs, as for the example below.
\begin{example}
	The two distributions below are bisimilar
	\begin{gather*}
		\Delta = \singleton{\conf{\dket{\Phi^+}, \meas[M_{01}]{q_1}{x}.c!q_1 \parallel \nil_{q_2}, \nil}} \text{ and }
		\Theta = \singleton{\conf{\dket{\Phi^+}, \meas[M_\pm]{q_1}{x}.c!q_1 \parallel \nil_{q_2}, \nil}}
	\end{gather*}
	Taken any $O[\blank]$, they evolve in $O[\Delta']$ with
	$\Delta' = \Big(\singleton{\conf{\dket{00}, c!q_1, \nil}} \psum{{1}/{2}} \singleton{\conf{\dket{11}, c!q_1, \nil}}\Big)  \parallel \nil_{q_2}$ and $O[\Theta']$ with
	$\Theta' = \Big(\singleton{\conf{\dket{++}, c!q_1, \nil}}  \psum{{1}/{2}} \singleton{\conf{\dket{--}, c!q_1, \nil}}\Big) \parallel \nil_{q_2}$.
	
	Finally, note that $\tr_{q_2}(\Delta') \sim_{cs} \tr_{q_2}(\Theta')$ holds because $\tr_{q_2}(\Delta') = \singleton{\conf{\dket{0}, c!q_1, \nil}} \psum{{1}/{2}} \singleton{\conf{\dket{1}, c!q_1, \nil}}$ and $\tr_{q_2}(\Theta') = \singleton{\conf{\dket{+}, c!q_1, \nil}}  \psum{{1}/{2}} \singleton{\conf{\dket{-}, c!q_1, \nil}}$, and the two are equated by~\autoref{thm:propertyA}.
\end{example}

Finally, we report a general proof technique.
While proving bisimilarity of two distributions usually requires giving a bisimulation relating the two, 
up-to techniques allows using smaller relations in place of proper bisimulations.
Given a relation $\rel \subseteq \dist{S} \times \dist{S}$, its \emph{convex hull} $Cv(\rel)$ is
the least relation such that $(\sum_{i \in I} \distelem{p_i}{\Delta_i})\;Cv(\rel)\;(\sum_{i \in I} \distelem{p_i}{\Theta_i})$ whenever $\Delta_i\;\rel\;\Theta_i$ for all $i \in I$.
Bisimulations up to $Cv$~\cite{bonchi_power_2017} are then defined as follows.
	\begin{definition}[Bisimulation up to convex hull] \label{def:uptocv}
		A relation $\rel \subseteq \dist{\confbot} \times \dist{\confbot}$ is a
		\emph{cs-bisimulation up to $Cv$} if $\Delta\,\rel\,\Theta$ implies that
		for any context $O[\blank]$ it holds
		\begin{itemize}
			\item $\Delta\!\downarrow_{b}^p$ if and only if $\Theta\!\downarrow_{b}^p$;
			\item whenever $O[\Delta] \longsquiggly_\pi \Delta'$, there exists
			$\Theta'$ such that $O[\Theta] \longsquiggly_\pi \Theta'$ and $\Delta'\;
			Cv(\rel)\;\Theta'$;
			\item whenever $O[\Theta] \longsquiggly_\pi \Theta'$, there exists $\Delta'$
			such that $O[\Delta] \longsquiggly_\pi \Theta'$ and $\Delta'\;Cv(\rel)\;\Theta'$.
		\end{itemize}
	\end{definition}
Giving a bisimulation up to convex hull is a sound proof technique for $\sim_{cs}$.
\begin{restatable}{proposition}{uptocv}\label{thm:uptocssound}
	If $\Delta\;\rel\;\Theta$ with $\rel$ a bisimulation up to convex hull, then $\Delta \sim_{cs} \Theta$.
\end{restatable}
\begin{proofsketch}
	We need to show that $Cv$ is \emph{compatible}~\cite{sangiorgi_enhancements_2011}, 
	which follows from the linearity of barbs and from the fact that $\longsquiggly_\pi$ is decomposable.
\end{proofsketch}

\section{Constrained Bisimilarity at Work}\label{sec:real-world}

We discuss real-world protocols: quantum teleportation, superdense coding and quantum coin flipping. 
We show how lqCCS models them, and how $\sim_{cs}$ is used for proving their properties.

\subsection{Quantum Teleportation}
The objective of quantum teleportation~\cite{qteleportation} is to allow Alice to send quantum information to Bob %
 without a quantum channel.
Alice and Bob must have each a qubit of an entangled pair $\ket{\Phi^+}$. 
The protocol works as follows: Alice performs a fixed set of unitaries to the qubit to transfer and to their part of the entangled pair; 
then Alice measures the qubits and sends the classical outcome to Bob, which applies different unitaries to their own qubit according to the received information.
In the end, the qubit of Bob will be in the state of Alice's one, and the entangled pair is discarded.
Note that Alice is not required to know the state of the qubit to send.

Consider the following encoding of the protocol where we assume that Alice ($\proc{A}$) and Bob ($\proc{B}$) already share an entangled pair ($q_1,q_2$) (we write $(n)_2$ to stress that $n$ is in binary representation)
\begin{align*}
  \proc{A}    & = \unitary{CNOT}{q_0,q_1}.\unitary{H}{q_0}.\meas[M_{01}]{q_0,q_1}{x}.(m!x \parallel \nil_{q_0,q_1})                       \\
	\proc{B}    & = m?y.\ite{y = (00)_2}{\unitary{I}{q_2}.\mathit{out}!q_2 }
	{( \ite{y = (01)_2}{\unitary{X}{q_2}.\mathit{out}!q_2                                                                                                                   \\&\quad\;\;\,}
		{( \ite{y = (10)_2}{\unitary{Z}{q_2}.\mathit{out}!q_2 }
	{\unitary{ZX}{q_2}.\mathit{out}!q_2})}) }                                                                                                                               \\
	\proc{Tel}  & = (\proc{A} \parallel \proc{B}) \setminus m                                                                                                              
\end{align*}
We let $\Delta = \singleton{\conf{\dket{\psi\Phi^+}, \proc{Tel}, \nil}}$, with $\Theta = \singleton{\conf{\dket{\psi\Phi^+}, \unitary{SWAP}{q_0,q_2}.\tau.\tau.\tau.\tau.(\mathit{out}!q_2 \parallel \nil_{q_0,q_1}, \nil}}$ its specification, for $\kp = \alpha\kz + \beta\ko$, and sketch the proof for $\Delta \sim_{cs} \Theta$ below.

Since there is no qubit in $\dket{\psi\Phi^+}$ apart from the ones in $\Sigma_{\proc{Tel}}$, $R$ is in deadlock for any context $O[\blank] = [\blank] \parallel R$.
Thus, all subsequent transitions are of the kind $O[\Delta] \longsquiggly_\diamond O[\Delta']$ with $\Delta \longsquiggly_\diamond \Delta'$
until a send operation on an unrestricted channel is reached (and the same for $\Theta$).
For simplicity, we will therefore omit the contexts in the first steps
\begin{align*}
  \Delta \longsquiggly_\diamond^3&
  \sum\nolimits_{i = 0}^3 \distelem{\frac{1}{4}}{\sconf{\dket{i} \otimes \dket{\psi_i}, (m!i \parallel \nil_{q_0,q_1} \parallel \proc{B}) \setminus m, \nil}} \\
  \quad\longsquiggly_\diamond^2&\ \Delta' = \sum\nolimits_{i = 0}^3 \distelem{\frac{1}{4}}{\sconf{\dket{i} \otimes \dket{\psi_i}, (\nil_{q_0,q_1} \parallel \mathit{out}!q_2) \setminus m, R}} \\
  \Theta \longsquiggly_\diamond& \sconf{\dket{\Phi^+\psi}, \tau.\tau.\tau.\tau.(\mathit{out}!q_2 \parallel \nil_{q_0,q_1}), \nil} \longsquiggly_\diamond^4 \Theta' = \sconf{\dket{\Phi^+\psi}, \mathit{out}!q_2 \parallel \nil_{q_0,q_1}, R}
\end{align*}
where $\ket{\psi_{0}} = \kp$, $\ket{\psi_{1}} = \beta\kz + \alpha\ko$, $\ket{\psi_{2}} = \alpha\kz - \beta\ko$, $\ket{\psi_{3}} = \beta\kz - \alpha\ko$, 
and where, abusing notation, we use $\ket{0} = \ket{00}, \ket{1} = \ket{01}, \ket{2} = \ket{10}$ and $\ket{3} = \ket{11}$ when speaking of pairs of qubits.
All intermediate steps happen with the same label ($\diamond$) and no barb is expressed, as the channel $m$ is restricted.
Finally, since $out!q_2$ is a deterministic process, it is immediate to prove that $\tr_{q_0,q_1}(O[\Delta']) \sim_{cs} \conf{\dket{\psi},\mathit{out}!q_2,R} = \tr_{q_0,q_1}(O[\Theta'])$ by applying~\autoref{thm:propertyA}.
The bisimilarity of $\Delta'$ and $\Theta'$ then follows from \autoref{thm:discarded}, therefore $\Delta \sim_{cs} \Theta$.

\subsection{Superdense Coding}

We consider a generalization of the superdense coding protocol~\cite{SDC}.
Assume Alice and Bob have each a qubit of a Bell pair $\ket{\Psi^+}$. The protocol allows Alice to communicate a distribution of two-bit integers to Bob by sending their single qubit to Bob. 

The protocol works as follows: Alice chooses a distribution of integers in $[0,3]$ and encode it by performing suitable transformations to their qubit, which is then sent to Bob;
Bob receives the qubit and decodes the distribution by performing CNOT and $\text{H} \otimes \text{I}$ on the pair of qubits (the received qubit and their original one), and then a measurement on the standard basis.
By measuring the qubits, Bob recovers the distribution chosen by Alice, and can use it as they like.

We consider the following instantiation of the protocol, where Bob uses the received value to decide (in an unspecified way) on which channel to send the received qubit (either channel $a$ or $b$)
\begin{align*}
	\proc{A} & = \sop{E}{q_0}.c!q_0                                                        \\
  \proc{B}   & = c?x. \unitary{CNOT}{x, q_1}.\unitary{H}{x}.\meas{x,q_1}{y}.((a!x + b!x) \parallel \nil_{q_1}) \\
	\proc{SDC}   & = \proc{A} \parallel \proc{B} \setminus c
\end{align*}
More in detail, Alice encodes: the point distribution $\bar{0}$ by applying the unitary I, $\bar{1}$ with $X$,
$\bar{2}$ with Z, and
$\bar{3}$ with ZX\@.
In general, they apply a superoperator $\mathcal{E}$ with Kraus decomposition
	\[
    \{ \sqrt{p_0} I, \sqrt{p_1} X, \sqrt{p_2} Z, \sqrt{p_3} ZX \}\quad\text{for some $p_i \in [0,1]$ such that} \sum\nolimits_{i} p_i = 1
	\]

Consider now the following where Rob (in place of Bob) forgets to measure the qubits
\begin{align*}
  \proc{R} = c?x. \unitary{CNOT}{x, q_1}.\unitary{H}{x}.\tau.((a!x + b!x) \parallel \nil_{q_1})
\end{align*}
Ideally, Rob cannot base their decision on the value sent by Alice, hence we expect $\proc{SDC}$ to be distinguishable from the case where $\proc{R}$ is substituted for $Bob$.
Indeed,
$\singleton{\conf{\ketbra{\Psi^+}, \proc{A} \parallel \proc{B} \setminus c}}$ and $\singleton{\conf{\ketbra{\Psi^+}, \proc{A} \parallel \proc{R} \setminus c}}$ are not constrained bisimilar in general.
Consider, for example, $\mathcal{E}$ with Kraus decomposition $\{ \frac{1}{2} \text{I}, \frac{1}{2} \text{X}, \frac{1}{2} \text{Z}, \frac{1}{2} \text{ZX} \}$ (i.e. Alice encodes a fair distribution of all the possible values).

Assume we always take the empty observer $O[\blank] = [\blank]$ whenever we do not specify otherwise, and note that
the following steps are forced
\begin{align*}
	 \singleton{\conf{\ketbra{\Psi^+}, \proc{SDC}}}
	 &\longsquiggly_\diamond \singleton{\conf{\rho, c!q_0 \parallel \proc{B} \setminus c}} \quad \text{, with } \rho = \frac{1}{4}\left(\ketbra{\Phi^+} + \ketbra{\Phi^-} + \ketbra{\Psi^+} + \ketbra{\Psi^-}\right) \\
   & \longsquiggly_\diamond \singleton{\conf{\rho, \unitary{CNOT}{q_0, q_1}.\unitary{H}{q_0}.\meas{q_0,q_1}{y}.((a!q_0 + b!q_0) \parallel \nil_{q_1}) \setminus c}}                                                     \\
	 & \longsquiggly_\diamond^2 \singleton{\conf{\frac{1}{4}I,
			\meas{q_0,q_1}{y}.((a!q_0 + b!q_0) \parallel \nil_{q_1}) \setminus c}}
	\longsquiggly_\diamond \Delta_{\proc{B}}
\end{align*}
where $\Delta_{\proc{B}}$ is defined as
\begin{align*}
	 & \quad \sum\nolimits_{j = 0}^3 \distelem{\frac{1}{4}}{\singleton{\conf{\ketbra{j}, (a!q_0 + b!q_0) \parallel \nil_{q_1} }}}
\end{align*}

Similarly, for Rob we have that the following are forced
\begin{align*}
   &\singleton{\conf{\ketbra{\Psi^+}, \proc{A} \parallel \proc{R} \setminus c}} \longsquiggly_\diamond \singleton{\conf{\rho, c!q_0 \parallel \proc{R} \setminus c}} \\
   & \quad\longsquiggly_\diamond \singleton{\conf{\rho, \unitary{CNOT}{q_0, q_1}.\unitary{H}{q_0}.\tau.((a!q_0 + b!q_0) \parallel \nil_{q_1}) \setminus c}}           \\
	 & \quad\longsquiggly_\diamond^2 \singleton{\conf{\frac{1}{4}I,
	\tau.((a!q_0 + b!q_0) \parallel \nil_{q_1}) \setminus c}}
	 \longsquiggly_\diamond \Delta_{\proc{R}} =
	\singleton{\conf{\frac{1}{4}I,(a!q_0 + b!q_0) \parallel \nil_{q_1}}}
\end{align*}

Take now the following context
\[
	O'[\blank] = [\blank] \parallel a?z.\meas{z}{res}(\ite{res = 0}{success!z}{\mathit{fail}!z} + b?z.\tau.\nil_{z})
\]
and assume $O'[\Delta_{\proc{B}}]$ evolves as
\begin{align*}
	 & \quad \frac{1}{4} \bullet \singleton{\conf{\ketbra{00}, \nil_{q_1}, \nil \parallel \meas{q_0}{res}(\ite{res = 0}{success!z}{\mathit{fail}!z}}} \\
	 & +\frac{1}{4} \bullet \singleton{\conf{\ketbra{01}, \nil_{q_1}, \nil \parallel \meas{q_0}{res}(\ite{res = 0}{success!z}{\mathit{fail}!z}}}      \\
	 & +\frac{1}{4} \bullet \singleton{\conf{\ketbra{10}, \nil_{q_1}, \nil \parallel \tau.\nil_{q_0}}}
	+\frac{1}{4} \bullet \singleton{\conf{\ketbra{11}, \nil_{q_1}, \nil \parallel \tau.\nil_{q_0}}}
\end{align*}
After a reduction, $O'[\Delta_{\proc{B}}]$ expresses the barbs \emph{success} and \emph{fail} with probability $1/2$ and $0$ respectively.

On the contrary, $O'[\Delta_{\proc{R}}]$ may only evolve as either
\begin{align*}
  \singleton{\conf{\frac{1}{4}I, \nil_{q_1}, \nil \parallel \meas{q_0}{\text{res}}(\ite{\text{res} = 0}{\mathit{success}!z}{\mathit{fail}!z}}}\quad\text{ or }\quad
	\singleton{\conf{\frac{1}{4}I, \nil_{q_1}, \nil \parallel \tau.\nil_{q_0}}},
\end{align*}
and, after a reduction, must express both the barbs $success$ and $\mathit{fail}$ with either probability $1/2$ or $0$.

This important result is due to the~\autoref{thm:nondetVSite}, which allows $\proc{B}$ to behave as if any boolean conditional was in place of $+$, thus to send the qubit on $a$ if the outcome of the measurement is strictly lower than $2$, and on $b$ otherwise.
The other proposed behavioural equivalences for addressing the problem of the standard probabilistic bisimilarity deem the two distributions indistinguishable (see~\autoref{sec:rw} for an in depth comparison).

\subsection{Quantum Coin Flipping}
We present now a more complex example, namely the Quantum Coin Flipping (QCF) protocol. 
Suppose Alice and Bob do not trust each other, and want to randomly select a winner between them.
\citet{bb84} propose a protocol in which Alice chooses either the $01$ or $\pm$ basis at random, then generates a sequence of random bits and encodes them into a sequence of qubits in the selected basis (the first element of the basis stands for bit $0$, the second for $1$).
The qubits are then sent to Bob, who 
measures each of them in a random basis ($01$ or $\pm$). Finally, Bob tries to guess the basis chosen by Alice: Bob wins if the guess is correct. 

Bob has no way to find Alice's basis from the received qubits, so the guess will be correct or wrong with equal probability. 
Alice could cheat, lying about their basis. 
To protect Bob, at the end of the protocol, Alice must reveal their basis and the original bit sequence.
Bob then compares the original sequence with the previously stored outcomes of the measurements.
For qubits where Alice's basis coincides with Bob's one, the outcomes must 
coincide with the original bit sequence.

Hereafter, we let $x_i$ be the $i$-th bit of the integer $x$, and resort to some minor extensions of lqCCS\@
\begin{itemize}
\item $\conf{\rho, RandBit(x).P}$ evolves as $\overline{\conf{\rho, P[0/x]}} \psum{{1}/{2}} \overline{\conf{\rho, P[1/x]}}$, and could be implemented with an additional qubit.
\item We assume a mapping $\beta$ from bit to bases with $\beta(0)$ and $\beta(1)$ the $01$ and $\pm$ basis. 
\item We assume a polyadic extension of lqCCS type system and semantics where $c!\tilde{v}$ and $c?\tilde{x}$ allow substituting tuples of values $\tilde{v} = v_1, \ldots ,v_n$ for variables $\tilde{x} = x_1, \ldots ,x_n$.
\end{itemize} 

We formalize QCF as follows, where $n$ is the number of qubits, and the outcome is sent on $a$ and $b$ ($1$ if Bob wins and $0$ otherwise).
We use $b =_{int} b'$ for comparing digits, defined as $(1 - b)(1 - b') + bb'$.
\begin{align*}
  \proc{Alice}       & = Rand(\text{secretvalue}).\proc{Alice}_{\beta(\text{secretvalue})}                                                               \\
  \proc{Alice}_{01}  & = \unitary{H}{\tilde{q}}.M_{01}(\tilde{q} \rhd w).(
	\mathit{AtoB}!\tilde{q} \parallel
	\mathit{guess}?g.(a!(g =_{int} 0)
	 \parallel \mathit{secret}!0 \parallel \mathit{witness}!w
	))                                                                                                       \\
	\proc{Alice}_{\pm} & = I(\tilde{q}).M_\pm(\tilde{q} \rhd w).(
	\mathit{AtoB}!\tilde{q} \parallel
	\mathit{guess}?g.(a!(g =_{int} 1)
	\parallel \mathit{secret}!1 \parallel \mathit{witness}!w
	))                                                                                                       \\
	\proc{Bob}         & = \mathit{AtoB}?\tilde{z}.\Big(
	\Big(\Big(\bigparallel\nolimits_{i = 1}^n \proc{Server}_i \Big) \parallel \proc{Bob}' \Big) \setminus \{\mathit{base}_i\}_{i = 1}^n \setminus \{\text{bit}_i\}_{i = 1}^n
	\Big)                                                                                                       \\
	\proc{Bob}'        & = \mathit{base}_1?b_1\ldots \mathit{base}_n?b_n.
	\mathit{bit}_1?x_1\ldots\mathit{bit}_n?x_n.
	Rand(g).(\mathit{guess}!g \parallel \proc{Bob}'')                                                                  \\
	\proc{Bob}''       & =  \mathit{secret}?g'.\mathit{witness}?w.\Big(
	b!(g =_{int} g') \parallel
	\Big(\bigparallel\nolimits_{i = 1}^n
	\ite{(b_i = g' \land x_i \neq w_i}{\mathit{cheat}!0}{\nil}
	\Big)\Big)                                                                                                       \\
	\proc{Server}_i    & = Rand(b).M_{\beta(b)}(z_i \rhd x).(\mathit{base}_i!b \parallel \mathit{bit}_i!x \parallel \nil_{z_i}) \\
	\proc{QCF} 				& = (\proc{Alice} \parallel \proc{Bob}) \setminus \{\mathit{AtoB, guess, secret, witness}\} 									
\end{align*}
Thanks to cs-bisimilarity, we can analyse three properties of QCF, namely that the outcome is fair, that Bob cannot cheat, and neither Alice can. As reported by~\citet{bb84}, the first two properties hold, but an attack exists allowing Alice to decide the outcome of the protocol.
\paragraph{Fairness}
We show that $\conf{\ketbra{0^n}, \proc{QCF}} \sim_{cs} \conf{\ketbra{0^n}, \proc{FairCoin}}$, with $\proc{FairCoin}$ defined as $\tau^{4n + 5}.Rand(x).(a!x \parallel \tau.\tau.b!x \parallel \nil_{\tilde{q}})$.
As before, it suffices to consider the empty context, and we show the evolution of the protocol for $n = 1$, as the other cases follow the same pattern.
\[
\conf{\ketbra{0}, \proc{QCF}}
\longsquiggly_\diamond \left(\overline{\conf{\ketbra{0}, \proc{Alice}_{01}\parallel \proc{Bob}}} \psum{\frac{1}{2}}
\overline{\conf{\ketbra{0}, \proc{Alice}_{\pm}\parallel \proc{Bob}}}\right) \setminus C
\]
where $C = \{\mathit{AtoB, guess, secret, witness}\}$. We will focus just on the execution of $Alice_{01}$, as the other one is symmetrical. The first actions of Alice are to prepare the qubit at random and send it to Bob. The latter will then measure it in a random basis and record the result. Formally
\begin{align*}
& \conf{\ketbra{0}, \proc{Alice}_{01}\parallel \proc{Bob}}\\
& \longsquiggly_\diamond^3
\sum\nolimits_{j \in \{0, 1\}} \distelem{\frac{1}{2}}{\overline{\conf{\ketbra{j}, \proc{Alice}_{01}'[j/w] \parallel  \proc{Bob}' \parallel \proc{Server}}}}		\\
& \longsquiggly_\diamond^2
\quad \distelem{\frac{1}{2}}{\left(\sum\nolimits_{j \in \{0, 1\}} \distelem{\frac{1}{2}}{\overline{\conf{\ketbra{j}, \proc{Alice}_{01}'[j/w] \parallel \proc{Bob}' \parallel (\mathit{base}!0 \parallel \mathit{bit}!j \parallel \nil_{q}) }}}\right)} \\
& \qquad \ +\distelem{\frac{1}{2}}{\left(\sum\nolimits_{j \in \{0, 1\}} \distelem{\frac{1}{2}}{\overline{\conf{H\ketbra{j}H, \proc{Alice}_{01}'[j/w] \parallel \proc{Bob}' \parallel (\mathit{base}!1 \parallel \mathit{bit}!j \parallel \nil_{q}) }}}\right)} 
\end{align*}																							
where $\proc{Alice}_{01}' = \mathit{guess}?g.(a!(g =_{int} 0) \parallel \mathit{secret}!0 \parallel \mathit{witness}!w)$.
After the measurement (and after synchronising with the server processes) Bob send their random guess of the secret basis to Alice, who reveals the correct one.
Bob checks the consistency of Alice response, and if the protocol is executed correctly the two parties will agree on the outcome: either $0$ or $1$ with equal probability.
\begin{align*}
&
\quad \distelem{\frac{1}{2}}{\left(\sum\nolimits_{j, k \in \{0, 1\}}\distelem{\frac{1}{4}}{\overline{\conf{\ketbra{j},\nil_{q} \parallel \proc{Alice}_{01}'[j/w] \parallel \mathit{guess}!k \parallel \proc{Bob}''[0/b][j/x][k/g] }}}\right)} \\
& +\distelem{\frac{1}{2}}{\left(\sum\nolimits_{j, k \in \{0, 1\}}\distelem{\frac{1}{4}}{\overline{\conf{H\ketbra{j}H,\nil_{q} \parallel \proc{Alice}_{01}'[j/w] \parallel \mathit{guess}!k \parallel \proc{Bob}''[1/b][j/x][k/g] }}}\right)} \\
& \longsquiggly_\diamond^3
{\left(\sum\nolimits_{j, k \in \{0, 1\}} \distelem{\frac{1}{4}}{\overline{\conf{\ketbra{j},\nil_{q} \parallel a!k \parallel b!k }}}\right)} 
\ \psum{\frac{1}{2}} \ {\left(\sum\nolimits_{j, k \in \{0, 1\}} \distelem{\frac{1}{4}}{\overline{\conf{H\ketbra{j}H,\nil_{q} \parallel a!k \parallel b!k }}}\right)} 
\end{align*}
 It is easy to see that in this last step $\proc{QCF}$ expresses the barbs $\downarrow a$ and $\downarrow b$, sending the same values as the specification $\proc{FairCoin}$, and the two are indeed bisimilar.

\paragraph{Dishonest Bob}
In order to cheat, Bob needs to discover
Alice's secret from the sent qubits alone. 
This is impossible, because \autoref{thm:propertyA} deems the initial prefixes of $\proc{Alice}_{01}$ and $\proc{Alice}_\pm$ bisimilar.
\[
	\proc{A}_{01} = \sum\nolimits_{j = 0}^{2^n-1} \distelem{\frac{1}{2^n}}{\conf{\ketbra{j}, AtoB!\tilde{q}}} \quad \sim_{cs}  \quad \sum\nolimits_{j = 0}^{2^n-1} \distelem{\frac{1}{2^n}}{\conf{H^{\otimes n}\ketbra{j}H^{\otimes n}, AtoB!\tilde{q}}} = \proc{A}_\pm
\]
Note that $A_{01} \not\sim_s A_{\pm}$, as shown in~\autoref{ex:broken-nondet}.
	Traditional probabilistic bisimilarity à la~\citet{hennessy_exploring_2012} 
	fails in analysing Quantum Coin Flipping and similar protocols.

\paragraph{Dishonest Alice}

Interestingly, Alice can cheat by using additional qubits entangled with the ones they send to Bob. 
By measuring their entangled qubits in Bob's chosen basis, Alice forges a fake witness for deceiving Bob (in the process, Alice wins and $0$ is sent on $a$ and $b$). 
We call this attacker 
Alison, and show that $\conf{\ketbra{0^{2n}}, (\proc{Alison} \parallel \proc{Bob}) \setminus C} \sim_{cs} \conf{\ketbra{0^{2n}}, \proc{UnfairCoin}}$.
\begin{align*}
	\proc{Alison}  & = Set_{\Phi^+}(q_1, q_1')\ldots Set_{\Phi^+}(q_n, q_n').\left(
	\mathit{AtoB}!\tilde{q} \parallel \proc{Alison}'
	\right)                                                                                                                                                 \\
	\proc{Alison}' & = \mathit{guess}?g.(a!0 \parallel \mathit{secret}!(1 - g) \parallel M_{\beta(1-g)}(q' \rhd x').(\mathit{witness}!x' \parallel \nil_{\tilde{q'}})) \\
	\proc{UnfairCoin}     & = \tau^{5n + 3}.(a!0 \parallel \tau.\tau.\tau.b!0 \parallel \nil_{\tilde{q}} \parallel \nil_{\tilde{q'}})
\end{align*}
The $\proc{UnfairCoin}$ specification always selects Alison as the winner, and never expresses the $\downarrow_{cheat}$ barb. In other words, $(\proc{Alison} \parallel \proc{Bob}) \setminus C \sim_{cs} \proc{UnfairCoin}$ means that Alison is always capable of tricking Bob without being discovered. 

The behaviour of Bob is identical to the previous case.
Indeed, the reduced density operator of the qubits sent by Alison is indistinguishable from the one of the honest Alice. 
But after receiving Bob's guess, Alison can measure their own qubits, which have decayed as the ones of Bob. 
In this way, her fake witness $w'$ will always be correct, as we show for the case $n = 1$
\begin{align*}
  & \conf{\ketbra{00}, (\proc{Alison} \parallel \proc{Bob}) \setminus C}  \\
& \longsquiggly_\diamond^{11}
{\left(\sum\nolimits_{j \in \{0, 1\}} \distelem{\frac{1}{2}}{\overline{\conf{\ketbra{jj},\nil_{q} \parallel a!0 \parallel b!0 }}}\right)} 
\ \psum{\frac{1}{2}} \ {\left(\sum\nolimits_{j \in \{+, -\}} \distelem{\frac{1}{2}}{\overline{\conf{\ketbra{jj},\nil_{q} \parallel a!0 \parallel b!0 }}}\right)}
\end{align*}

\section{Related Works}\label{sec:rw}

We focus on the quantum process calculi most similar to our proposal,
as well as likely the better established and developed, namely QPAlg~\cite{lalire_process_2004,lalire_relations_2006}, CQP~\cite{gay_communicating_2005,gay_types_2006,davidson_formal_2012}, and qCCS~\cite{feng_probabilistic_2007,ying_algebra_2009,feng_bisimulation_2012,feng_symbolic_2014,deng_open_2012,feng_toward_2015-1,deng_bisimulations_2018}\@.
When comparing the proposed behavioural equivalences, we abstract from the ``classical'' details 
and focus on the quantum-related features,
restricting ourselves to the strong version of the bisimulations.
A first difference with lqCCS is that they are mostly based on labelled bisimilarities.
\autoref{tab:examples} summarizes distinctive prototypical processes (in the lqCCS syntax) deemed bisimilar or not according to different approaches.

One of the discrepancies is the visibility of qubits that are neither sent nor discarded.
The first three lines contain processes whose bisimilarity depends on the assumption about the visibility of these unsent qubits.
Our linear type system makes the former assumptions irrelevant.
In the fourth line, the processes send pairs of qubits with the same partial trace if taken separately, even if a pair is entangled and the other is not.
The fifth line compares processes where the state of the sent qubits is represented by the same density operator,
whose bisimilarity is implied by~\autoref{thm:propertyA}.
Finally, the sixth line compares two processes where a qubit is sent non-deterministically over two channels:
for each channel, the state of the sent qubits is represented by the same density operator, but if the chosen channel depends on the outcome of the measurement the two processes can be distinguished.
	Bisimilarities that do not distinguish these two processes cannot satisfy~\autoref{thm:nondetVSite}.

\subsection*{QPAlg}
The Quantum Process Algebra (QPAlg) is an extension of synchronous value-passing
CCS with primitives for unitary transformations and measurements.
As common, quantum operations are silent, and quantum communication updates quantum variables.
They propose a probabilistic branching
bisimilarity, adapted for stateful computations by requiring
bisimilar processes to send the same quantum state, defined as the partial trace of the global quantum
state.
However, this bisimilarity is coarser than prescribed by the theory when states are entangled~\cite{davidson_formal_2012}.
\begin{example}\label{ex:qpalg}
	The processes in the fourth line of~\autoref{tab:examples} are bisimilar in QPAlg, as they send pairs of qubits with the same partial trace
	($tr_{q_0}(\Phi^+) = tr_{q_1}(\Phi^+) = tr_{q_0}(\frac{1}{4}I) = tr_{q_1}(\frac{1}{4}I) = \frac{1}{2}I$).
	Such processes are not bisimilar according to $\sim_{cs}$, as the context of~\autoref{ex:entanglement} discriminates the two.
\end{example}
Finally, behaviourally equivalent processes are not required to behave similarly on unsent qubits,
e.g.\ $H(q).\nil$ and $X(q).\nil$ of the first line of~\autoref{tab:examples}.
In lqCCS, the above processes are not legal, but a similar result holds for $H(q).\nil_{q}$ and $X(q).\nil_{q}$.

\subsection*{CQP}

The calculus of Communicating Quantum Processes (CQP) is inspired by the $\pi$-calculus and it is enriched with qubits declarations (that extend the quantum state) and quantum transformations,
but without guards or match operators, thus lacking the form of classical control used in lqCCS\@.
They introduce an affine type system prescribing that every qubit is sent at most once.
CQP comes with a reduction semantics~\cite{gay_communicating_2005,gay_types_2006} and a labelled one~\cite{davidson_formal_2012}, both based on pure quantum states.
Davidson proposes mixed configurations, i.e.\ configurations where a single process is paired with a probability distribution of classical and quantum states.
Mixed configurations represent non-observable probabilities due to measurements whose result is not communicated yet, and they are treated as a single state by the semantics.
Our proposal generalizes CQP mixed configurations to distributions of configurations with possibly different processes.
This is necessary for extending the approach to processes with boolean guards.

A branching bisimilarity is defined for the labelled semantics of CQP\@.
To avoid the problem of QPAlg with entangled states, bisimilarity requires that the reduced state obtained by collecting \emph{all} the sent qubits coincide.
Moreover, mixed configurations allow relating processes sending indistinguishable quantum states.
The resulting bisimilarity is a congruence for parallel
composition and is capable of equating interesting cases such as the one of the fifth line of~\autoref{tab:examples}.
\begin{example}
  Consider $\conf{\dket{+}, \meas[M_{01}]{q}{x}.c!q}$ and $\conf{\dket{0}, \meas[M_{\pm}]{q}{x}.c!q}$.
	Even though they end up in different quantum states, the two configurations above are bisimilar according to CQP, because they send on channel $c$ qubits with the same partial trace.
	Our proposed bisimilarity replicates this result by resorting to contexts with constrained non-determinism (see~\autoref{ex:zopm}). 
	Indeed, a context receiving the qubit cannot behave differently depending on its state if not through measurement, and when measured, the states of the two qubits coincide.
\end{example}
Our cs-bisimilarity extends the one of CQP in a context-based fashion over standard distributions.
CQP behaves as QPAlg with respect to unsent qubits (see the first line of~\autoref{tab:examples}).

\subsection*{qCCS}

Our process calculus takes its most direct inspiration from qCCS, a synchronous
CCS-style calculus with superoperators and measurements where
syntactic restrictions guarantee each qubit is sent at most once.
A feature of qCCS is the support for recursive processes, which we postpone to future work.
Two different labelled bisimilarities are proposed for qCCS\@: the first is based on standard probabilistic bisimulations; the second one relies on transition consistency and subdistributions.

The probabilistic bisimilarity~\cite{feng_probabilistic_2007,feng_bisimulation_2012} (denoted by $\sim_{p}$ in \autoref{tab:examples}) requires bisimilar processes to send the same names on the same channels and to produce the same quantum state of qubits that are not owned any more.
In addition, bisimulations must be closed under applications of trace-preserving superoperators on not-owned qubits.
Therefore, the 
processes of the first line of~\autoref{tab:examples} are distinguished.
We recover this assumption by sending the qubits (see the third line).

Moreover, \autoref{thm:qccs} shows that cd-bisimilarity replicates the $\sim_p$ requirements over superoperators and not owned qubits.
The probabilistic bisimilarity of qCCS is proved to be a congruence with respect to parallel composition.
Further extensions are the (weak) open bisimilarity proposed in~\citet{deng_open_2012}, proven to be a weak barbed congruence, and a symbolic version of the bisimilarity in~\citet{feng_symbolic_2014},
that relieves from considering all the (universally quantified) quantum states of a configuration when verifying bisimilarity.

Configurations like the ones of the fifth line of~\autoref{tab:examples} are not bisimilar for $\sim_p$, even though they are indistinguishable according to quantum theory.
This discrepancy was signalled in~\citet{kubota_application_noyear} and lead to a new proposal called distribution bisimilarity~\cite{feng_toward_2015-1,deng_bisimulations_2018} (denoted by $\sim_{d}$ in \autoref{tab:examples}), which is directly defined on distributions and is based on transition consistency.
A distribution is called transition consistent if any configuration in its support has exactly the same set of enabled visible actions.
In addition to closure with respect to superoperator application, bisimilar distributions are required to be such that: $(i)$ the weighted sums of the state of not owned qubits in the support coincide; $(ii)$ the possible transitions of one transition consistent distribution
are matched by the other one, and; $(iii)$ if the distributions are not transition consistent, then they must be decomposable in bisimilar transition consistent distributions.
On the one hand, considering distributions as a whole when comparing the quantum states equates processes like the ones of the fifth line of~\autoref{tab:examples}.
On the other hand, transition consistent decompositions recover a weakened version of decomposability, avoiding equating distributions that cannot evolve because the processes in the supports enable different actions only.

The use of transition consistency in $\sim_d$ implicitly constrains non-determinism.
In the last line of~\autoref{tab:examples}, we compare such constraints with the ones we impose on lqCCS\@.
While our $\sim_{cs}$ is capable of distinguishing the two processes, $\sim_d$ cannot.
In fact, the lifting of the \emph{labelled} semantics of qCCS forbids processes from replicating the moves of their refinements (as lqCCS does).
\begin{example}
	Consider the pair of processes of the last line of~\autoref{tab:examples}. They reduce to
	\begin{gather*}
		\Delta = \overline{\conf{\dket{0}, c!q + d!q}} \psum{1/2} \overline{\conf{\dket{1}, c!q + d!q}} \quad
		\Theta = \overline{\conf{\dket{+}, c!q + d!q}} \psum{1/2} \overline{\conf{\dket{-}, c!q + d!q}}
	\end{gather*}
	We show in~\autoref{ex:non-detproc2} that $\Delta \not\sim_{cs} \Theta$, as $\Delta$ can choose to send the qubit over $c$ only when it is set to $\kz$, while $\Theta$ cannot.
	In the distribution bisimulation of~\citet{feng_toward_2015-1}, instead, only the moves that choose the same channel in all the configurations of the support are considered, deeming the two distributions bisimilar.
  This means that $\meas[M_{01}]{q}{x}. (c!q + d!q)$ cannot use the value of $x$ to choose the channel over which to send, as it would be expected, and thus that the constraints over non-determinism are arguably too strong in~\citet{feng_toward_2015-1}.
\end{example}

Finally, \citet{feng_toward_2015-1} acknowledge that \emph{weak} distribution bisimilarity is not a congruence.
Its strong version is not a congruence either: take $\Delta \sim_{d} \Theta$ and $B[\blank]$ of~\autoref{ex:broken-nondet}, it is easy to show that $B[\Delta] \not\sim_{d} B[\Theta]$.
The same also holds for $\sim_{cs}$, which is a congruence with respect to observers but not to parallel composition.
We believe that $\sim_{d}$ actually verifies the indistinguishability property over deterministic processes of~\autoref{thm:propertyA}, 
and
$B[\Delta] \not\sim_{d} B[\Theta]$ shows that the result cannot be extended to general processes.
On a similar note, 
$\sim_{d}$ does not preserve the expressiveness 
of non-deterministic choices based on classical information stated in~\autoref{thm:nondetVSite}.

\section{Conclusions and Future Work}\label{sec:conc}
We presented lqCCS, a quantum process calculus with asynchronous communication and a linear type system guaranteeing that each qubit is sent or discarded \emph{exactly} once.
 The latter lifts the semantics from making arbitrary assumptions about the observability of unsent qubits, which was a discrepancy among related works.

The main result of this work is a novel stateful reduction semantics together with a saturated probabilistic bisimilarity that relies on contexts for distinguishing quantum processes. These choices allowed us to investigate and compare the discriminating capabilities of the bisimilarity against the principles of quantum theory. 
By employing standard contexts, we found that the problems highlighted by~\citet{davidson_formal_2012} and~\citet{kubota_application_noyear} are caused by the interaction between non-determinism and quantum features.
In particular the standard notion of non-determinism subverts a defining feature of quantum theory
by allowing contexts to perform moves based on the possibly unknown quantum state, without 
performing a measurement and thus perturbing it.
We enhanced the semantics of lqCCS and constrained non-determinism by requiring the contexts to perform the same move in all the configurations of a given distribution (when no classical branching is possible). 
The resulting bisimilarity relation is strictly coarser than the unconstrained one.  

We prove two main properties: $(i)$ indistinguishability of quantum states can be lifted to classes of bisimilar distributions of lqCCS configurations; and 
$(ii)$ non-deterministic choices can perform moves according to known classical values, simulating the semantics of boolean guards.
Intuitively, the first guarantees that constraints are indeed sufficient to 
prevent non-determinism from subverting quantum features, while the second 
ensures that constraints are not too restrictive.
Moreover, we showed by counterexample in~\autoref{tab:examples} that no bisimilarity in the literature
satisfies both of them.

Furthermore, we proved that the novel bisimilarity is linear with respect to probabilistic composition, and closed for superoperator application on qubits not appearing in the processes, which are also required to be equal in bisimilar distributions.
Moreover, discarded qubits can be ``traced out'' without affecting bisimilarity. 
An up-to technique is given to aid bisimilarity proofs.
We tested our approach by modeling and analysing three real-world protocols.
Finally, we compared our findings with previously proposed bisimilarities, using simple prototypical cases that highlights dinstinguishability features required by quantum theory.

\paragraph{Discussion}
Our work starts with an example-based analysis of the proposed bisimilarities and their adherence to expected indistinguishability results prescribed by quantum theory.
Through our analysis, we identify some desired properties that we later prove for our proposed bisimilarity (mainly,~\autoref{thm:propertyA} and \autoref{thm:nondetVSite}).
Bisimilarities for quantum processes are difficult to justify and validate, as
we have no touchstone but the prescriptions of quantum mechanics about what can be operationally distinguished.
Our set of examples and the properties they suggest help with this problem, and they are sufficient to tell apart cs-bisimilarity and the previous proposals.

We opted to work with a well established and fairly standard calculus, limiting changes to the ones needed for addressing the problems at hand.
Thus, we left unaltered the semantics of processes, only restricting the behaviour of observers, and we stick with a classical probabilistic approach.

Alternative approaches may be considered.
One could characterise feasible non-deterministic choices in general, constraining them also in processes. 
Unfortunately, the expressivity of processes would be weaker than the ones considered in the related works,  since this approach may result in overly constrained processes, thus making the comparison less direct.
Otherwise, one could look for a more suitable notion of \emph{quantum} distribution that naturally satisfy the desired properties of indistinguishability, similarly to what is done by~\citet{davidson_formal_2012} and \citet{feng_symbolic_2014}.

Both these approaches seem promising follow-ups for our work, that can serve as a preliminary investigation on how process semantics should be changed to accommodate to quantum theory.

\paragraph{Future Work}
One aim of our future work is to extend lqCCS, namely with qubit declaration primitives and recursive processes.
We will also explore weak versions of the constrained bisimilarity, as done by~\citet{feng_probabilistic_2007} and \citet{davidson_formal_2012},
and enhanced proof methodologies, i.e.\ by investigating pruning techniques and by looking for an equivalent labelled bisimulation, following the approach of~\citet{bonchi_general_2014}.
The advantage would be two-fold: to avoid the
universal quantification over contexts and to identify an adequate set
of observable properties.

Moreover, we will further investigate some
of the shortcomings of probabilistic bisimulation with respect to non-determinism.
The solution presented in this work is only one of many possible approaches, where non-determinism is unconstrained in the processes and constrained in the contexts.
An alternative approach is to characterize feasible choices in general, constraining non-determinism also in processes, by
taking into account all the legit reasons for configurations of the same distribution to behave differently. 
We guess this will give a bisimilarity that is a congruence and that still satisfies our desired properties, something which is missing among current proposals.

Finally we will investigate how spurious non-deterministic moves impact on model checking when applied to quantum systems.

\begin{acks}
This study was carried out within the National Centre on HPC, Big Data and Quantum Computing - SPOKE 10 (Quantum Computing) and received funding from the European Union Next-GenerationEU - National Recovery and Resilience Plan (NRRP) – MISSION 4 COMPONENT 2, INVESTMENT N. 1.4 – CUP N. I53C22000690001.
\end{acks}

\bibliography{references}

\clearpage
\appendix

\renewcommand{\com}[1]{}

\section{General Quantum Operations}\label{sec:proofpreliminaries}
In order to prove some quantum-related properties, we need some additional definitions, which will be used in the rest of the appendices. We extend the notion of density operators and superoperators, in order to include also \emph{sub-probability distributions} and  \emph{trace-nonincreasing} superoperators.

A sub-probability distribution on a set $S$ is a function $\Delta : S \rightarrow [0,1]$ such that $\sum_{s\in S} \Delta(s) \leq 1$. 

\begin{definition}[Partial Density Operators]
Given a Hilbert Space $\hilbert$, to each subprobability distribution of states $\Delta$ we associate a partial density operator $\rho = \sum_i \Delta(\psi_i)\dket{\psi_i}$. We indicate with $\pdensity{\hilbert}$ the set of partial density operators of $\hilbert$.
\end{definition}

Notice that density operators and partial density operators are defined in the same way, but the latter can have trace less than one, since they come from sub-probability distributions.

\begin{definition}[Trace-nonincreasing superoperators]
A \emph{superoperator} $\mathcal{E}\colon \pdensity{\hilbert} \rightarrow \pdensity{\hilbert}$ is a function with a Kraus operator sum decomposition $\{E_i\}_i$, with $1 \leq i \leq \dim(\hilbert)$ such
that 
\begin{gather*} 
\sop{E}{\rho} = \sum_i E_i\rho E_i^\dag\ \text{ and }\ \sum_i E_i^\dag E_i \sqsubseteq I_\hilbert 
\end{gather*}
Where $A \sqsubseteq B$ means that $B - A$ is a positive semidefinite matrix (i.e., $\bra{\psi} (B - A) \ket{\psi} \geq 0$ for any $\ket{\psi}$),
and $I_{\hilbert}$ is the identity operator on $\hilbert$. We call $\soset{\hilbert}$ the set of (trace-nonincreasing) superoperators on $\hilbert$, and
$\tsoset{\hilbert} \subseteq \soset{\hilbert}$ the set of all
\emph{trace-preserving} superoperators, i.e.\ superoperators such that for any
$\rho \in \density{\hilbert}$: $\tr(\sop{E}{\rho}) = \tr(\rho)$.
\end{definition}

For a trace-nonincreasing superoperator $\mathcal{E}$ we have that $tr(\sop{E}{\rho}) \leq tr(\rho)$, while the superoperators used in the rest of the paper were \emph{trace-preserving superoperators} ($tr(\sop{E}{\rho}) = tr(\rho)$). Given a measurement in lqCCS $\meas{\tilde{x}}{y}.P$, where the symbol $M$ denotes measurement $\{M_m\}_m$, we define for each outcome a trace-nonincreasing superoperator 
\[\sop[\tilde{x}, m]{M}{\rho} = (M_m \otimes I) \rho (M_m^\dagger \otimes I)\]
where $I$ is the identity matrix on the qubits not in $\tilde{x}$. 

\begin{lemma}\label{ptrace drops sop}
Let $\rho \in \density{\hilbert_{\tilde{p}} \otimes \hilbert_{\tilde{q}}}$, with $\tilde{p} = p_0 \ldots p_{n-1}$ and $\tilde{q} = q_0 \ldots q_{m-1}$.
 Then, for any two superoperator $\mathcal{E}_{\tilde{p}} \in \soset{\hilbert_{\tilde{p}}}$
 and $\mathcal{F}_{\tilde{q}} \in \tsoset{\hilbert_{\tilde{q}}}$
\[
 \ptrace{q}{(\mathcal{E}_{\tilde{p}} \otimes \mathcal{F}_{\tilde{q}}) (\rho)} = \sop[\tilde{p}]{E}{\ptrace{q}{\rho}}
\]
\end{lemma}
\begin{proof}
  The density operator $\rho$ is described in the form $\rho = \sum_{ijkl} \rho_{ij,kl}\ketbra{ij}{kl}$, where $i,k$ are elements of some orthonormal basis of $\hilbert_{\tilde{p}}$ and $j,l$ of $\hilbert_{\tilde{q}}$
  \begin{align*}
    \ptrace{q}{(\mathcal{E}_{\tilde{p}} \otimes \mathcal{F}_{\tilde{p}})\rho} &= \sum_{ijkl}\rho_{ij,kl}\ptrace{q}{\mathcal{E}_{\tilde{p}}(\ketbra{i}{k}) \otimes \mathcal{F}_{\tilde{q}}(\ketbra{j}{l})} \\
    &= \sum_{ijkl}\rho_{ij,kl}\,\mathcal{E}_{\tilde{p}}(\ketbra{i}{k})\tr(\mathcal{F}_{\tilde{q}}\ketbra{j}{l}) \\
    &= \mathcal{E}_{\tilde{p}}\left(\sum_{ijkl}\rho_{ij,kl}\ketbra{i}{k}\tr(\ketbra{j}{l})\right) \\
    &= \mathcal{E}_{\tilde{p}}\left(\sum_{ijkl}\rho_{ij,kl}\ptrace{q}{\ketbra{ij}{kl}}\right) \\
    &= \mathcal{E}_{\tilde{p}}\left(\ptrace{q}{\sum_{ijkl}\rho_{ij,kl}\ketbra{ij}{kl}}\right)
    = \mathcal{E}_{\tilde{p}}(\ptrace{q}{\rho}) \qedhere
  \end{align*}
\end{proof}

\section{Proofs of Section~\ref{sec:lqCCS}}\label{sec:typeappendix}
\uniqueness*
\begin{proof}
	By induction on the derivations $\Sigma \vdash P$ and $\Sigma'
		\vdash P$. It is trivial to show that the derivation rules are
	non-overlapping, i.e.\ two distinct rules cannot be applied for the same
	process, thus the derivation tree must have the same rule applications.
	Only the \textsc{QSend} and \textsc{Par} rules can introduce quantum variables into the context.
	Since the \textsc{QSend} rule is completely determined by
	the syntax of the process, both derivations must introduce the same quantum
	variables in their respective contexts. Furthermore, for the \textsc{Par}
	rule, by induction it must hold that the two possibly smaller contexts of the subprocesses are equal
	between the two derivations, thus the resulting union is equal.
\end{proof}

\begin{lemma}\label{lem:quantumsubst}
	If $\Sigma \cup \{x\} \vdash P$ and $v \not\in \Sigma$ then $\Sigma \cup \{v\} \vdash P[\sfrac{v}{x}]$.
\end{lemma}
\begin{proof}
	By induction on the derivation of $\Sigma \cup \{x\} \vdash P$.
	Let us analyse the interesting cases: For the \textsc{QOp} rule it must be
	that $P = \sop{E}{\tilde{x}}.Q$ for some process $Q$. By induction
	hypothesis, it holds that $\Sigma \cup \{v\} \vdash Q$, but $v
		\not\in \tilde{x}$ by the hypothesis $v \not\in \Sigma$, thus we can
	apply the \textsc{QOp} rule. The same line of reasoning is valid for the
	\textsc{QMeas} rule. The \textsc{QSend} rule is also guaranteed by the $v
		\not\in \Sigma$ requirement. Finally, the derivation of \textsc{Par}
	imposes that only one of the components, w.l.o.g. $P_i$, contains the
	variable $x$ in its quantum environment $\Sigma_i$. Thus, by induction we
	obtain $\Sigma_i[\sfrac{v}{x}] \vdash P_i$. While for the other
	components, there is no change since they do not contain the variable $x$.
	However, since $v \not\in \Sigma$, we can conclude that all smaller
	environments are still pairwise distinct, thus we can apply the
	\textsc{Par} rule again.
\end{proof}

\typingpreservation*
\begin{proof}

	By induction on the transition relation $\longrightarrow$ rules. Most
  rules follow trivially by induction on the premises, in particular no rules affect $\Sigma_\rho$. For measurements, with
	the \textsc{QMeas} rule, a substitution property for classical values is
	required, however, it is trivial in the standard type system for naturals and
	booleans. For transitions with the \textsc{Congr} rule, it must be that
  $\Sigma_P \vdash Q$ by definition of $\equiv$. By induction $(\Sigma_\rho, \Sigma_P)
    \vdash \Delta$ and, again, by definition of $\equiv$ it must hold that
  $(\Sigma_\rho, \Sigma_P) \vdash \Delta'$. The remaining interesting case is the
	\textsc{Reduce} rule. Since $P = (c!v + R) \parallel ((c?x.P') + Q)$ then if $c$
	is not a quantum channel $\Sigma \vdash P'$. By classical substitution results
	it holds that $\Sigma \vdash P'[\sfrac{v}{x}]$. If $c$ is a quantum channel
	$\Sigma \cup \{x\} \vdash P'$. By \autoref{lem:quantumsubst} it holds that
	$\Sigma \vdash P'[\sfrac{v}{x}]$, because $v$ cannot be in $\Sigma$ due to the
	application of the \textsc{Par} rule for the $\parallel$ operator.
\end{proof}

\section{Proofs of Section~\ref{behaviouralAssessment}}\label{sec:assessmentappendix}
\subsection{Proof of \autoref{thm:probsmallerthanconst}}
To prove \autoref{thm:probsmallerthanconst} we need to introduce tagged processes. The idea is that whenever a distribution $O[\Delta]$ can perform an indexed transition $\longsquiggly_\pi$, then the tagged version $\ptag{O[\Delta]}$ express a barb named $\pi$ with probability one. This allows us to interpret the enhanced semantics as standard probabilistic reductions, and show that constrained saturated bisimilarity does not have greater discriminating power (thanks to indices) with respect to saturated bisimilarity.

\begin{definition}[Tagging operation]
We define the $\ptag{\blank}_k^\pi$ operation, which tags (the observer in a) configuration with barbs in $\{\ell, r, \iota\}^*$. $k$ is the starting string of the barb, and $\pi$ is used when we want to add an additional $\iota$ in front of the barbs starting with $\pi$
\begin{align*}
	\ptag{\conf{\rho, P, R}}^\pi_k &= \conf{\rho, P \parallel \ptag{R}^\pi_k} 
  & & & \ptag{\bot}^\pi_k &= \bot 
  \\
	\ptag{P \parallel Q}^{\ell\pi}_k &= \ptag{P}^\pi_{k\ell} \parallel \ptag{Q}_{kr} 
	& & & \ptag{P \parallel Q}^{r\pi}_k &= \ptag{P}_{k\ell} \parallel \ptag{Q}^{\pi}_{kr} 	 
	\\
	\ptag{P}^\varepsilon_k &= \ptag{P}_{k\barr}     
  & & & \ptag{\nil_{\tilde e}}_k &= \nil_{\tilde e}
  \\
  \ptag{\sop{E}{\tilde{x}}.P}_k &= \sop{E}{\tilde{x}}.\ptag{P}_{k\iota} + k!0 
  & & & \ptag{\meas{\tilde{x}}{y}.P}_k &= \meas{\tilde{x}}{y}.\ptag{P}_{k\iota} + k!0 
  \\
  \ptag{c?x.P}_k &= c?x.\ptag{P}_{k\iota} + k!0 
  & & & \ptag{c!e}_k &= c!e + k!0 
  \\
  \ptag{P + Q}_k &= \ptag{P}_k + \ptag{Q}_k 
  & & & \ptag{\ite{e}{P}{Q}}_k &= \ite{e}{\ptag{P}_k}{\ptag{Q}_k}
\end{align*}
We extend the tagging operation via linearity: $\ptag{\Delta_1 \psum{p} \Delta_2}_k^\pi = \ptag{\Delta_1}_k^\pi \psum{p} \ptag{\Delta_2}_k^\pi$. We will often write $\ptag{\blank}$ instead of $\ptag{\blank}_\varepsilon$.
\end{definition}

\begin{definition}[Tagged transitions]
We define the set of tags a configuration \[T(\iconf) = \{\pi \in \{l, r, \barr\}^* \mid \iconf\downarrow_\pi\}\]
We define two kinds of tagged transition from (tagged) configurations to distributions
	\begin{align*}
	\ptag{\iconf} \rightarrow \Delta : \diamond \qquad &\text{iff } T(\ptag{\iconf}) = T(\ptag{\iconf'}) \ \forall \ptag{\iconf'} \in \lceil \Delta \rceil \\
	\ptag{\iconf} \rightarrow \Delta : \lambda \qquad &\text{iff } \lambda \in T(\ptag{\iconf}) \text{ and } \lambda \not\in T(\ptag{\iconf'}),  T(\ptag{\iconf})\setminus\{\lambda\} \subseteq T(\ptag{\iconf'}) \ \forall \ptag{\iconf'} \in \lceil \Delta \rceil
	\end{align*}
	And we extend this notion to distributions via linearity
	\[
		\ptag{\Delta_1} \psum{p} \ptag{\Delta_2} \rightarrow \Delta'_1 \psum{p} \Delta'_2 : \pi  \text{ if and only if } \ptag{\Delta_i} \rightarrow \Delta_i : \pi \text{ for } i \in {1, 2}
	\]
\end{definition}

We now prove that the indexed transitions of $\conf{\rho, P, R}$ are essentially the same of the tagged transitions of $\ptag{\conf{\rho, P, R}}$.

\begin{lemma}\label{interpret_diamond} 
	For any observer $R$ we have
	\[
	\forall k \ \left( 
	\conf{\rho, P, R} \longsquiggly_\diamond \Delta 	
	\Longleftrightarrow	
	\conf{\rho, P \parallel \ptag{R}_k}  \rightarrow \Delta' : \diamond	\text{ with } \ptag{\Delta}_k \equiv \Delta'
	\right)\]
\end{lemma}
\begin{proof}
The proof follows easily from the definition of $\diamond$-tagged transitions.
\end{proof}

\begin{lemma}\label{interpret_fw} %
  For any observer $R$ and index $\lambda \neq \diamond$, we have
  \[
    \forall k \ldotp \conf{\rho, P, R} \longsquiggly_\lambda \Delta \Longrightarrow
    \conf{\rho, P \parallel \ptag{R}_k} \to \ptag{\Delta}^\lambda_k : k\lambda
  \]
\end{lemma}
\begin{proof}
  We will proceed by induction on the transition $\longsquiggly_\lambda$.

  If rule \textsc{OQOp} is valid, then $R \equiv \sop{E}{\tilde{q}}.R_1$, we
  have
  \[
    \iconf = \conf{\rho, P, \sop{E}{\tilde{q}}.R_1} \longsquiggly_\varepsilon
    \sconf{\mathcal{E}_{\tilde{q}}(\rho), P, R_1} = \Delta
  \]
  and
  \[
    \ptag{\iconf}_k =
    \conf{\rho, P \parallel (\sop{E}{\tilde{q}}.\ptag{R_1}_{k\barr} +
    k!0)}
    \to \sconf{\mathcal{E}_{\tilde{q}}(\rho), P \parallel \ptag{R_1}_{k\barr}} =
    \ptag{\Delta}^\varepsilon_k
  \]
  We must further check that:
  $\ptag{\iconf}_k \to \ptag{\Delta}^\varepsilon_k :
  k$. The first condition, $k \in T(\ptag{\iconf}_k)$, is true by $k!0$.
  The second, $k \not\in T(\iconf')$, is true because any new barb expressed by $\ptag{R_1}_{k\barr}$ must start with
  $k\barr$. The third condition, $T(\ptag{\iconf}_k) \setminus \{k\} \subseteq
  T(\ptag{\Delta}^\varepsilon_k)$, is also trivially true since the only syntactic difference between the $\iconf$ and
  $\Delta$ is the barb-expressing $k!0$.

  If rule \textsc{OQMeas} is valid, then $R \equiv \meas{\tilde{q}}{y}.R_1$, we
  have
  \[
    \iconf = \conf{\rho, P, \meas{\tilde{q}}{y}.R_1} \longsquiggly_\varepsilon
    \sum_{i=0}^{|M| - 1}\distelem{p_m}{\sconf{\frac{\rho_m}{p_m}, P, R_1[\sfrac{m}{y}]}} = \Delta
  \]
  and, noting as before that $\ptag{R[\sfrac{v}{x}]}^\lambda_k =
  \ptag{R}^\lambda_k[\sfrac{v}{x}]$
  \[
    \ptag{\iconf}_k
    = \conf{\rho, P \parallel (\meas{\tilde{q}}{y}.\ptag{R_1}_{k\barr} +
    k!0)}
    \to \sum_{i=0}^{|M| - 1}\distelem{p_m}{\sconf{\frac{\rho_m}{p_m}, P \parallel
    (\ptag{R_1}_{k\barr}[\sfrac{m}{y}] + k!0)}} =
    \ptag{\Delta}^\varepsilon_k
  \]
  $\ptag{\iconf}_k \to \ptag{\Delta}^\varepsilon_k :
  k$ holds as for the \textsc{OQOp} noting that each element in
  $\ptag{\Delta}^\varepsilon_k$ has syntactically the same sendings.

  If rule \textsc{Input} is valid, then $R \equiv c?x.R_1 + R_2$, we have
  \[
    \iconf =
    \conf{\rho, (c!v + P_1) \parallel P_2 \setminus D, c?x.R_1 + R_2} 
    \longsquiggly_\varepsilon \sconf{\rho, P_2 \setminus D, R_1[\sfrac{v}{x}]}
    = \Delta
  \]
  and
  \begin{align*}
    \ptag{\iconf}_k
    &= \conf{\rho, ((c!v + P_1) \parallel P_2 \setminus D) \parallel (c?x.\ptag{R_1}_{k\barr} + k!0 + \ptag{R_2}_k)} \\
    &\equiv \conf{\rho, ((c!v + P_1) \parallel P_2 \parallel (c?x.\ptag{R_1}_{k\barr} + k!0 + \ptag{R_2}_k)) \setminus D} \\
    &\to \sconf{\rho, (P_2 \parallel \ptag{R_1}_{k\barr}[\sfrac{v}{x}]) \setminus D} \\
    &\equiv \sconf{\rho, (P_2 \setminus D) \parallel \ptag{R_1[\sfrac{v}{x}]}_{k\barr}}
    = \ptag{\Delta}^\varepsilon_k
  \end{align*}
  $\ptag{\iconf}_k \to \ptag{\Delta}^\varepsilon_k :
  k$ holds trivially as for the \textsc{OQOp}.

  If rule \textsc{Output} is valid, then $R \equiv c!v$, we have
  \[
    \iconf =
    \conf{\rho, (c?x.P_1) + P_2) \parallel P_3 \setminus D, c!v} 
    \longsquiggly_\varepsilon \sconf{\rho, (P_1[\sfrac{v}{x}] \parallel P_3) \setminus D, \nil}
    = \Delta
  \]
  and
  \begin{align*}
    \ptag{\iconf}_k
    &= \conf{\rho, (((c?x.P_1) + P_2) \parallel P_3 \setminus D) \parallel (c!v + k!0)} \\
    &\equiv \conf{\rho, (((c?x.P_1) + P_2) \parallel P_3 \parallel (c!v + k!0)) \setminus D} \\
    &\to \sconf{\rho, (P_1[\sfrac{v}{x}] \parallel P_3 \parallel \nil) \setminus D} \\
    &\equiv \sconf{\rho, (P_1[\sfrac{v}{x}] \parallel P_3 \setminus D) \parallel \nil}
    = \ptag{\Delta}^\varepsilon_k
  \end{align*}
  $\ptag{\iconf}_k \to \ptag{\Delta}^\varepsilon_k :
  k$ holds trivially as for the \textsc{QInput}.

  If rule \textsc{ParL} is valid, then $R \equiv R_1 \parallel R_2$, we have
  \[
    \iconf =
    \conf{\rho, P, R_1 \parallel R_2}
    \longsquiggly_{\ell\lambda} \Theta \parallel R_2
    = \Delta
  \]
  By inductive hypothesis, for any $k$
  \[
    \conf{\rho, P, R_1} \longsquiggly_\lambda \Theta \Longrightarrow \conf{\rho, P \parallel \ptag{R_1}_k} \to \ptag{\Theta}^\lambda_k : k\lambda
  \]
  and
  \begin{align*}
    \ptag{\iconf}_k
    = \conf{\rho, P \parallel \ptag{R_1}_{k\ell} \parallel \ptag{R_2}_{kr}}
    \to \ptag{\Theta}^\lambda_{k\ell} \parallel \ptag{R_2}_{kr}
    = \ptag{\Delta}^{\ell\lambda}_k
  \end{align*}
  $\ptag{\iconf}_k \to \ptag{\Delta}^{\ell\lambda}_k : k\ell\lambda$ holds by inductive hypothesis by setting $k = k\ell$.
  The case for \textsc{ParR} is analogous to the case for \textsc{ParL}.

  Finally, the \textsc{Congr} rule is simply resolved by induction, since its trivial to show that the congruence relation does not change the barbs expressed by a distribution.
\end{proof}

\begin{lemma}\label{interpret_bw} %
	For any observer $R$ and index $\lambda \neq \diamond$, we have
	\[\forall k \ \left( \conf{\rho, P \parallel \ptag{R}_k}  \rightarrow \Delta : k\lambda
	\Longrightarrow	
	\conf{\rho, P, R} \longsquiggly_\lambda \Delta' \text{ with } \ptag{\Delta'}_k^{\lambda} \equiv \Delta
	\right)\]
\end{lemma}
\begin{proof}

Note that any transition $\ptag{\mathcal{C}}_k = \conf{\rho, P \parallel \ptag{R}_k}  \rightarrow \Delta : k\lambda$ must involve $R$, because a transition modifying only $P$ would not be tagged as $k\lambda$. We will then proceed on induction on the syntax of $R$.

If $R = \nil_{\tilde{e}}$, there are no $k\lambda$-tagged transitions going out from $\ptag{\conf{\rho, P\parallel \nil_{\tilde{e}}}} =  \conf{\rho, P\parallel \nil_{\tilde{e}}}$.

If $R = \sop{E}{\tilde{x}}.R' $, we have \[\ptag{\mathcal{C}}_k \equiv \conf{\rho, P \parallel \sop{E}{\tilde{x}}.\ptag{R'}_{k\barr} + k!0} \rightarrow \Delta \equiv \singleton{\conf{\sop[\tilde{e}]{E}{\rho}, P \parallel \ptag{R'}_{k\barr}}}~:~k\]
But then 
\begin{align*}
&\conf{\rho, P, \sop{E}{\tilde{x}}.R'} \longsquiggly_\varepsilon \Delta' = \singleton{\conf{\sop[\tilde{e}]{E}{\rho}, P, R'}}, \text{ and }\\
&\ptag{\Delta'}_k^{\varepsilon} = \singleton{\conf{\sop[\tilde{e}]{E}{\rho}, P \parallel \ptag{R'}_{k}^{\varepsilon}}} = \singleton{\conf{\sop[\tilde{e}]{E}{\rho}, P \parallel \ptag{R'}_{k\barr}}} \equiv \Delta.
\end{align*}

If $R = \meas{\tilde{x}}{y}.R'$, we have
\[\conf{\rho, P \parallel \meas{\tilde{x}}{y}.\ptag{R'}_{k\barr} + \varepsilon!0} \rightarrow \Delta \equiv \sum_m \distelem{p_m}{\singleton{\conf{\rho_m, P \parallel \ptag{R'}_{k\barr}[m/y]}}} : \varepsilon\]
for some set of $p_m$ and $\rho_m$. Then the proof is identical to the previous case, noticing that $\ptag{R[v/x]}_k^\lambda = \ptag{R}_k^\lambda[v/x]$ as the tags depends on the position of the processes and not on their content.

If $R = c!e $, then $P$ must be receiving on the unrestricted channel $c$, so we have \[\ptag{\mathcal{C}}_k \equiv \conf{\rho, (((c?x.P') + P'') \parallel Q \parallel (c!v + \varepsilon!0)) \setminus D} \rightarrow \Delta \equiv \singleton{\conf{\rho, ( P'[v/x] \parallel Q )\setminus D}}~:~\varepsilon\]
for some (possibly empty) set of restrictions $D$ and processes $P', P'', Q$. But then \\ $\mathcal{C} \equiv_o \conf{\rho, (((c?x.P') + P'') \parallel Q)  \setminus D, c!v} \longsquiggly_\varepsilon \Delta' = \singleton{\conf{\rho, ( P'[v/x] \parallel Q )\setminus D, \nil}}$ and $\ptag{\Delta'}_k^{\varepsilon} \equiv \Delta$.

If $R$ is a sum of reception, then $P$ must be sending on at least one of the channel in $R$, so we have 
\[\ptag{\mathcal{C}}_k \equiv \conf{\rho, (((c!v + P) \parallel Q \parallel (c?x.\ptag{R'}_{k\barr} + \ptag{T}_k + \varepsilon!0)) \setminus D} \rightarrow \Delta \equiv \singleton{\conf{\rho, ( Q \setminus D )\parallel R'[v/x]}}~:~\varepsilon\]
for some (possibly empty) set of restrictions $D$ and process $Q$, and the proof is identical to the previous case.

If $R = \ite{e}{R_1}{R_2}$, then the only way to have a transition is after applying the structural congruence and obtaining either $R_1$ or $R_2$, from which the desired property is obtained thanks to the inductive hypothesis and the fact that $\ptag{R}_k \equiv \ptag{R'}_k$ if and only if $R \equiv_o R'$.

If $R =  R_1 \expar R_2 $ and $\ptag{\mathcal{C}}_k \equiv \conf{\rho, P \parallel \ptag{R_1}_{kl} \parallel \ptag{R_2}_{kr}} \rightarrow \Delta : k\lambda$, then the transition cannot be a synchronisation between $R_1$ and $R_2$ because all the tags of $\ptag{C}_k$ must be in $\Delta$ except from $k\lambda$, and a synchronisation would destroy two tags, not just one. Suppose then that the transition involves $R_1$, the $R_2$ case is symmetric. From the $\rulename{SemPar}$ rule must be $\conf{\rho, P \parallel \ptag{R_1}_{k\ell}} \rightarrow \Delta_\ell$ with $\Delta = \Delta_\ell \parallel\ptag{R_2}_{kr} $.  This transition is still tagged with $k\lambda$, as the $k\lambda$ tag is present in $\ptag{R_1}_{k\ell} \parallel \ptag{R_2}_{kr}$ but is absent in $\Delta_\ell \parallel\ptag{R_2}_{kr}$, no must be in $\ptag{R_1}_{kl}$. Besides, we know that $k\lambda = k\ell\lambda'$ for some $\lambda'$, as it is easy to verify that all the tags in $\ptag{R_1}_{k\ell}$ start with $k\ell$. Then we can apply the inductive hypothesis 
\begin{gather*}
\conf{\rho, P \parallel \ptag{R_1}_{k\ell}} \rightarrow \Delta_\ell :k\ell\lambda' 
\qquad \Rightarrow  \qquad
\conf{\rho, P, R_1} \longsquiggly_{\lambda'} \Delta_\ell' \text{ with } \ptag{\Delta_\ell'}_{k\ell}^{\lambda'} = \Delta_\ell
\end{gather*}
and derive 
\[\conf{\rho, P, R_1 \expar R_2} \longsquiggly_{\lambda} \Delta' = \Delta_\ell'\parallel R_2 \text{ with } \ptag{\Delta'}_k^{\lambda} = \Delta
\]
since $\lambda = \ell\lambda'$ and $\ptag{\Delta'}_k^{\lambda} = \ptag{\Delta_\ell'\parallel R_2}_k^{\ell\lambda'} = \ptag{\Delta_\ell'}_{k\ell}^{\lambda'} \parallel \ptag{R_2}_{kr}  = \Delta_\ell \parallel \ptag{R_2}_{kr} = \Delta$.
\end{proof}

\begin{lemma}\label{interpret_distr}
For any $\Delta, \Delta'$ that do not contain $\bot$, we have 
	\[
	\forall k \ \left( 
	\Delta \longsquiggly_\pi \Delta' 	
	\Longleftrightarrow	
	\ptag{\Delta}  \rightarrow \Delta'' : \pi	\text{ with } \ptag{\Delta'}_k \equiv \Delta''
	\right)\]
\end{lemma}
\begin{proof}
The proof follows from \autoref{interpret_diamond}, \autoref{interpret_fw} and \autoref{interpret_bw}, thanks to linearity of $\longsquiggly_\pi$ and $\rightarrow~:~\pi$
\end{proof}

We can now finally prove \autoref{thm:probsmallerthanconst}

\probsmallerthanconst*
\begin{proof}
The direction $\sim_{cs} \not\subseteq \sim_{s}$ is given by \autoref{ex:zopm}.
For $\sim_{s} \subseteq \sim_{cs}$, we define the relation $\rel = \{(\Delta, \Theta) \mid \ptag{\Delta} \sim_{s} \ptag{\Theta}\} $
  And then we prove that $\rel$ is a constrained distribution bisimulation. Notice that, taken $\Delta, \Theta \in \sim_{s}$, when we interpret them as distribution of triples they have a $\nil$ observer, and so $\ptag{\Delta} = \Delta$ and $\ptag{\Theta} = \Theta$. This means that $\sim_{s} \subseteq \rel$, and if $\rel$ is a bisimulation then $\sim_s \subseteq \sim_{cs}$.
Notice also that $\rel$ is context-closed, meaning that if $\Delta, \Theta \in \rel$, then also $O[\Delta], O[\Theta] \in \rel$ for any $O[\blank]$, because $\sim_s$ is context-closed. Finally, $\rel$ is linear and decomposable, as $\ptag{\blank}$ is defined by linearity and $\sim_s$ is linear and decomposable (as proven in~\cite{hennessy_exploring_2012}, proposition 5.8).

Take $\Delta, \Theta \in \rel$. Since $\ptag{\Delta}$ and $\ptag{\Theta}$ are bisimilar, they express the same barbs. But $\ptag{\Delta}$ expresses at least all the barbs of $\Delta$ (and similarly $\ptag{\Theta}$), so also $\Delta$ and $\Theta$ must express the same barbs.

Suppose $\Delta \longsquiggly_\pi \Delta'$, we will prove that there exists a transition $\Theta \longsquiggly_\pi \Theta'$ with $\Delta' \rel \Theta'$, the other direction is symmetric. Let $\Delta' = \Delta'_{\not\bot} \psum{p} \singleton{\bot}$, where $\Delta'_{\not\bot}$ does not contain $\bot$.  Then, by decomposability of $\longsquiggly_\pi$ and of $\rel$, it must be $\Delta = \Delta_{\not\bot} \psum{p} \Delta_\bot$, and also $\Theta = \Theta_{\not\bot} \psum{p} \Theta\bot$. Since $\Delta_\bot \longsquiggly_\pi \bot$, from \autoref{interpret_distr} it follows that there is no $\Delta'$ such that $\ptag{\Delta_\bot} \rightarrow \Delta'~:~\pi$. But since $\ptag{\Delta_\bot} \sim_{s} \ptag{\Theta_\bot}$, using the third property of bisimilarity (from $\Theta$ to $\Delta$) we get that there is no  $\Theta'$ either such that $\ptag{\Theta_\bot} \rightarrow \Theta'~:~\pi$. Then, from \autoref{interpret_distr} it follows that $\Theta_\bot \longsquiggly_\pi \bot$, and of course $(\bot, \bot) \in \rel$.

Consider now $\Delta_{\not\bot}$. From $\Delta_{\not\bot} \longsquiggly_\pi \Delta'_{\not\bot}$, we know that $\ptag{\Delta_{\not\bot}} \rightarrow \ptag{\Delta'_{\not\bot}}^\pi~:~\pi$ tanks to \autoref{interpret_distr}, and so 
$\ptag{\Theta_{\not\bot}} \rightarrow \Theta'$ for a $\Theta' \sim_{s} \ptag{\Delta'_{\not\bot}}^\pi$. But since they are bisimilar, they express the same barbs, and from $\ptag{\Theta_{\not\bot}} \rightarrow \Theta'~:~\pi$ we 
get $\Theta_{\not\bot} \longsquiggly_\pi \Theta''$ with $\ptag{\Theta''}^\pi \equiv \Theta'$.
To sum up, we have 
\[
\begin{matrix}%
  \Delta = \Delta_{\not\bot} \psum{p} \Delta_{\bot} &
  \rel &
	\Theta = \Theta_{\not\bot} \psum{p} \Theta{\bot} \\
  \downsquiggly[\pi] & & \downsquiggly[\pi] \\
  \Delta' = \Delta'_{\not\bot} \psum{p} \bot &
  ? &
	\Theta' = \Theta'' \psum{p} \bot
\end{matrix}
\]
with $\ptag{\Delta'_{\not\bot}}^\pi \sim_{s} \ptag{\Theta''}^\pi$. It's possible to verify that $\ptag{\Delta}^\pi \sim_{s} \ptag{\Theta}^\pi$ implies $\ptag{\Delta} \sim_{s} \ptag{\Theta}$, as there is a bijection $f$ between the free channels of $\ptag{\Delta}$ and of $\ptag{\Delta}^\pi$, and so any context $O[\blank]$ applied on $\ptag{\Delta}$ has the same behaviour of the context $f(O[\blank])$ applied on $\ptag{\Delta}^\pi$, and vice versa. Given that, since $\bot, \bot \in \rel$ and $\rel$ is linear, we can conclude $\Delta' \rel \Theta'$.
\end{proof}

\subsection{Proof of \autoref{thm:linearity} and \autoref{thm:propertyA}}\label{uptoappendix}

Before proving \autoref{thm:linearity} and \autoref{thm:propertyA}, it is convenient to prove the soundness of the bisimilarity up to convex hull technique~\cite{bonchi_power_2017}.
To do so, let us define the function on relations $b$, which is the function of which bisimilarity is the greatest fix point
\begin{footnotesize}
	\[
	b(\rel) \coloneqq \left\{
	(\Delta, \Theta) \biggm| \begin{array}{c}
	\Delta \downarrow_{b}^p \Leftrightarrow \Theta \downarrow_{b}^p \\
	O[\Delta] \longsquiggly_\pi \Delta' \Rightarrow \exists \Theta' \ O[\Theta] \longsquiggly_\pi \Theta' \wedge \Delta'\,\rel\,\Theta'
	\\
	O[\Theta] \longsquiggly_\pi \Theta' \Rightarrow \exists \Delta' \ O[\Delta] \longsquiggly_\pi \Delta' \wedge \Delta'\,\rel\,\Theta'  
	\end{array}
	\right\}
	\]
\end{footnotesize}

and recall the definition of convex hull from~\cite{bonchi_power_2017}
\[
  Cv(\rel) \coloneqq \bigg\{\bigg(\sum_{i \in I} \distelem{p_i}{\Delta_i}, \sum_{i \in I} \distelem{p_i}{\Theta_i}\bigg)\ \bigg|\ \forall i \in I \ldotp \Delta_i\,\rel\,\Theta_i \bigg\}
\]

Observe that both $b$ and $Cv$ are monotone functions on the lattice of relations.

\begin{lemma}[$Cv$ is $b$-compatible]\label{thm:cvcomp}
	We have that $\forall \rel\ldotp Cv(b(\rel)) \subseteq b(Cv(\rel))$.
\end{lemma}
\begin{proof}
	Assume $(\Delta, \Theta) \in Cv(b(\rel))$. Then it must be $\Delta = \sum_{i \in I}\distelem{p_i} \Delta_i$ and $\Theta = \sum_{i \in I}\distelem{p_i} \Theta_i$ for a certain set of probabilities $\{p_i\}_{i \in I}$, with
	$\Delta_i\,b(\rel)\,\Theta_i$. So we have, for any $i \in I$
	\begin{gather*}
	\sum_{\mathcal{C} \downarrow_c} \Delta_i(\mathcal{C}) = \sum_{\mathcal{C} \downarrow_c} \Theta_i(\mathcal{C}) = q_i \\
	\sum_{\mathcal{C} \downarrow_c} \sum_{i \in I} p_i \Delta_i(\mathcal{C})  =
	\sum_{\mathcal{C} \downarrow_c} \sum_{i \in I} p_i \Theta_i(\mathcal{C}) = \sum_{i \in I} p_i q_i = q
	\end{gather*}
	meaning that $\Delta \downarrow_c^q$ if and only if $\Theta \downarrow_c^q$. The same holds also for $\Delta \downarrow_\bot^q$ and $\Theta \downarrow_\bot^q$, as 
		\begin{gather*}
	\Delta_i(\bot) = \Theta_i(\bot) = q_i \\
	\sum_{i \in I} p_i \Delta_i(\bot)  =
	\sum_{i \in I} p_i \Theta_i(\bot) = \sum_{i \in I} p_i q_i = q
	\end{gather*}
	
	Suppose $O[\Delta] = \sum_{i \in I} \distelem{p_i} O[\Delta_i]  \longsquiggly_\pi \Delta'$.
	From~\cite{hennessy_exploring_2012}, $\longsquiggly_\pi$ is left-decomposable, so it must be $\Delta'
	= \sum_{i \in I} \distelem{p_i} \Delta_i'$ with $O[\Delta_i] \longsquiggly_\pi\Delta_i'$. But
	since $\Delta_i\,b(\rel)\,\Theta_i$ it must be $O[\Theta_i] \longsquiggly_\pi
	\Theta_i'$ with $\Delta_i'\,\rel\,\Theta_i'$ for any $i \in I$, from which it follows that
	$\sum_{i \in I} \distelem{p_i} O[\Theta_i] \longsquiggly_\pi \sum_{i \in I} \distelem{p_i} \Theta_i' = \Theta'$.
	
	In other words, $\Delta$ and $\Theta$ express the same barbs and whenever
	$O[\Delta] \longsquiggly_\pi \Delta'$, there exists a transition $O[\Theta]
	\longsquiggly_\pi \Theta'$ such that $\Delta'\,Cv(\rel)\,\Theta'$ (the symmetrical
	argument is the same). So we can conclude that $(\Delta, \Theta) \in
	b(Cv(\rel))$.
\end{proof}

From~\cite{sangiorgi_enhancements_2011}, we know that if $Cv$ is compatible, then bisimulations up to $Cv$ are a sound proof technique.
Besides, it also allows us to show linearity of $\sim_{cs}$ as a corollary (note that the same proof holds also for $\sim_s$).

\linearity*
\begin{proof}
	We can easily prove that $Cv(\sim_{cs}) \subseteq\ \sim_{cs}$. Since $Cv$ is
	$b$-compatible, we have that $Cv(b(\sim_{cs})) \subseteq b(Cv(\sim_{cs}))$, and
	since $\sim_{cs}$ is the greatest fix point of $b$, we have $Cv(\sim_{cs}) \subseteq
	b(Cv(\sim_{cs}))$, meaning that $Cv(\sim_{cs})$ is a bisimulation, and so
	$Cv(\sim_{cs}) \subseteq\ \sim_{cs}$.
\end{proof}

In order to prove \autoref{thm:propertyA}, we first need a lemma on deterministic distributions.

\begin{lemma}\label{sum of deterministic is deterministic}
If $\Delta, \Theta$ are deterministic, then $\Delta \psum{p} \Theta$ is deterministic for each probability $p$.
\end{lemma}
\begin{proof}
We will prove that, for any probability $p$
\[
 \mathcal{A} = \left\{  \Delta \psum{p} \Theta \mid \Delta, \Theta \text { are deterministic}\right\}
\]
is a deterministic set.
If $(\Delta \psum{p} \Theta) \longsquiggly_\pi \Xi'$, and $(\Delta \psum{p} \Theta) \longsquiggly_\pi \Xi''$, since $\longsquiggly_\pi$ is decomposable it must be $\Xi' = \Delta' \psum{p} \Theta'$ and $\Xi'' = \Delta'' \psum{p} \Theta''$, with $\Delta \longsquiggly_\pi \Delta'$, $\Delta \longsquiggly_\pi \Delta''$, $\Theta \longsquiggly_\pi \Theta'$ and $\Theta \longsquiggly_\pi \Theta''$. But then, since $\Delta$ and $\Theta$ are deterministic, we have $\Delta' \sim_{cs} \Delta''$, $\Theta' \sim_{cs} \Theta''$ and they are all deterministic, so $\Xi', \Xi'' \in \mathcal{A}$ and $\Xi'\sim_{cs} \Xi''$, for linearity of $\sim_{cs}$.
\end{proof}

\propertyA*
\begin{proof}
For any deterministic $P$, we define
\[
	\rel' = \left\{ \left(\singleton{\conf{\rho \psum{p} \sigma, P, R}} \ , \ \singleton{\conf{\rho, P, R}} \psum{p} \singleton{\conf{\sigma, P, R}}\right) \mid \rho, \sigma, p, R \right\}
\]
  and prove that $\rel = \rel' \cup \{(\bot, \bot)\}$ is a bisimulation up to $Cv$ and up to bisimilarity. That is, we require that if $\Delta \longsquiggly_\pi \Delta'$ then $\Theta \longsquiggly_\pi \Theta'$ with $\Delta' \sim_{cs} Cv(\rel) \sim_{cs} \Theta'$. 
  Since $Cv$ is compatible, bisimulation up to $Cv$ and up to bisimilarity is a valid proof technique~\cite{sangiorgi_enhancements_2011}.

The case $\Delta = \Theta = \bot$ is straightforward. Otherwise, let $\Delta = \singleton{\conf{\nu, P, R}}, \Theta = \singleton{\conf{\rho, P, R}} \psum{p} \singleton{\conf{\sigma, P, R}}$ be in $\rel$, with $\nu = \rho \psum{p} \sigma$. Since they have the same process and observer, they are typed by the same context $\Sigma$. Notice that we do not need to quantify over any context $O[\blank]$, because  $\rel$ is a saturated relation, meaning that if $\Delta, \Theta \in \rel$, then $O[\Delta], O[\Theta] \in \rel$ for any context.

For the first condition, we have that $\Delta\downarrow_c^p$ if and only if $p = 1$ and $P$ or $R$ express the barb $c$. Then, also $\Theta$ expresses the barb $c$ with probability one, and vice versa.

For the second condition, suppose that $\singleton{\conf{\nu, P, R}} \longsquiggly_\pi \Delta'$. We will proceed by induction on $\longsquiggly_\pi$ to prove that $\Theta \longsquiggly_\pi \Theta'$ with $\Delta' Cv(\rel) \Theta'$. The ``classical'' cases, that do not modify the quantum state, are trivial, since if $\Delta \longsquiggly_\pi \Delta'$ then $\Theta$ can go in $\Theta'$ performing the same move in both configurations, and $\Delta' \rel \Theta'$. The only interesting cases are $\textsc{QOp}$ and $\textsc{QMeas}$ for the process, and  $\textsc{OQOp}$ and $\textsc{OQMeas}$ for the observer.

In the $\rulename{QOp}$ case, if 
\[\singleton{\conf{\nu, \sop{E}{\tilde{x}}.P', R}} \longsquiggly_\diamond \singleton{\conf{\sop[\tilde{x}]{E}{\nu}, P', R}}\]
then
\begin{gather*}
\singleton{\conf{\rho, \sop{E}{\tilde{x}}.P', R}} \psum{p} \singleton{\conf{\sigma, \sop{E}{\tilde{x}}.P', R}}\ \longsquiggly_{\diamond}\ \singleton{\conf{\sop[\tilde{x}]{E}{\rho}, P', R}} \psum{p} \singleton{\conf{\sop[\tilde{x}]{E}{\sigma}, P', R}}
\end{gather*}
and $\sop[\tilde{x}]{E}{\rho} \psum{p} \sop[\tilde{x}]{E}{\sigma} = \sop[\tilde{x}]{E}{\rho \psum{p} \sigma} = \sop[\tilde{x}]{E}{\nu}$, thanks to linearity of superoperators. The $\rulename{OQOp}$ case is identical.

In the $\rulename{QMeas}$ case, we have 
\[ \singleton{\conf{\nu, \meas{\tilde{x}}{y}.P', R}} \longsquiggly_\diamond  \Delta' \]
with
\[\Delta' = \sum_m \distelem{tr_m(\nu)}\singleton{\conf{\nu'_m, P'[\sfrac{m}{y}], R}} \qquad \nu'_m = \frac{\sop[m]{M}{\nu}}{tr_m(\nu)}
\]
where $\sop[m]{M}{\nu} = (M_m\otimes I) \nu (M_m\otimes I)^\dagger$ is the trace-nonincreasing superoperator corresponding to outcome $m$ when measuring qubits $\tilde{x}$, and $tr_m{\nu} = tr(\sop[m]{M}{\nu})$ is the probability of said outcome.
Then
\[
\singleton{\conf{\rho, \meas{\tilde{x}}{y}.P', R}} \psum{p} \singleton{\conf{\sigma, \meas{\tilde{x}}{y}.P', R}}
\longsquiggly_\diamond \Theta'_\rho \psum{p} \Theta'_\sigma
\]
with
\begin{align*}
\Theta'_\rho &= \sum_m \distelem{tr_m(\rho)} \singleton{\conf{\rho'_m, P'[\sfrac{m}{y}], R}} \qquad \rho'_m = \frac{\sop[m]{M}{\rho}}{tr_m(\rho)} 
\\
\Theta'_\sigma &= \sum_m \distelem{tr_m(\sigma)} \singleton{\conf{\sigma'_m, P'[\sfrac{m}{y}], R}} \qquad \sigma'_m = \frac{\sop[m]{M}{\sigma}}{tr_m(\sigma)}
\end{align*}
Observe that $tr_m(\nu) = tr(\sop[m]{M}{\rho \psum{p} \sigma})$ is equal to $tr_m(\rho) \psum{p} tr_m(\sigma)$, thanks to linearity of superoperators and trace. So, according to the rules of probability distributions, $\Theta'_\rho \psum{p} \Theta'_\sigma$ can be rewritten as
\[
\sum_m \distelem{tr_m(\nu)} (\singleton{\conf{\rho'_m, P'[\sfrac{m}{y}], R}} \psum{q} \singleton{\conf{\sigma'_m, P'[\sfrac{m}{y}], R}})
\]
with $q = \frac{p\cdot tr_m(\rho)}{tr_m(\rho) \psum{p} tr_m{\sigma}}$. It's easy to show that $\rho'_m \psum{q} \sigma'_m = \nu'_m$, from which it follows that 
\[
	\singleton{\conf{\nu'_m, P'[\sfrac{m}{y}], R}} \ \rel \ \left(\singleton{\conf{\rho'_m, P'[\sfrac{m}{y}], R}} \psum{q} \singleton{\conf{\sigma'_m, P'[\sfrac{m}{y}], R}}\right)
\]
and $\Delta' \ Cv(\rel)  \ \left(\Theta'_\rho \psum{p} \Theta'_\sigma\right)$. 

The $\rulename{OQMeas}$ case is identical, and all the other cases of $\rightarrow$ and $\longsquiggly_\pi$ are trivial, as they do not modify the quantum state.

For the third condition, we must prove that if $\Theta \longsquiggly_\pi \Theta'$, then $\Delta \longsquiggly_\pi \Delta'$ for some $\Delta' \sim_{cs} Cv(\rel) \sim_{cs} \Theta'$. Notice that $\conf{\rho, P, R} \longsquiggly_\pi \bot$ if and only if $\conf{\sigma, P, R} \longsquiggly_\pi \bot$, so we have that either $\Theta \longsquiggly_\pi \bot$ or $\Theta \longsquiggly_\pi \Theta'$ with $\bot \not\in \lceil\Theta' \rceil$. In the first case, also $\Delta \longsquiggly_\pi \bot$, and the condition is satisfied.
For the second case, we say $\Theta \longsquiggly_\pi$ if there exists a $\Theta'$ such that $\Theta \longsquiggly_\pi \Theta'$ and $\bot \not\in \lceil\Theta' \rceil$. We will show that, if $\Theta \longsquiggly_\pi$, there exist $\hat{\Theta}, \hat{\Delta}$ such that $\Theta \longsquiggly_\pi \hat{\Theta}$, $\Delta \longsquiggly_\pi \hat{\Delta}$ and $\hat{\Delta} Cv(\rel) \hat{\Theta}$. From this, since $\Theta$ is deterministic according to \autoref{sum of deterministic is deterministic}, it must be that for any $\Theta \longsquiggly_\pi \Theta'$ we have $\hat{\Delta} \sim_{cs} Cv(\rel) \sim_{cs} \Theta'$.

Assume $ \singleton{\conf{\rho, P, R}} \psum{p} \singleton{\conf{\sigma, P, R}} \longsquiggly_\pi$. To prove the existence of $\hat{\Theta}$ and $\hat{\Delta}$, in the case $\pi = \diamond$, we will proceed by induction on the syntax of $P$; in the case $\pi \neq \diamond$, we will proceed by induction on the syntax of $R$.  

Suppose that $\pi = \diamond$. If $P = \tau.P'$, the only possible transition is when $\hat{\Theta} = \singleton{\conf{\rho, P', R}} \psum{p} \singleton{\conf{\sigma, P', R}}$, and we have the transition $\singleton{\conf{\nu, \tau.P', R}} \longsquiggly_\diamond \singleton{\conf{\nu, P', R}} = 	\hat{\Delta}$, with $\hat{\Delta} \rel 	\hat{\Theta}$.

If $P = \sop{E}{\tilde{x}}.P'$ or $P = \meas{\tilde{x}}{y}.P$, there is only one possible $\hat{\Theta}$ and $\hat{\Delta}$, and the proof proceeds in the same way as seen for the previous direction of bisimulation.

If $P = P_1 + P_2$,  at least one between $P_1$ or $P_2$ is not in deadlock.
Suppose that $P_1$ is not in deadlock, then there exists a transition 
\[
\singleton{\conf{\rho, P_1, R}} \psum{p} \singleton{\conf{\sigma, P_1, R}} \longsquiggly_\diamond \hat{\Theta}
\]
For the inductive hypothesis, we have that $
\singleton{\conf{\nu, P_1, R}} \longsquiggly_\diamond \hat{\Delta} $
with $\hat{\Delta} \ Cv(\rel) \ \hat{\Theta}$.
If $P_1$ is in deadlock, then $P_2$ must not, and the proof is symmetrical.

If $P = P_1 \parallel P_2$, suppose that $P_1$ is not in deadlock, then there exists a transition 
\[
\singleton{\conf{\rho, P_1, R}} \psum{p} \singleton{\conf{\sigma, P_1, R}} \longsquiggly_\diamond \mathring{\Theta}
\]
and a transition $\Theta \longsquiggly_\diamond \hat{\Theta} = \mathring{\Theta} \parallel P_2$.
For the inductive hypothesis, we have that $\singleton{\conf{\nu, P_1, R}} \longsquiggly_\diamond \mathring{\Delta} $
with $\mathring{\Delta} \ Cv(\rel) \ \mathring{\Theta}$, from which it follows that 
\[\hat{\Delta} = \mathring{\Delta} \parallel P_2 \quad Cv(\rel) \quad \mathring{\Theta} \parallel P_2 = \hat{\Theta}\]

If $P_1$ is in deadlock and $P_2$ is not, then the proof is symmetrical. If both $P_1$ and $P_2$ are in deadlock, then there must be a synchronization, i.e. 
$P_1 \parallel P_2 \equiv (c!v + S) \parallel (c?x.P') + Q$. But then there exist the transitions $\Theta \longsquiggly_\diamond \hat{\Theta}$ with
\[
\hat{\Theta} = \singleton{\conf{\rho, P'[\sfrac{v}{x}], R}} \psum{p} \singleton{\conf{\sigma, P'[\sfrac{v}{x}], R}}
\]
and the transition $\Delta \longsquiggly_\diamond \hat{\Delta}$ with
\[
\hat{\Delta} = \singleton{\conf{\nu, P'[\sfrac{v}{x}], R}} 
\]

If $P = P \setminus c$  or $P = \ite{e}{P_1}{P_2}$, the required property follows trivially from the inductive hypothesis.

Suppose that $\pi \neq \diamond$, we proceed by induction on $R$. If $R = \sop{E}{\tilde{x}}.R'$ or $R = \meas{\tilde{x}}{y}.R$, the proof proceeds in the same way as before.

If $R$ is a sum of receptions, then $R \equiv c?x.R_1 + R_2$ and it must be $P
\equiv ((c!v + S) \parallel Q) \setminus D$; if $R = c!e$, then it must be $P
\equiv (((c?x . P') + P'') \parallel Q) \setminus D$. In both cases, there exists
two specific $\hat{\Theta}, \hat{\Delta}$ such that $\hat{\Theta} \ Cv(\rel) \
\hat{\Delta}$, as seen before for the synchronization  between processes.

If $R = R_1 \expar R_2$ and $\Theta \longsquiggly_{\ell \cdot \pi'}$, then the required property follows from the inductive hypothesis, as seen before for process parallelism. If $\Theta \longsquiggly_{r \cdot \pi'}$, the case is symmetrical.

If $R= \ite{e}{R_1}{R_2}$, the required property follows trivially from the inductive hypothesis.
\end{proof}

Consider now the refinement relation $\preceq$, defined as $P' \preceq P$ if $P'$ can be obtained from $P$ by substituting some occurrence of $Q + Q'$ with either $Q$, $Q'$ or $\ite{e}{Q}{Q'}$ for some $e$.
We extend the relation to configurations writing $\conf{\rho', P'} \preceq \conf{\rho, P}$ and $\conf{\rho', P', R'} \preceq \conf{\rho, P, R}$ when $\rho' = \rho$, $P' \preceq P$, and $R' = R$.
Moreover, we lift $\preceq$ to distributions by linearity and imposing $\bot \preceq \Delta$ for any $\Delta$.

\subsection{Proof of \autoref{thm:nondetVSite}}
To prove~\autoref{thm:nondetVSite} we start from some auxiliary lemmas.

\begin{lemma}\label{thm:nondetVSitelm0}
	Let $P' \preceq P$, if $P' \equiv Q'$ then $P \equiv Q$ for some $Q$ such that $Q' \preceq Q$.
\end{lemma}
\begin{proof}
	Trivial by cases on the rules for $\equiv$.
\end{proof}

\begin{lemma}\label{thm:nondetVSitelm1}
	Let $P' \preceq P$, if $\conf{\rho, P'} \rightarrow \Delta'$ then $\conf{\rho, P} \rightarrow \Delta$ for some $\Delta$ such that $\Delta' \preceq \Delta$.
\end{lemma}
\begin{proof}
	Consider the case $\Delta' \neq \bot$.
	We proceed by induction on the rules for $\conf{\rho, P'} \rightarrow \Delta'$.

	If the rule is \rulename{Tau}, then $P' = \tau.Q_1' + Q_2'$, with $\Delta' = \singleton{\conf{\rho, Q_1'}}$ and can only be the refinement of $P = \tau.Q_1 + Q_2$ for $Q_i' \preceq Q_i$.
	Then $\conf{\rho, P} \rightarrow \singleton{\conf{\rho, Q_1}}$, and the refinement holds by definition.

	The same applies to \rulename{Restrict}, \rulename{QOP} and \rulename{QMeas} by noticing that $P'\setminus c \preceq P$ only if $P = P'' \setminus c$, and that $P' \preceq P$ implies $P'[v/x] \preceq P[v/x]$ for any $v$ and $x$.
	
	Consider \rulename{Par}, then $P' = Q_1' \parallel Q_2'$, with $\Delta' = \Theta' \parallel Q_2'$ for some $\Theta'$ such that $\conf{\rho, Q_1'} \rightarrow \Theta'$.
	Notice that $P$ must be equal $Q_1 \parallel Q_2$ for some $Q_1$ and $Q_2$ such that $Q_i' \preceq Q_i$.
	By induction hypothesis, $\conf{\rho, Q_1} \rightarrow \Theta$ for some $\Theta$ such that $\Theta' \preceq \Theta$, and thus, by \rulename{Par} rule, $\conf{\rho, Q_1 \parallel Q_2} \rightarrow \Theta \parallel Q_2$, and $\Theta' \parallel Q_2' \preceq \Theta \parallel Q_2$ by definition.
	
	Consider \rulename{Congr} and assume $P' \equiv Q'$, $\conf{\rho, Q'} \rightarrow \Theta'$ and $\Theta' \equiv \Delta'$.
	By~\autoref{thm:nondetVSitelm0}, there exists some $Q$ such that $P \equiv Q$ and $Q' \preceq Q$.
	Then, by induction hypothesis, $\conf{\rho, Q} \rightarrow \Theta$ with $\Theta' \preceq \Theta$.
	From $\Theta' \equiv \Delta'$ and $\Theta' \preceq \Theta$, we know by~\autoref{thm:nondetVSitelm0} that $\Theta \equiv \Delta$ for some $\Delta$ such that $\Delta' \preceq \Delta$.
	Then $\conf{\rho, P} \rightarrow \Delta$ by applying \rulename{Congr}.
	
	Finally, consider the case in which $\Delta' = \bot$.
	Then, $\conf{\rho, P} \rightarrow \Delta$ either for $\Delta = \bot$ or for some $\Delta \neq \bot$, and in both cases $\bot \preceq \Delta$.
\end{proof}

\begin{lemma}\label{thm:nondetVSitelm2}
	Let $P' \preceq P$, if $\conf{\rho, P', R} \longsquiggly_\pi \Delta'$ then $\conf{\rho, P, R} \longsquiggly_\pi \Delta$ for some $\Delta$ such that $\Delta' \preceq \Delta$.
\end{lemma}
\begin{proof}
	Consider the case in which $\Delta' \neq \bot$.
	We proceed by induction on the rules for $\conf{\rho, P', R} \longsquiggly_\pi \Delta'$.
	
	Rules \rulename{OQOp}, \rulename{OQMeas}, \rulename{ParL} and \rulename{ParL} are trivial, since the rule is applicable for any process $P$.
	
	Consider \rulename{Input}, then $P' = (c!v + P_1') \parallel P_2' \setminus D$ and $R = c?x.R_1 + R_2$ with $c \notin D$.
	Moreover, $\pi = \epsilon$ and $\Delta' = \singleton{\conf{\rho, P_2' \setminus D, R_1[v/x]}}$.
	Notice that $P'$ can only be the refinement of $P = (c!v + P_1) \parallel P_2 \setminus D$ for $P_i' \preceq P_i$.
	Then the result holds by applying \rulename{Input} for obtaining $\conf{\rho, P, R} \longsquiggly_\pi \singleton{\conf{\rho, P_2 \setminus D, R_1[v/x]}}$.
	
	The proof for \rulename{Output} is the same, and follows from the fact that $Q' \preceq Q$ implies $Q'[v/x] \preceq Q[v/x]$ for any $v$ and $x$.
	
	Consider \rulename{process} and assume $\conf{\rho, P'} \rightarrow \Delta'$.
	Then, by~\autoref{thm:nondetVSitelm1}, $\conf{\rho, P} \rightarrow \Delta$ with $\Delta' \preceq \Delta$.
	We can thus apply \rulename{process} to derive $O[\conf{\rho, P}] \longsquiggly_\diamond O[\Delta]$, and $O[\Delta'] \preceq O[\Delta]$ holds by definition.
	
	Consider \rulename{Congr} and assume $P' \equiv Q'$, $R \equiv_O S$, $\conf{\rho, Q', S} \longsquiggly_\pi \Theta'$ and $\Theta' \equiv_O \Delta'$.
	From~\autoref{thm:nondetVSitelm0},  $P \equiv Q$ such that $Q' \preceq Q$, and thus, by induction hypothesis, $\conf{\rho, Q, S} \longsquiggly_\pi \Theta$ with $\Theta' \preceq \Theta$.
	Hence, from~\autoref{thm:nondetVSitelm0} and by definition of $\preceq$, a distribution $\Delta$ exists such that $\Theta \equiv \Delta$ and $\Delta \preceq \Delta'$.
	We can thus apply \rulename{Congr} to derive $\conf{\rho, P, R} \longsquiggly_\pi \Delta$.
	
	Finally, consider $\Delta' = \bot$.
	Then, $\conf{\rho, P, R} \longsquiggly_\pi \Delta$ either for $\Delta = \bot$ or for some $\Delta \neq \bot$, and in both cases $\bot \preceq \Delta$.
\end{proof}

\nondetVSite*
\begin{proof}
	By~\autoref{thm:nondetVSitelm2} and decomposability of $\longsquiggly_{\pi}$.
\end{proof}

\section{Proofs of Section~\ref{propertiesOfBisimilarity}}\label{sec:propertiesappendix}
\thmchinese*
\begin{proof}
 To prove the first point, suppose $\conf{\rho, P, \nil} \sim_{cs} \conf{\sigma, Q, \nil}$, with $\rho, \sigma \in \hilbert_{\tilde{q},\tilde{p}}$ and $\tilde{p} \vdash P$, $\tilde{p} \vdash Q$.  For any superoperator $\mathcal{TS}(\hilbert_{\tilde{q}})$ 
 we can construct a context 
 \[
 O[\blank] = [\blank] \parallel \sop{E}{{\tilde{q}}}.(a!0 \parallel c!\tilde{q})
 \] where $a$ is a fresh channel. We know that $O[\conf{\rho, P, \nil}]$ and $O[\conf{\sigma, Q, \nil}]$ are bisimilar, and $O[\conf{\rho, P, \nil}]$ can evolve in $\conf{\sop[\tilde{q}]{E}{\rho}, P, a!0 \parallel c!\tilde{q} \parallel \nil}$.
  Then $O[\conf{\sigma, Q, \nil}]$ must evolve in $\conf{\sop[\tilde{q}]{E}{\sigma}, Q, a!0 \parallel c!\tilde{q} \parallel \nil}$, because it must match the $\downarrow_a$ barb, and $a$ is fresh. So we have that  
\[\conf{\sop[\tilde{q}]{E}{\rho}, P, a!0 \parallel c!\tilde{q} \parallel \nil} \sim_{cs} 
  \conf{\sop[\tilde{q}]{E}{\sigma}, Q, a!0 \parallel c!\tilde{q} \parallel \nil}
\]
and from this it follows 
\[ \conf{\sop[\tilde{q}]{E}{\rho}, P, \nil} \sim_{cs} 
  \conf{\sop[\tilde{q}]{E}{\sigma}, Q, \nil}
\] simply by contradiction: if there was a context capable of distinguishing $\conf{\sop[\tilde{q}]{E}{\rho}, P, \nil}$ from $\conf{\sop[\tilde{q}]{E}{\rho}, P, \nil}$ then there would be a context able to distinguish also $\conf{\sop[\tilde{q}]{E}{\rho}, P, a!0 \parallel c!\tilde{q} \parallel \nil}$ from $\conf{\sop[\tilde{q}]{E}{\sigma}, Q, a!0 \parallel c!\tilde{q} \parallel \nil}$.

To prove the second point we proceed by contradiction, supposing  $\conf{\rho, P, \nil} \sim_{cs} \conf{\sigma, Q, \nil}$  and $tr_{\tilde{p}}(\rho) \neq tr_{\tilde{p}}(\sigma)$, with $\rho, \sigma \in \hilbert_{\tilde{q}, \tilde{p}}$ and $\tilde{p} \vdash P$, $\tilde{p} \vdash Q$.
 If $\ptrace{p}{\rho} \neq \ptrace{p}{\sigma}$, then there exists a measurement $M_{\tilde{q}} = \{M_1, \ldots ,M_m\}$ that distinguishes them, i.e.\ such that $p_m(\ptrace{p}{\rho}) = tr(M_m \ptrace{p}{\rho} M_m^\dagger) \neq tr(M_m \ptrace{p}{\sigma} M_m^\dagger) = p_m(\ptrace{p}{\sigma})$ for some $m$. But for \autoref{ptrace drops sop}, the same probabilities arise also from the measurement $M_{\tilde{q}\tilde{p}} = \{M_1 \otimes I_{\tilde{p}}, \ldots , M_m \otimes I_{\tilde{p}}\}$, that  can therefore distinguish the whole state $\rho$ from $\sigma$. So, taken the context 
\begin{align*}
O[\blank] =& [\blank] \parallel \meas{\tilde{q}}{x}. \\
&\nil_{\tilde{q}} \parallel \ite{x = 1}{c_1!0}{\ldots} \\
&\qquad\quad\ite{x = m-1}{c_{m-1}!0}{c_m!0}
\end{align*} 
where $c_1 \ldots c_m$ are fresh channels, we have that $O[\conf{\rho, P, \nil}]$ should be bisimilar to $O[\conf{\sigma, Q, \nil}]$. 
But 
\[O[\conf{\rho, P, \nil}] \longsquiggly_r \sum_m \distelem{p_m(\rho)} \conf{\rho_m, P, \nil \parallel c_m!0}\]
and $O[\conf{\sigma, Q, \nil}]$ can perform only the transition \[O[\conf{\sigma, Q, \nil}] \longsquiggly_r \sum_m \distelem{p_m(\sigma)} \conf{\sigma_m, Q, \nil \parallel c_m!0}\] to match the barbs, but we know that $p_m(\rho) \neq p_m(\sigma)$ for at least one $m$.
\end{proof}

In order to prove \autoref{thm:discarded}, we need an additional lemma.

\begin{lemma}\label{ptrace preserves arrow}
Assume that $\Delta$ is a distribution such that $\ptrace{q}{\Delta}$ is well-defined. Then, for $O[\blank]$ we have
\[ O[\Delta] \longsquiggly_\pi \Delta' \text{ if and only if } O[\ptrace{q}{\Delta}] \longsquiggly_\pi \ptrace{q}{\Delta'}
\]
\end{lemma}
\begin{proof}
We will prove that, for any $\iconf \in \confbot$
\[
\singleton{\mathcal{C}}\longsquiggly_\pi \Delta' \text{ if and only if } \ptrace{q}{\singleton{\mathcal{C}}} \longsquiggly_\pi \ptrace{q}{\Delta'}
\]
from which the desired lemma follows easily  by linearity.

We will proceed by induction on $\singleton{\mathcal{C}}\longsquiggly_\pi \Delta'$. First, if $\iconf = \bot$, then the property follows from the definitions of $O[\bot]$ and $\ptrace{q}{\bot}$. Otherwise, notice that $disc(\tilde{q})$ is a deadlock process that cannot take part in any synchronization. The only interesting base cases are $\rulename{QOp}$, $\rulename{QMeas}$, $\rulename{OQOp}$ and $\rulename{OQMeas}$, when the process (or the observer) modifies the quantum state. We will deal only with the process, as the proof for observer is the same. All the other rules of $\rightarrow$ and $\longsquiggly$ are trivial, as they do not modify the quantum state. 

Assume that \[\singleton{\conf{\rho, (\sop{E}{\tilde{x}}.P + Q \parallel \nil_{\tilde{q}}), R}} \rightarrow \singleton{\conf{\sop[\tilde{x}]{E}{\rho}, (P \parallel \nil_{\tilde{q}}), R}}\]
with $\sop[\tilde{x}]{E}{\rho} = (\mathcal{E}\otimes \mathcal{I}_{\tilde{y}} \otimes \mathcal{I}_{\tilde{q}})(\rho)$, where $\tilde{y}$ are all the qubits not in $\tilde{x}$ nor in $\tilde{q}$, and $ \mathcal{I}_{\tilde{y}}$ is the identity superoperator on said qubits.
Then we also have \[\singleton{\conf{\ptrace{q}{\rho},(\sop{E}{\tilde{x}}.P + Q), R}} \rightarrow \singleton{\conf{\sop[\tilde{x}]{E}{\ptrace{q}{\rho}}, P, R}}\]
with $\sop[\tilde{x}]{E}{\ptrace{q}{\rho}} = (\mathcal{E}\otimes \mathcal{I}_{\tilde{y}}) (\ptrace{q}{\rho})$.
  From \autoref{ptrace drops sop}, we have 
  \[(\mathcal{E}\otimes \mathcal{I}_{\tilde{y}}) (\ptrace{q}{\rho}) = \ptrace{q}{(\mathcal{E}\otimes \mathcal{I}_{\tilde{y}} \otimes \mathcal{I}_{\tilde{q}})(\rho)}\]
  and we conclude  
	\[\singleton{\conf{\sop[\tilde{x}]{E}{\ptrace{q}{\rho}}, P, R}} = \ptrace{q}{\singleton{\conf{\sop[\tilde{x}]{E}{\rho}, (P \parallel \nil_{\tilde{q}}), R}}}\]
 The other direction is similar.
	
	The $\rulename{QMeas}$ case is also similar: \autoref{ptrace drops sop} implies $tr\left(\ptrace{q}{\sop[\tilde{x}, m]{M}{\rho}}\right) = tr\left(\sop[\tilde{x}, m]{M}{\ptrace{q}{\rho}}\right)$, so $\Delta'$ and $\ptrace{q}{\Delta'}$ have the same probabilities.
\end{proof}

We can now prove that distribution bisimilarity is closed for additional discarded qubits.

\bisimilaritydiscarded*
\begin{proof}
We need to show that 
\[\rel = \left\{(\Delta, \Theta) \mid  \ptrace{q}{\Delta} \sim_{cs} \ptrace{q}{\Theta}\right\}
\]
is a distribution bisimulation. Notice that if $\Sigma \vdash \Delta, \Theta$, then $\Sigma \setminus \tilde{q} \vdash \ptrace{q}{\Delta}, \ptrace{q}{\Theta}$.

For the first condition, observe that $\Delta \downarrow^p_b$ if and only if $\ptrace{q}{\Delta} \downarrow^p_b$, since the partial trace operation is linear and the process $\nil_{\tilde{q}}$ does not express any barb.

For the second condition, it is sufficient to notice that $\ptrace{q}{O[\Delta]} = O[\ptrace{q}{\Delta}]$ for any context $O[\blank]$. Then, we have that if $O[\Delta] \longsquiggly_\pi \Delta'$, then from \autoref{ptrace preserves arrow} we know that $\ptrace{q}{O[\Delta]} = O[\ptrace{q}{\Delta}] \longsquiggly_\pi \ptrace{q}{\Delta'}$, and then $O[\ptrace{q}{\Theta}] \longsquiggly_\pi \Theta'' \sim_{cs} \ptrace{q}{\Delta'}$ and $O[\Theta] \longsquiggly_\pi \Theta'$ such that $\Theta'' = \ptrace{q}{\Theta'}$. The third condition is symmetric.
\end{proof}

\uptocv*
\begin{proof}
Soundness follows from compatibility, proven in~\autoref{thm:cvcomp}, as shown in~\cite{sangiorgi_enhancements_2011}.
\end{proof}

\end{document}